\documentclass[12pt]{article}

\usepackage[usenames,dvipsnames,svgnames,table]{xcolor}

\usepackage{hyperref}

\usepackage[xindy,toc,sanitize={name=false,description=false}]{glossaries}
\renewcommand*{\CustomAcronymFields}{%
  name={\the\glsshorttok},
  description={\the\glslongtok},
  first={\noexpand\emph{\the\glslongtok}\space(\the\glsshorttok)},%
  firstplural={\noexpand\emph{\the\glslongtok\noexpand\acrpluralsuffix}\space(\the\glsshorttok)},%
  text={\the\glsshorttok},%
  plural={\the\glsshorttok\noexpand\acrpluralsuffix}%
}

\SetCustomStyle

\usepackage{color}
\usepackage{balance}
\usepackage{latexsym}
\usepackage{graphicx}
\usepackage{epsfig}
\usepackage[hang]{subfigure}
\usepackage{algorithm}
\usepackage[noend]{algorithmic}
\usepackage{epic}
\usepackage{eepic}
\usepackage{amsfonts}
\usepackage{amsmath, amsthm, amssymb}
\usepackage{stmaryrd}
\usepackage[all]{xy}
\usepackage{booktabs}
\usepackage{multirow}
\usepackage{rotating}
\usepackage{multicol}

\newtheorem{theorem}{Theorem}[section]
\newtheorem{corollary}[theorem]{Corollary}

\newtheorem{proposition}[theorem]{Proposition}
\newtheorem{fact}[theorem]{Fact}

\theoremstyle{remark}
\newtheorem{remark}[theorem]{Remark}

\theoremstyle{example}
\newtheorem{example}[theorem]{Example}

\theoremstyle{definition}
\newtheorem{definition}[theorem]{Definition}

\algsetup{linenodelimiter=.} 
\algsetup{indent=1em}

\definecolor{light-gray}{gray}{0.90}

\newcounter{theorem-backup}

\newcommand{\N}{\mathbb{N}}
\newcommand{\Z}{\mathbb{Z}}
\newcommand{\R}{\mathbb{R}}

\newcommand{\B}{\mathbb{B}}

\usepackage{listings}
\lstdefinelanguage{PseudoC}[ISO]{C++} { 
morekeywords={foreach, and, not, or, is, FIFO_Queue, HashTable, FILE, Cache}, 
morekeywords=[2]{HashInsert, Enqueue, Dequeue, next}, 
mathescape=true,
alsoletter={'},
deletestring=[b]'
}

\lstdefinelanguage{PseudoCWithNumbers}[ISO]{C++} { 
morekeywords={foreach, and, not, or, is, FIFO_Queue, HashTable, FILE, Cache}, 
morekeywords=[2]{HashInsert, Enqueue, Dequeue, next}, 
mathescape=true,
alsoletter={'},
deletestring=[b]'
}

\lstdefinelanguage{Murphi}[]{Pascal} { 
morekeywords={ruleset, rule, invariant, startstate, return, endif, endfor, endswitch, forall, endforall, exists, endexists}, 
mathescape=true, 
morestring=[b]",
morestring=[b]', 
morecomment=[s]{/*}{*/} ,
morecomment=[l]{--}
}

\lstdefinelanguage{PRISM}[ISO]{C++} { 
morekeywords={probabilistic, stochastic, const, rate, module, endmodule, init, P, U},
mathescape=true, 
alsoletter={'},
deletestring=[b]'
}

\lstdefinelanguage{Yacc}[ISO]{C++} { 
morekeywords={token, left}, 
mathescape=false,
alsoletter={'},
deletestring=[b]'
}

\lstdefinestyle{PseudoC}{
language=PseudoC,
basicstyle=\ttfamily,
tabsize=1,
showlines=false,
emptylines=*1,
breaklines=true,
breakindent=5pt,
keywordstyle=\rmfamily\bfseries,
keywordstyle=[2]\rmfamily,
commentstyle=\itshape, 
columns=fixed,
showspaces=false, 
showstringspaces=false, 
showtabs=false, 
escapechar=\%
}

\lstdefinestyle{PseudoCWithNumbers}{
language=PseudoC,
basicstyle=\ttfamily,
tabsize=1,
showlines=false,
emptylines=*1,
breaklines=true,
breakindent=5pt,
keywordstyle=\rmfamily\bfseries,
keywordstyle=[2]\rmfamily,
commentstyle=\itshape, 
columns=fixed,
showspaces=false, 
showstringspaces=false, 
showtabs=false, 
escapechar=\%,
numbersep=5pt,framexleftmargin=15pt,numbers=left,
}

\lstdefinestyle{Murphi}{
language=Murphi,
basicstyle=\ttfamily,
tabsize=1,
showlines=false,
emptylines=*1,
breaklines=true,
breakindent=5pt,
keywordstyle=\rmfamily\bfseries,
commentstyle=\itshape, 
columns=fixed,
showspaces=false, 
showstringspaces=false, 
showtabs=false,
escapechar=\%}

\lstdefinestyle{PRISM}{
language=PRISM,
basicstyle=\ttfamily,
tabsize=1,
showlines=false,
emptylines=*1,
breaklines=true,
breakindent=5pt,
keywordstyle=\rmfamily\bfseries,
commentstyle=\itshape, 
columns=fixed,
showspaces=false, 
showstringspaces=false, 
showtabs=false,
escapechar=\%}

\lstdefinestyle{Yacc}{
language=PseudoC,
basicstyle=\ttfamily,
tabsize=1,
showlines=false,
emptylines=*1,
breaklines=true,
breakindent=5pt,
keywordstyle=\rmfamily\bfseries,
keywordstyle=[2]\rmfamily,
commentstyle=\itshape, 
columns=fixed,
showspaces=false, 
showstringspaces=false, 
showtabs=false, 
}

\newcommand{\qks}{\mbox{QKS}}

\newcommand{\fun}[1]{{\textsl{#1}}}
\newcommand{\ctrset}{{\sf C}({\cal H}, {\cal Q})}

\newcommand{\rimandotekrep}[1]{\cite{cav2010-tekrep-art-2011} (App.~\ref{#1})}

\makeglossaries

\begin{document}


%
%




\title{Model Based Synthesis of Control Software from System Level Formal Specifications}

%
%
%

\author{Federico Mari, Igor Melatti, Ivano Salvo, Enrico Tronci\\
\small \itshape Department of Computer Science\\
\small \itshape Sapienza University of Rome\\
\small \itshape via Salaria 113, 00198 Rome\\
\small email: \{mari,melatti,salvo,tronci\}@di.uniroma1.it\\
\small {\em Accepted for publication by ACM Transactions on Software Engineering and Methodology (TOSEM)}}

\maketitle

\begin{abstract}

Many \emph{Embedded Systems} are indeed \emph{Software Based Control Systems},
that is control systems whose controller consists of \emph{control software} running on a 
microcontroller device. 
This motivates investigation on \emph{Formal Model Based Design} 
approaches for automatic synthesis of embedded systems control software. 
We present an algorithm, along with a tool \qks \ implementing it, that 
from a formal model (as a \emph{Discrete Time Linear Hybrid System})
of the controlled system (\emph{plant}), 
\emph{implementation specifications} 
(that is, number of bits in the \emph{Analog-to-Digital}, AD, conversion) and
\emph{System Level Formal Specifications}
(that is, safety and liveness requirements for the \emph{closed loop system})
returns correct-by-construction control software that has a
\emph{Worst Case Execution Time} (WCET) linear in the number of AD bits
and meets the given specifications. 
We show feasibility
of our approach by presenting experimental results on using it 
to synthesize control software for a buck DC-DC converter, 
a widely used mixed-mode analog circuit, and for the inverted pendulum.
\end{abstract}

\section{Introduction}
\label{intro.tex}

\newacronym{LTS}{LTS}{Labeled Transition System}
\newacronym{QFC}{QFC}{Quantized Feedback Control}%




Many \emph{Embedded Systems} are indeed
\newacronym{SBCS}{SBCS}{Software Based Control System}%
\glspl{SBCS}.
An \gls{SBCS}
consists of two main subsystems:
the \emph{controller} and the \emph{plant}.
Typically, the plant is a physical system
consisting, for example, of mechanical or electrical devices
whereas the controller 
consists of \emph{control software} running on a microcontroller (see Fig.~\ref{fig:closed_loop}).
In an endless loop, 
the controller
reads \emph{sensor} outputs from the plant and
sends commands to plant \emph{actuators}
in order to guarantee that the 
\emph{closed loop system}
(that is, the system consisting of both
plant and controller) meets given
\emph{safety} and \emph{liveness} specifications
(\emph{System Level Formal Specifications}).
Missing such goals can cause failures or damages to the plant, thus making an \gls{SBCS} a hard real-time system. 

Software generation from models
and formal specifications forms the core of
\emph{Model Based Design} of embedded software 
\cite{Henzinger-Sifakis-fm06}.
This approach is particularly interesting for \glspl{SBCS}
since in such a case system level (formal) specifications are much easier
to define than the control software behavior itself.


%

Fig. \ref{control-loop-figure.tex} shows the 
typical control loop skeleton for an \gls{SBCS}.
Measures from plant \emph{sensors} go through an
\newacronym{AD}{AD}{Analog-to-Digital}%
\gls{AD} conversion (\emph{quantization})
before being processed (line \ref{quantized-feedback})
and commands from the control
software go through a 
\newacronym{DA}{DA}{Digital-to-Analog}%
\gls{DA} conversion
before being sent to plant \emph{actuators}
(line \ref{DA-conversion}).
Basically, the control software design problem for \glspl{SBCS} consists
in designing software implementing functions
\texttt{Control\_Law} and
\texttt{Controllable\_Region}
computing, respectively,
the command to be sent to the plant
(line \ref{control-law})
and the set of states on which the \texttt{Control\_Law}
function works correctly
(\emph{Fault Detection} in line \ref{fault-detection}).
Fig.~\ref{fig:closed_loop} summarizes the complete closed loop system forming an \gls{SBCS}.


\begin{figure}[!t]
  \centering
  \begin{minipage}{.62\hsize}
    \begin{center}
      \begin{algorithmic}[1]
        \STATE {\bf Every} $T$ seconds (\emph{sampling time)} {\bf do} \label{sampling-and-hold}
        \STATE \hspace*{1ex} {\bf Read} \gls{AD} conversion $\hat{x}$ of plant sensor outputs $x$ \label{quantized-feedback}
        \STATE \hspace*{1ex} {\bf If} ($\hat{x}$ is not in the \texttt{Controllable\_Region}) \label{fault-detection}
        \STATE \hspace*{2ex} {\bf Then} {\em // Exception (Fault Detected):} 
        \STATE \hspace*{3ex}      Start Fault Isolation and Recovery (\texttt{FDIR})
        \STATE \hspace*{2ex} {\bf Else} {\em // Nominal case:}
        \STATE \hspace*{3ex}  Compute (\texttt{Control\_Law}) command $\hat{u}$ from $\hat{x}$ \label{control-law}
        \STATE \hspace*{3ex} {\bf Send} \gls{DA} conversion $u$ of $\hat{u}$ to plant actuators \label{DA-conversion}
      \end{algorithmic}
      \caption{A typical control loop skeleton.}
      \label{control-loop-figure.tex}
    \end{center}
  \end{minipage}
  \begin{minipage}{.37\hsize}
    \begin{center}
      \scalebox{0.35}{\input{closed_loop.pstex_t}} 
      \caption{Software Based Control System.}
      \label{fig:closed_loop}
    \end{center}
  \end{minipage}
\end{figure}

\subsection{The Separation-of-Concerns Approach}
\label{separation-of-concerns}

For \gls{SBCS} system level specifications are typically
given with respect to the desired behavior 
of the closed loop system. The 
\emph{control software} 
(that is, \texttt{Control\_Law} and \texttt{Controllable\_Region})
is designed using a 
\emph{separation-of-concerns} approach. That is, 
\emph{Control Engineering} techniques 
(e.g., see \cite{modern-control-theory-1990})
are used
to design, from the closed loop system level specifications,
\emph{functional specifications} (\emph{control law})
for the \emph{control software}
whereas
\emph{Software Engineering} techniques are used to
design control software implementing the given
functional specifications.

Such a separation-of-concerns approach has 
several
drawbacks. 

First, usually control engineering techniques do not
yield a formally verified specification for the
control law or controllable region when quantization is taken into account. 
This is particularly the case when the plant has to be modelled 
as a \emph{Hybrid System}
\cite{alur-sfm04,alg-hs-tcs95,HHT96,AHH96}
(that is a system with continuous as well as discrete state changes).
As a result, even if the control software meets its functional
specifications there is no formal guarantee that
system level specifications are met since
quantization effects are not formally accounted for.

\sloppy

Second, issues concerning computational resources, 
such as control software
\newacronym{WCET}{WCET}{Worst Case Execution Time}%
\gls{WCET}, %
can only be considered very late in the \gls{SBCS} design activity,
namely once the software has been designed. As a result, since the
\gls{SBCS} is a hard real-time system (Fig.~\ref{fig:closed_loop}),
the control software may have a \gls{WCET} greater than the sampling
time (line \ref{sampling-and-hold} in
Fig. \ref{control-loop-figure.tex}). This invalidates the
schedulability analysis (typically carried out before the control
software is completed) and may trigger redesign of the software or
even of its functional specifications (in order to simplify its
design).

\fussy


Last, but not least, the classical separation-of-concerns approach does not effectively support 
design space exploration for the control software. 
In fact, although in general there will be many functional specifications
for the control software that will allow meeting the given
system level specifications, the software engineer only
gets one 
 to play with. This overconstrains a priori the design space
for the control software implementation preventing, for example, effective performance trading 
(e.g., 
between number of bits in \gls{AD} conversion,  
\gls{WCET}, 
RAM usage, 
CPU power consumption, 
etc.).

We note that the above considerations also apply
to the typical situation where
Control Engineering techniques are used to design a control
law and then tools like Berkeley's Ptolemy \cite{ptolemy03},
Esterel's SCADE \cite{scade12}
or MathWorks Simulink \cite{simulink12}
are used to generate the control software. 
Even when the control law is automatically generated and proved correct
(for example, as in \cite{pessoa-cav10}) such an approach
does not yield any formal guarantee about the software correctness
since quantization of the state measurements
is not taken into account in the computation of the control
law. Thus such an approach cannot answer questions like:
1) Can 8 bit \gls{AD} be used or instead we need, say, 12 bit \gls{AD}?
2) Will the control software code run \emph{fast enough} 
   on a, say, 1 MIPS microcontroller 
(that is, is the control software \gls{WCET} less than the sampling time)?
3) What is the controllable region?

The previous considerations motivate research on 
\emph{Software Engineering} 
methods and tools focusing on control software synthesis
(rather than on control law synthesis as in \emph{Control Engineering}).
The objective is that from the plant model (as a hybrid system), 
from formal specifications for the closed loop system behavior 
(\emph{System Level Formal Specifications}) 
and from \emph{Implementation Specifications} 
(that is, number of bits used in the quantization process)
such methods and tools can generate correct-by-construction control software
satisfying the given specifications.
This is the focus of the present paper.

For a more in-depth discussion of the literature related to the present
paper, we refer the reader to Sect.~\ref{related-works.tex} and 
Tab.~\ref{tab:related-works.tab.tex}.



\subsection{Our Main Contributions}
We model the controlled system (plant) as a
\newacronym{DTLHS}{DTLHS}{Discrete Time Linear Hybrid System}%
\gls{DTLHS} (see Sect.~\ref{dths.tex}), that is a discrete time hybrid system
whose dynamics is defined as a \emph{linear predicate}
(i.e., a boolean combination of linear constraints, see Sect.~\ref{basic.tex}) on its
variables. 
%
We model system level safety as well as liveness specifications as sets of states
defined, in turn, as linear predicates. 
In our setting, as always in control problems, liveness constraints define the
set of states that any evolution of the closed loop system should eventually reach
(\emph{goal states}). 
Using an approach similar to the one in
 \cite{dt-ctr-rect-aut-icalp97,decidability-hybrid-automata-jcss98,AgrawalSTY06}, in \cite{ictac2012} we prove that
both existence of a controller for a \gls{DTLHS} and existence of a 
\emph{quantized} controller for a \gls{DTLHS} are undecidable problems.
Accordingly, we can only hope for semi- or incomplete algorithms.

We present an algorithm computing a sufficient condition and a necessary condition for existence
of a solution to our control software synthesis problem (see Sects.~\ref{lts-ctr.tex} and~\ref{ctr-abs.tex}).
Given a \gls{DTLHS} model ${\cal H}$ for the plant,
a quantization schema (i.e. how many bits we use for \gls{AD} conversion) and system level formal specifications,
our algorithm (see Sect.~\ref{ctr-syn-algorithm.tex}) will return 1 if they are able to decide if a solution exists or not, and
0 otherwise (unavoidable case since our problem is undecidable).
Furthermore, when our sufficient condition is satisfied, 
we return a pair of C functions (see Sect.~\ref{sec:controlSoftware})
\texttt{Control\_Law}, \texttt{Controllable\_Region}
such that:
function \texttt{Control\_Law}
implements
a 
\emph{Quantized Feedback Controller} (QFC) for ${\cal H}$
meeting the given system level formal specifications
and function \texttt{Controllable\_Region}
computes the set of states on which
\texttt{Control\_Law} is guaranteed to work correctly
(\emph{controllable region}).
While {\em WCET analysis} is actually performed after control software
generation, our contribution is to supply both functions with a
\glsfirst{WCET} guaranteed to be linear in the number of bits of the
state quantization schema (see Sect.~\ref{wcet.tex}).
Furthermore, function
\texttt{Control\_Law} is \emph{robust}, that is, it meets the given
closed loop requirements
notwithstanding (nondeterministic) \emph{disturbances} such as
variations in the plant parameters.

\sloppy

We implemented our algorithm on top of the 
CUDD 
package and of
the GLPK 
\newacronym{MILP}{MILP}{Mixed Integer Linear Programming}%
\gls{MILP} solver, 
thus obtaining tool
\newacronym{QKS}{QKS}{Quantized feedback Kontrol Synthesizer}%
\gls{QKS} (publicly available at \cite{QKS}).
This allows us to
present
experimental results on using
\gls{QKS}
to synthesize robust control software for a widely
used mixed-mode analog circuit: the buck DC-DC converter 
(e.g. see \cite{fuzzy-dc-dc-1996}). 
This is an interesting and challenging example
(e.g., see \cite{reliability-power-systems-2008}, \cite{time-optimal-dc-dc-2008})
for automatic synthesis
of correct-by-construction control software from
system level formal specifications. 
Moreover, in order to show effectiveness of our approach,
we also present experimental results on using \gls{QKS} for the inverted pendulum \cite{KB94}.

\fussy

Our experimental results address both computational feasibility and closed loop performances.
As for computational feasibility, we show that within about
40 hours
of CPU time
and within 100MB of RAM we can synthesize
control software
for
a 10-bit
quantized buck DC-DC converter.
As for closed loop performances, our synthesized control software set-up time (i.e., the time needed to reach the steady state) and ripple (i.e., the wideness of the oscillations around the steady state once this has been reached) compares well with those available from the Power Electronics community \cite{fuzzy-dc-dc-1996,time-optimal-dc-dc-2008} and from commercial products \cite{texas-instruments-buck-dc-dc}.






\section{Background}
\label{basic.tex}
\label{vars-notation.lab}

We denote with $[n]$ an initial segment $\{1,\ldots, n\}$ of the natural numbers. 
We denote with $X$ = $[x_1, \ldots, x_n]$ a
finite sequence (list) of variables.
By abuse of language we may regard sequences as sets and
we use $\cup$ to denote list concatenation.
Each variable $x$ ranges on a known (bounded or unbounded)
interval ${\cal D}_x$ either of the reals or of the integers (discrete
variables). 
We denote with ${\cal D}_X$ the set $\prod_{x\in X} {\cal D}_x$.
To clarify that a variable $x$ is {\em continuous} 
(i.e. real valued) we may write $x^{r}$. Similarly, 
to clarify that a variable $x$ is {\em discrete} 
(i.e. integer valued) we may write $x^{d}$. 
Boolean variables are discrete variables
ranging on the set $\B$ = \{0, 1\}. 
We may write $x^{b}$ to denote a boolean variable.
Analogously $X^{r}$ ($X^{d}$, $X^{b}$) denotes the sequence
of real (integer, boolean) variables in $X$.
Unless otherwise stated, we suppose
${\cal D}_{X^r} = \R^{|X^r|}$ and ${\cal D}_{X^d} = \Z^{|X^d|}$. 
Finally, if $x$ is a boolean variable 
we write $\bar{x}$ for $(1 - x)$.

\subsection{Predicates}
\label{subsection:predicates}
 
A {\em linear expression} $L(X)$ over a list of variables $X$ is a linear
combination of variables in $X$ with rational coefficients, $\sum_{x_i \in X}{a_i x_i}$. 
%
A {\em linear constraint} over $X$ (or simply a {\em constraint})  is an expression of the form
$L(X) \leq b$,
where $L(X)$ is a linear expression over $X$ 
and $b$ is a rational constant. In the following, we also write $L(X) \geq b$
for $-L(X) \leq -b$.

{\em Predicates} are inductively defined as follows.
A constraint $C(X)$ over a list of variables $X$ is a predicate over 
$X$. 
If $A(X)$ and $B(X)$ are predicates over $X$, then $(A(X) \land B(X))$
and $(A(X) \lor B(X))$ are predicates over X.  Parentheses may be
omitted, assuming usual associativity and precedence rules of logical
operators.
A {\em conjunctive predicate} is a conjunction of constraints.
For conjunctive predicates we will also write: 
$L(X) = b$ for (($L(X) \leq b$) $\wedge$ ($L(X) \geq b$)) and $a \leq x \leq
b$ for $x \geq a \;\land\; x \leq
b$, where $x \in X$.
%


A {\em valuation} over a list of variables $X$
is a function $v$ that maps each variable $x \in X$ to a value $v(x)
\in {\cal D}_x$.
Given a valuation $v$, we denote with $X^\ast\in {\cal D}_X$ the sequence of values 
$[v(x_1),\ldots,v(x_n)]$. By abuse of language, 
we call valuation also the sequence of values $X^\ast$.
A \emph{satisfying assignment} to a predicate $P$ over $X$ is a
valuation $X^{*}$ such that $P(X^{*})$ holds. If a satisfying assignment to a
predicate $P$ over $X$ exists, we say that $P$ is {\em feasible}.
Abusing notation, we may
denote with $P$ the set of satisfying assignments to the predicate 
$P(X)$. Two predicates $P$ and $Q$ over $X$ are {\em equivalent},
denoted by $P\equiv Q$, if they have the same set of 
satisfying assignments.

A variable $x\in X$ 
is said to be {\em bounded} in $P$ if 
there exist $a$, $b \in {\cal D}_x$
such that $P(X)$ implies $a \leq x \leq b$.
A predicate $P$ is bounded if all its variables are bounded.
%
%
%
%
%

Given a constraint $C(X)$ and a fresh boolean variable ({\em guard}) $y \not\in X$,
the {\em guarded constraint} $y \to C(X)$ (if $y$ then $C(X)$) denotes
the predicate $((y = 0) \lor C(X))$. Similarly, we use $\bar{y} \to
C(X)$ (if not $y$ then $C(X)$) to denote the predicate $((y = 1) \lor
C(X))$.
A {\em guarded predicate} is a conjunction of 
either constraints or guarded constraints.
It is possible to show that, 
if a guarded predicate $P$ is bounded,
then $P$ can be transformed into a (bounded) conjunctive predicate,
see~\cite{ICSEA2012}. 

\subsection{Mixed Integer Linear Programming}
\label{subsection:milp}

A \glsfirst{MILP} problem with
\emph{decision variables} $X$ is a tuple $(\max,$ $J(X),$ $A(X))$
where: $X$ is a list of variables, $J(X)$ (\emph{objective function})
is a linear expression on $X$, and $A(X)$ (\emph{constraints}) is a
conjunctive predicate on $X$.  
A {\em solution} to $(\max, J(X), A(X))$ is a valuation $X^{*}$ 
such that $A(X^{*})$ and $\forall Z \; (A(Z) \;
\rightarrow \; (J(Z) \leq J(X^{*})))$. 
$J(X^{*})$ is the {\em optimal value} of the \gls{MILP} problem. 
A {\em feasibility} problem is a \gls{MILP} problem of the form
$(\max, 0, A(X))$. We write also $A(X)$
for  $(\max, 0, A(X))$. We write $(\min, J(X), A(X))$ for $(\max, -J(X), A(X))$.

In algorithm outlines, \gls{MILP} solver invocations are denoted by function
\fun{feasible}($A(X)$) that returns {\sc True} if $A(X)$ is 
feasible and {\sc False} otherwise,
and function \fun{optimalValue}($\max$, $J(X)$, $A(X)$) 
that returns either the optimal value of the \gls{MILP} problem 
($\max$, $J(X)$, $A(X)$) or $+\infty$ if such \gls{MILP} problem is
unbounded or unfeasible.

\subsection{Labeled Transition Systems}
\label{lts.tex}

A \gls{LTS}
is a tuple
${\cal S} = (S, A, T)$ where 
$S$ is a (possibly infinite) set of states, 
$A$ is a (possibly infinite) set of \emph{actions}, and 
$T$ : $S$ $\times$ $A$ $\times$ $S$ $\rightarrow$ $\B$
is the \emph{transition relation} of ${\cal S}$.
We say that $T$ (and ${\cal S}$) is {\em deterministic} if $T(s, a, s') \land
T(s, a, s'')$ implies $s' = s''$, and {\em nondeterministic}
otherwise. 
Let $s \in S$ and $a \in A$.
%
%
%
%
%
%
We denote with 
$\mbox{\rm Adm}({\cal S}, s)$ the set of actions
admissible in $s$, that is $\mbox{\rm Adm}({\cal S}, s)$ = $\{a \in A
\; | \; \exists s': T(s, a, s') \}$
and with
$\mbox{\rm Img}({\cal S}, s, a)$ the set of next
states from $s$ via $a$, that is $\mbox{\rm Img}({\cal S}, s, a)$ =
$\{s' \in S \; | \; T(s, a, s') \}$.
We call {\em transition} a triple $(s, a, s') \in S$ $\times$ $A$ $\times$ $S$,
and {\em self loop} a transition $(s, a, s)$. A transition $(s, a, s')$ [self
loop $(s, a, s)$] is a
{\em transition [self loop] of ${\cal S}$} iff $T(s, a, s')$ [$T(s, a, s)$].
A {\em run} or \emph{path}
for an \gls{LTS} ${\cal S}$ 
is a sequence 
$\pi$ =
$s_0, a_0, s_1, a_1, s_2, a_2, \ldots$ 
of states $s_t$ and actions $a_t$ 
such that
$\forall t \geq 0$ $T(s_t, a_t, s_{t+1})$.
The length $|\pi|$ of a finite run $\pi$ is the number of actions
in $\pi$. 
We denote with $\pi^{(S)}(t)$ the $(t + 1)$-th state element of
$\pi$, and with $\pi^{(A)}(t)$ the $(t + 1)$-th action element of
$\pi$. That is $\pi^{(S)}(t)$ = $s_t$, and $\pi^{(A)}(t)$ = $a_t$.

Given two \glspl{LTS} ${\cal S}_1$ $=$ $(S$, $A$, $T_1)$ and ${\cal S}_2$ $=$ $(S$,
  $A$, $T_2)$, we say that ${\cal S}_1$ \emph{refines} ${\cal
    S}_2$ (denoted by ${\cal S}_1 \sqsubseteq {\cal S}_2$) iff $T_{1}(s,
  a, s')$ implies $T_{2}(s, a, s')$ for each state $s,s' \in S$ and action $a \in A$.
The refinement relation is a partial order on \glspl{LTS}.

\section{Discrete Time Linear Hybrid Systems} \label{dths.tex}

In this section we introduce our class of 
\glsfirst{DTLHS},
together with the \gls{DTLHS} representing the buck DC-DC converter 
on which our experiments will focus.

 

\begin{definition}[\gls{DTLHS}] 
\label{dths.def}
A
\glsdesc{DTLHS} 
is a tuple ${\cal H} = (X,$ $U,$ $Y,$ $N)$ where:
\begin{itemize}

\item 
  $X$ = $X^{r} \cup X^{d}$ 
  is a finite sequence of real ($X^{r}$) and 
  discrete ($X^{d}$) 
  {\em present state} variables.  
  We denote with $X'$ the sequence of 
  {\em next state} variables obtained 
  by decorating with $'$ all variables in $X$.
  
\item 
  $U$ = $U^{r} \cup U^{d}$ 
  is a finite sequence of 
  \emph{input} variables.

\item 
  $Y$ = $Y^{r} \cup Y^{d}$ 
  is a finite sequence of
  \emph{auxiliary} variables. 
  Auxiliary variables are typically used to
  model \emph{modes} (e.g., from switching elements such as diodes) 
  or ``local'' variables.

\item 
  $N(X, U, Y, X')$ is a conjunctive predicate 
  over $X \cup U \cup Y \cup X'$ defining the 
  {\em transition relation} (\emph{next state}) of the system.
  $N$ is {\em deterministic} if $N(x, u, y_1, x') \land N(x, u, y_2, x'')$
  implies $x' = x''$, and {\em nondeterministic} otherwise.
\end{itemize}
A \gls{DTLHS} is {\em bounded} if predicate $N$ is bounded.
A \gls{DTLHS} is {\em deterministic} if $N$ is deterministic.
\end{definition}

Since any bounded guarded predicate can 
be transformed into a conjunctive predicate (see Sect.~\ref{subsection:predicates}), for the sake of readability we will 
use bounded guarded predicates to describe the transition relation of 
bounded \glspl{DTLHS}. To this aim, we will also clarify which variables are boolean,
and thus may be used as guards in guarded constraints.


\begin{example}
\label{ex:dths}
Let 
$x$ be a continuous variable, $u$ be a boolean variable, and
$N(x, u, x') \equiv [\overline{u} \rightarrow x' = \alpha x] 
\land
[u \rightarrow x' = \beta x] \land -4 \leq x \leq 4$ be a guarded predicate with $\alpha = \frac{1}{2}$ and $\beta = \frac{3}{2}$.
Then ${\cal H}=(\{x\},\{u\},\varnothing,N)$ is a bounded \gls{DTLHS}.
Note that ${\cal H}$ is deterministic. 
Adding nondeterminism to ${\cal H}$ allows us to address
the problem  of (bounded) variations in the \gls{DTLHS} parameters.
For example, variations in the parameter
$\alpha$ can be modelled with a tolerance 
$\rho \in [0, 1]$ for $\alpha$. 
This replaces $N$ with:
$N^{(\rho)} \equiv [\overline{u} \rightarrow x' \leq (1 + \rho) \alpha x]$ 
$\land$
$[\overline{u} \rightarrow x' \geq (1 - \rho) \alpha x]$ 
$\land$
$[u \rightarrow x' = \beta x]$.
We have that ${\cal H}^{(\rho)}=(\{x\},\{u\},\varnothing,N^{(\rho)})$, for $\rho \in (0, 1]$, is a
nondeterministic \gls{DTLHS}. Note that, as expected, ${\cal H}^{(0)} = {\cal H}$.
\end{example}

In the following definition, we give the semantics of 
\glspl{DTLHS} in terms of \glspl{LTS}. 

\begin{definition}[\gls{DTLHS} dynamics]
Let ${\cal H}$ = ($X$, $U$, $Y$, $N$) 
be a \gls{DTLHS}.
The dynamics of ${\cal H}$ 
is defined by the \glsdesc{LTS}
$\mbox{\rm LTS}({\cal H})$ = (${\cal D}_X$, ${\cal D}_U$,
$\tilde{N}$) where:
$\tilde{N} : {\cal D}_X \; \times \; {\cal D}_U \; \times \; {\cal D}_X \rightarrow \B$ 
is a function s.t.  $\tilde{N}(x, u, x') \equiv \exists y \in {\cal D}_Y: N(x, u, y, x')$.
A \emph{state} $x$ for ${\cal H}$ is a state $x$ for 
$\mbox{\rm LTS}({\cal H})$ and a \emph{run} 
(or \emph{path}) for ${\cal H}$ is
a run for $\mbox{\rm LTS}({\cal H})$ (Sect. \ref{lts.tex}).
\end{definition}

\begin{example}
\label{ex:dths-lts}
Let ${\cal H}$ be the \gls{DTLHS} of Ex.~\ref{ex:dths}.
Then a sequence $\pi$ is a run for ${\cal H}$ iff 
state $\pi^{(S)}(i + 1)$ is obtained by multiplying $\pi^{(S)}(i)$ by $\frac{3}{2}$
when $\pi^{(A)}(i) = 1$, and by $\frac{1}{2}$ when $\pi^{(A)}(i) = 0$.
\end{example}


\subsection{Buck DC-DC Converter as a \gls{DTLHS}}
\label{example-buck.tex}
\label{sec:buck}
The buck DC-DC converter (Fig. \ref{buck.eps}) is a mixed-mode analog
circuit converting the DC input voltage ($V_i$ in Fig.
\ref{buck.eps}) to a desired DC output voltage ($v_O$ in Fig.
\ref{buck.eps}).
As an example, buck DC-DC converters are used off-chip to scale down
the typical laptop battery voltage (12-24) to the just few volts
needed by the laptop processor (e.g. \cite{fuzzy-dc-dc-1996}) as well
as on-chip to support
\newacronym{DVFS}{DVFS}{Dynamic Voltage and Frequency Scaling}%
\gls{DVFS}
in multicore processors
(e.g. \cite{gigascale-integration-07,buck-dc-dc-at-intel-pesc2004}).
Because of its widespread use,
%
control schemas for buck DC-DC converters have been widely studied
(e.g. see
\cite{gigascale-integration-07,buck-dc-dc-at-intel-pesc2004,fuzzy-dc-dc-1996,time-optimal-dc-dc-2008}).
The typical software based approach (e.g. see \cite{fuzzy-dc-dc-1996}) is to control
the switch $u$ in Fig. \ref{buck.eps} (typically implemented with a
MOSFET) with a microcontroller.

\begin{figure}[hbt!]
  \begin{center}
    \scalebox{0.45}{\input{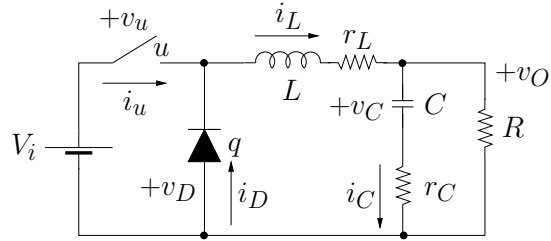}} 
    \caption{Buck DC-DC converter.}
    \label{buck.eps}
  \end{center}
\end{figure}
Designing the software to run on the microcontroller to properly
actuate the switch is the control software design problem for the buck DC-DC
converter in our context.

The circuit in Fig.~\ref{buck.eps} can be modeled as a \gls{DTLHS} ${\cal
  H}$ = ($X$, $U$, $Y$, $N$) in the following way~\cite{buck-tekrep-art-2011}.
  As for
  the sets of variables, we have 
$X$ $=$ $X^{r}$ $=$ 
$[i_L$, $v_O]$,
$U$ $=$ $U^{b}$ $=$ $[u]$,
$Y$ $=$ $Y^{r}\cup Y^{b}$ with $Y^{r}$ $=$ 
$[i_u$, $v_u$, $i_D$, $v_D]$ and $Y^{b}$ $=$ 
$[q]$.  
As for $N$, it is given by the conjunction of the following (guarded) constraints:
\begin{eqnarray}
  {i_L}' & = & (1 + Ta_{1,1})i_L + Ta_{1,2}v_O + Ta_{1,3}v_D  \label{next-il} \label{next-il.eq} \\
  {v_O}' & = & Ta_{2,1}i_L + (1 + Ta_{2,2})v_O + Ta_{2,3}v_D. \label{next-vc}\label{next-vc.eq}
\end{eqnarray}
\begin{center}
  \begin{tabular}{lcr}  
    \begin{minipage}{0.3\textwidth}
      \begin{eqnarray}
        q & \rightarrow & v_D = 0 \\
        q & \rightarrow & i_D \geq 0 \\
        u & \rightarrow & v_u = 0 
      \end{eqnarray}
    \end{minipage} &
    \begin{minipage}{0.3\textwidth}
      \begin{eqnarray}
        i_D & = & i_L - i_u  \\
        \bar{q} & \rightarrow & v_D \leq 0 \\
        \bar{q} & \rightarrow &  v_D = R_{off}i_D 
      \end{eqnarray}
    \end{minipage} &
    \begin{minipage}{0.3\textwidth}
      \begin{eqnarray}
        \bar{u} & \rightarrow & v_u = R_{off} i_u   \label{mosfet-off-eq}\label{buck-lastbutone-eq}\\
        v_D & = & v_u - V_i \label{eq:Vi}\label{Vi.eq}\label{buck-last-eq}
      \end{eqnarray}
    \end{minipage} \\ 
  \end{tabular}
\end{center}

\noindent where the coefficients $a_{i, j}$ depend on the circuit parameters $R$, $r_L$, $r_C$, $L$ and $C$ in the following way:
$a_{1,1} = -\frac{r_L}{L}$,
$a_{1,2} = -\frac{1}{L}$,
$a_{1,3} = -\frac{1}{L}$,
$a_{2,1} = \frac{R}{r_c + R}[-\frac{r_c r_L}{L} + \frac{1}{C}]$,
$a_{2,2} = \frac{-1}{r_c + R}[\frac{r_c R}{L} + \frac{1}{C}]$,
$a_{2,3} = -\frac{1}{L}\frac{r_c R}{r_c + R}$.



\section{Quantized Feedback Control}
\label{lts-ctr.tex}

In this section, we formally define the Quantized Feedback Control
Problem for \glspl{DTLHS} (Sect.~\ref{quantized-feedback-control.tex}). 
To this end, first we give the definition
of Feedback Control Problem for \glspl{LTS} (Sect.~\ref{lts-cts.tex-lts}),
and then for \glspl{DTLHS} (Sect.~\ref{ctr-dths.tex}).
Finally, we show that our definitions are well founded (Sect.~\ref{proof.mgo}).


\subsection{Feedback Control Problem for \glspl{LTS}}
\label{lts-cts.tex-lts}

We begin by extending to possibly infinite \glspl{LTS} 
the definitions in
\cite{Tro98,strong-planning-98} 
for finite \glspl{LTS}.
%
In what follows, let ${\cal S} = (S, A, T)$ be an \gls{LTS}, 
and $I, G \subseteq S$ be, respectively, 
the {\em initial} and {\em goal} regions.

\begin{definition}[\gls{LTS} control problem]
  \label{def:ctroller-lts}\label{def:ctrproblem-lts}
  A \emph{controller} for an \gls{LTS} ${\cal S}$ is a function $K : S \times A
  \to \B$ such that $\forall s \in S$, $\forall a \in A$, if $K(s,
  a)$ then $a\in{\rm Adm}({\cal S},s)$.
  We denote with $\mbox{\rm Dom}(K)$ the set of states for which a
  control action is defined. Formally, $\mbox{\rm Dom}(K)$ $=$ $\{s
  \in S \; | \; $$\exists a: K(s, a)\}.$
  ${\cal S}^{(K)}$ denotes the \emph{closed loop system}, that
  is the \gls{LTS} $(S, A, T^{(K)})$, where $T^{(K)}(s, a, s') = T(s, a, s')
  \wedge K(s, a)$.
  A \emph{control law} for a controller $K$ is a (partial) function
  $k : S \rightarrow A$ s.t. for all $s \in {\rm Dom}(K)$ we have that $K(s, k(s))$ holds.
  By abuse of language we say that a controller is a control law
  if for all $s \in S$, $a, b \in A$ it holds that $(K(x, a) \wedge K(x, b)) \rightarrow (a = b)$.
  An \emph{\gls{LTS} control problem} is a triple $({\cal S},$
  $I,$ $G)$.
\end{definition}

%
\begin{example}
\label{example-ctr.tex}
Let $S=\{-1,0,1\}$ and $A=\{0, 1\}$. 
Let ${\cal S}_0$ be the \gls{LTS} $(S, A, T_0)$, 
where the transition relation $T_0$
consists of the continuous arrows in 
Fig.~\ref{strongWeakSolutions.fig}. 
A function $K$
is a controller for ${\cal S}_0$ iff $(s \neq 0) \to (K(s, 1) = 0)$. As an example,
we have that $K$ defined as $K(s, a) = ((s \neq 0) \to (a \neq 1))$ is a controller but not
a control law, and that $k(s) = 0$ is a control law for $K$ (note that $K(s, a) = (a = 0)$ is a control law).
\end{example}

Def.~\ref{def:ctrproblem-lts} also introduces the formal definition of {\em control law},
as our model of control software,
i.e. of how function \texttt{Control\_Law} in Fig.~\ref{control-loop-figure.tex} must behave.
Namely, while a controller may enable many actions in a given state,
a control law (i.e. the final software implementation) must provide only one action.
Note that the notion of controller is important because it contains
all possible control laws.

In the following we give formal definitions of strong and weak
solutions to a control problem for an \gls{LTS}.

We call a path $\pi$ {\em fullpath} 
if either it is infinite or its last state  $\pi^{(S)}(|\pi|)$ 
has no successors (i.e. $\mbox{\rm Adm}({\cal S}, \pi^{(S)}(|\pi|)) = \varnothing$).
We denote with ${\rm Path}({\cal S}, s, a)$ the set of fullpaths of ${\cal S}$ starting in state
$s$ with action $a$, i.e. the set of fullpaths $\pi$ such that $\pi^{(S)}(0)=s$
and $\pi^{(A)}(0)=a$.

\sloppy

Given a path $\pi$ in ${\cal S}$, we define the measure
$J({\cal S},G,\pi)$ on paths as the distance of $\pi^{(S)}(0)$ to the goal on $\pi$. That is, 
if there exists $n > 0$ s.t.
$\pi^{(S)}(n)\in G$, then $J({\cal S},G,\pi)$ $=$ $\min\{n$ $|$ $n >
0 \land \pi^{(S)}(n)\in G\}$. Otherwise, $J({\cal S},G,\pi) = +\infty$.
We require $n > 0$ since
our systems are nonterminating and each controllable state (including a goal state)
must have a path of positive length to a goal state.
Taking $\sup \varnothing = +\infty$ and $\inf \varnothing = -\infty$, the {\em
worst case distance} (pessimistic view) of a state $s$ from the goal region  $G$
is  $J_{\rm strong}({\cal S},G,s)=\sup \{ J^{(S)}({\cal S},G,s, a)~|~ a \in {\rm
Adm}({\cal S},s)\}$, 
where: $J^{(S)}({\cal S},G,s,a)=\sup \{ J({\cal S},G,\pi)~|~ \pi
\in{\rm Path}({\cal S}, s,a)\}$. 
The {\em best case distance}  (optimistic view) of a
state $s$  from the goal region $G$ is  $J_{\rm weak}({\cal S},G,s)$ $=$ $\sup \{
J^{(W)}({\cal S},G,s, a)$ $|$ $a \in {\rm Adm}({\cal S},s)\}$, 
where: $J^{(W)}({\cal
S},G,s,a)$ $=$ $\inf \{ J({\cal S},G,\pi)$ $|$ $\pi \in{\rm Path}({\cal S}, s,a)\}$.
%
%

\fussy

\begin{definition}[Solution to \gls{LTS} control problem]
\label{def:sol}
Let ${\cal P}$ = $({\cal S}$, $I$, $G)$ 
be an \gls{LTS} control problem and $K$ be  
a controller  for ${\cal S}$ such that 
$I$ $\subseteq$ $\mbox{\rm Dom}(K)$.
$K$ is a {\em strong [weak] solution} to ${\cal P}$ 
if for all $s \in \mbox{\rm Dom}(K)$, 
$J_{\rm strong}({\cal S}^{(K)}, G, s)$ [$J_{\rm weak}({\cal S}^{(K)}, G, s)$]
is finite.
An \emph{optimal strong [weak] solution} to ${\cal P}$
is a strong [weak] solution $K^{*}$ to ${\cal P}$ such that 
for all strong [weak] solutions
$K$ to ${\cal P}$, for all $s \in S$ we have that
$J_{\rm strong}({\cal S}^{(K^{*})}, G, s) 
\leq J_{\rm strong}({\cal S}^{(K)}, G, s)$
[$J_{\rm weak}({\cal S}^{(K^{*})}, G, s) \leq 
J_{\rm weak}({\cal S}^{(K)}, G, s)$].
\end{definition}

Intuitively, a strong solution $K$ takes a \emph{pessimistic} view by
requiring that for each initial state, \emph{all} runs in the closed
loop system ${\cal S}^{(K)}$ reach the goal, no matter nondeterministic outcomes.
A weak solution $K$ takes an \emph{optimistic} view about nondeterminism: 
it just asks that for each action $a$ enabled in a given state $s$, 
there exists at least a path in ${\rm Path}({\cal S}^{(K)}, s,a)$ leading to the goal.
Unless otherwise stated, we say \emph{solution} for
\emph{strong solution}.

Finally, we define the \emph{most general optimal strong [weak] solution} 
to ${\cal P}$ ({\em strong [weak] mgo} in the following) as the unique strong [weak] optimal
solution to ${\cal P}$ enabling as many actions as possible (i.e., the most liberal one).
In Sect.~\ref{proof.mgo} we show that the definition of mgo 
is well posed.

\begin{example}
\label{example-strong.tex}
\sloppy
Let ${\cal S}_0, {\cal S}_1$ be the \glspl{LTS} in 
Fig.~\ref{strongWeakSolutions.fig} (see also Ex.~\ref{example-ctr.tex}).
Let ${\cal P}_0 = ({\cal S}_0, I, G)$ and 
${\cal P}_1 = ({\cal S}_1, I, G)$ be two control problems, 
where 
$I$ 
$=$ $\{-1, 0, 1\}$ and 
$G$ $=$ 
$\{0\}$. 
The controller $K(s, a) \equiv [s \neq 0 
\to a = 0]$
is a strong solution to the control problem ${\cal P}_0$. 
Observe that $K$ is not optimal.
Indeed, the controller $\tilde{K}(s, a)\equiv a=0$ is such that
$J_{\rm strong}({\cal S}^{(\tilde{K})}_0, G, 0)=1 < 2 = J_{\rm strong}({\cal S}^{(K)}_0, G, 0)$. 
The control problem ${\cal P}_1$
has no strong solution. 
As a matter of fact, to drive the system to the goal region $\{0\}$, 
any solution $K$ must enable action $0$ in states $-1$ and $1$: 
in such a case, however, we have that  
$J_{\rm strong}({\cal S}^{(K)}_1, \hat{G}, 1)$ $=$ 
$J_{\rm strong}({\cal S}^{(K)}_1, \hat{G}, -1)$ $=$ $\infty$
because of the self loops 
$(1,0,1)$ 
and $(-1,0,-1)$ of $T_1$.
Finally, note that $K$ is the weak mgo for ${\cal P}_1$ and $\tilde{K}$ is the strong mgo for ${\cal P}_0$.
\fussy
\end{example}

\begin{figure}[!t]
\centering
$
 \xymatrix@C=7mm{
      & 
	  *+=<25pt>[o][F-]{-1} 
			\ar@(dl,dr)@{.>}[]_{1}
			\ar@(ul,ur)@{.>}[]^{0}
			\ar@/^/[r]^{0}
        & *+=<25pt>[o][F=]{0} 
 			\ar@/^/[l]^{1}
			\ar@/^/[r]^{1}
			\ar@(dl,dr)[]_{0}
			\ar@(ul,ur)[]^{1}
	& *+=<25pt>[o][F-]{1} 
			\ar@(dl,dr)@{.>}[]_{1}
			\ar@(ul,ur)@{.>}[]^{0}
 			\ar@/^/[l]^{0}
	&\\
	}
$ 
\caption{\glspl{LTS} ${\cal S}_0$ (continuous arrows) and ${\cal S}_1$ (all arrows). Double circle represents the goal state.}
\label{strongWeakSolutions.fig}
\end{figure}
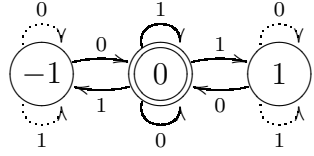

\begin{remark}
\label{stability:remark}
Note that if $K$ is a strong solution to $({\cal S}$, $I$, $G)$ and 
$G \subseteq I$ (as is usually the case in control problems) then
${\cal S}^{(K)}$ is {\em stable} from $I$ to $G$,
that is each run in ${\cal S}^{(K)}$
starting from a state in $I$ leads to a state in $G$. 
In fact, 
from Def.~\ref{def:sol} we have that each state $s \in I$ reaches 
a state $s' \in G$ in a finite number of steps. Moreover,
since $G \subseteq I$, we have that any state $s \in G$ reaches a state
$s' \in G$ in a finite number of steps. 
Thus, any path starting in $I$ 
in the closed loop system ${\cal S}^{(K)}$
\emph{touches} $G$ an infinite number of times ({\em liveness}).
\end{remark}



\subsection{Feedback Control Problem for \glspl{DTLHS}}
\label{ctr-dths.tex}

A control problem for a \gls{DTLHS} ${\cal H}$ is the \gls{LTS}
control problem induced by the dynamics of ${\cal H}$.
For \glspl{DTLHS}, we only consider control problems
where $I$ and $G$ can be represented as 
predicates over present state variables of ${\cal H}$.

\begin{definition}[\gls{DTLHS} control problem]
\label{def:DTLHScp}
Given a \gls{DTLHS} ${\cal H} =(X, U, Y, N)$ and predicates $I$ and $G$ over $X$, the
{\em \gls{DTLHS} (feedback) control problem} $({\cal H}, I, G)$ is the \gls{LTS} control
problem $(LTS({\cal H}), I, G)$. Thus, a controller $K: {\cal D}_X \times {\cal
D}_U \to \B$ is a {\em strong [weak] solution} to $({\cal H}, I, G)$ iff it is a
strong [weak] solution to $(LTS({\cal H}), I, G)$.
\end{definition}
 
For \gls{DTLHS} control problems, usually \emph{robust} controllers 
are desired. That is, controllers that, notwithstanding 
nondeterminism in the plant (e.g. due to parameter variations, see Ex.~\ref{ex:dths}), 
drive the plant state to the goal region. 
For this reason we focus on strong solutions.

Observe that the feedback controller for a \gls{DTLHS} will only measure
present state variables (e.g., output voltage and inductor current
in Sect.~\ref{example-buck.tex}) and will not measure auxiliary
variables (e.g. diode state in Sect.~\ref{example-buck.tex}).

\begin{example}
\label{example-goal}
The typical goal of a controller for the buck DC-DC converter in
Sect.~\ref{example-buck.tex} is keeping 
the output voltage $v_{O}$ \emph{close enough} to a  given reference value 
$V_{\rm ref}$.
This leads to the \gls{DTLHS} control problem 
${\cal P}$ = $({\cal H}$, $I$, $G)$
where 
${\cal H}$ is defined in Sect.~\ref{example-buck.tex},
$I \equiv$
$(|i_L| \leq 2)$ $\wedge$  
$(0 \leq v_O \leq 6.5)$,   
$G \equiv$
$(|v_O - V_{\rm ref}| \leq \theta)$
$\wedge$
$(|i_L| \leq 2)$, 
and $\theta = 0.01$ 
is the desired buck precision.
\end{example}

\subsection{Quantized Feedback Control Problem}
\label{quantized-feedback-control.tex}

\sloppy
Software running on a microcontroller ({\em control software} in the
following)
cannot handle real values. 
For this reason real valued state feedback from plant sensors
undergoes an 
\glsfirst{AD}
conversion before being sent to the
control software.
This process is called \emph{quantization} 
(e.g. see \cite{quantized-ctr-tac05} and citations thereof).
A 
\glsfirst{DA}
conversion is needed to transform
the control software digital output into real values to be sent 
to plant actuators.
%
%
In the following, we formally define quantized solutions to a 
\gls{DTLHS} feedback control problem.

\fussy

\begin{definition}[Quantization function]
\label{def:quantization-function}
A {\em quantization function} $\gamma$
for a real interval $I=[a,b]$ is a non-decreasing   
function $\gamma:[a,b]\to \hat{I}$, 
where $\hat{I}$ is a bounded integer interval $[\gamma(a),\gamma(b)]\subseteq\Z$.
The \emph{quantization step} of $\gamma$, denoted by $\|\gamma\|$, 
is defined as $\sup\{ \; |w-z|
  \; | \; w,z \in I \land \gamma(w)=\gamma(z)\}$. 
\end{definition}


%

For ease of notation, we extend quantizations to integer intervals, by
stipulating that in such a case the quantization function is the identity
function (i.e. $\gamma(x) = x$). Note that, with this convention, the
quantization step on an integer interval is always $0$.



\begin{definition}[Quantization for \glspl{DTLHS}]
\label{def:quantization}
  Let ${\cal H} = (X, U, Y, N)$ be a \gls{DTLHS}, and 
  let $W=X\cup U$. 
  A \emph{quantization} ${\cal Q}$ for $\cal H$
  is a pair $(A, \Gamma)$, where:
\begin{itemize}
\item $A$ is a predicate of form $\wedge_{w \in W} (a_w \leq w \leq b_w)$ with $a_w, b_w \in {\cal D}_w$.
  For each $w\in W$, we define $A_w = \{v \in {\cal D}_w\;|\; a_w \leq v \leq b_w\}$
  as the 
  {\em admissible region} for variable $w$. Moreover, we define $A_V=\prod_{v\in V} A_v$, with $V \subseteq W$, as the 
  admissible region for variables in $V$.
  
  \item 
  $\Gamma$ is a set of maps $\Gamma = \{\gamma_w$ $|$
  $w \in W$ and $\gamma_w$ is a  
  quantization function for $A_w\}$. 
  \end{itemize}
  Let $V = [w_1, \ldots, w_k]$ 
  and $v = [v_1, \ldots, v_k] \in A_{V}$, where $V \subseteq W$. 
  We write $\Gamma(v)$ (or $\hat{v}$) for the tuple $[\gamma_{w_1}(v_1),$ 
  $\ldots,$ $\gamma_{w_k}(v_k)]$ and $\Gamma^{-1}(\hat{v})$ for the set $\{v \in A_V$ $|$ $\Gamma(v) = \hat{v}\}$.
  Finally, the \emph{quantization step} $\|\Gamma\|$ for $\Gamma$ is
  defined as $\sup \{ \; \|\gamma\| \; | \; \gamma \in
  \Gamma \}$.
\end{definition}

For ease of notation, in
the following we will also consider quantizations for primed variables $x' \in
X'$, by stipulating that $\gamma_{x'} \equiv \gamma_x$.


\begin{example}
\label{ex:quantization} 
Let ${\cal H}$  be the \gls{DTLHS} described in Ex.~\ref{ex:dths}.
Let us consider the quantization ${\cal Q}=(A,\Gamma)$,
where $A \equiv$
$-2.5\leq x\leq 2.5\land 0\leq u\leq 1$.
$A$ defines the admissible region $A_x=A_X=[-2.5,2.5]$.
Let $\Gamma=\{\gamma_x,\gamma_u\}$, with $\gamma_x(x)=round(x/2)$ 
(where $round(x) = \lfloor x \rfloor + 
\lfloor 2(x - \lfloor x\rfloor)\rfloor$ 
is the usual rounding function) and
$\gamma_u(u) = u$. Note that $\gamma_x(x) = -1$ for all $x \in [-2.5, -1]$,
$\gamma_x(x) = 0$ for all $x \in (-1, 1)$ and $\gamma_x(x) = 1$ for all $x
\in [1, 2.5]$.
Thus, we have that $\Gamma(A_x)=\{-1,0,1\}$, $\Gamma(A_u)=\{0,1\}$ 
and $\|\Gamma\|$ = 1.
\end{example}

Quantization, i.e. representing reals with integers, unavoidably introduces errors in
reading real-valued plant sensors in the control software. We address this
problem in the following way. First, 
we introduce the definition
of $\varepsilon$-solution. 
Essentially, we require that the
controller drives the plant ``near enough'' (up to a given error $\varepsilon$) to the goal region $G$.

\begin{definition}[$\varepsilon$-relaxation of a set]
\label{def:relaxation}
Let $\varepsilon \geq 0$ be a real number and
$W\subseteq \R^n\times \Z^m$. 
The $\varepsilon$-\emph{relaxation} of $W$ is the set
(\emph{ball} of radius $\varepsilon$)
${\cal B}_{\varepsilon}(W)$
= $\{(z_1, \ldots z_n$, $q_1, \ldots q_m)$
$|$ $\exists (x_1$, $ \ldots$, $ x_n$, $ q_1$, $ \ldots, q_m):
(x_1$, $ \ldots$, $ x_n$, $ q_1$, $ \ldots$, $ q_m) \in W$ and $\forall i \in \{1, \ldots n\} \; 
|z_i - x_i| \leq \varepsilon\}$.
\end{definition}

\begin{definition}[$\varepsilon$-solution to \gls{DTLHS} control problem]
\label{def:epsilon-solution}
Let ${\cal P}=({\cal H}, I, G)$ be a \gls{DTLHS} control problem 
and let $\varepsilon>0$ be a real number.
A {\em strong [weak] $\varepsilon$-solution} to ${\cal P}$ is a strong [weak]
solution to the \gls{LTS} control problem $(LTS({\cal H}), I, {\cal B}_\varepsilon(G))$.

\end{definition} 

\begin{example}
\label{ex:ctr-dths}
Let ${\cal H}$  be the \gls{DTLHS} described in Ex.~\ref{ex:dths}.
We consider the control problem defined by the
initial region $I=[-2.5,2.5]$ and the goal region $G=\{0\}$ 
(represented by the predicate $x = 0$).
The \gls{DTLHS} control problem ${\cal P}$ $=$ $({\cal H}, I, G)$ 
has no solution (because of the Zeno phenomenon), 
but for all $\varepsilon>0$ it has
the $\varepsilon$-solution $K$ such that $\forall x \in I. \; K(x,0)$.
\end{example}

Second, we introduce the definition of \emph{quantized solution} to a  \gls{DTLHS}
control problem for a given quantization ${\cal Q} = (A, \Gamma)$. Essentially,
a quantized solution models the fact that  in an SBCS control decisions are
taken by the control software by just looking at quantized state values. Despite
this, a quantized solution guarantees that each \gls{DTLHS} initial state reaches a
\gls{DTLHS} goal state (up to an error at most $\|\Gamma\|$).

\sloppy

\begin{definition}[\glsdesc{QFC} solution to \gls{DTLHS} control problem]
  \label{def:qfc}
  Let ${\cal H} = (X, U, Y, N)$ be a \gls{DTLHS},
  ${\cal Q}=(A,\Gamma)$ be a quantization for ${\cal H}$
  and ${\cal P} = ({\cal H}, I, G)$ be a \gls{DTLHS} control problem.
  A ${\cal Q}$ 
  \gls{QFC}
  \emph{strong [weak] solution}
  to ${\cal P}$ is a
  strong [weak] $\|\Gamma\|$-solution $K : {\cal D}_X \times {\cal D}_U \to \B$ 
  to ${\cal P}$ such that 
  $K(x, u) = 0$ if $(x, u) \notin A_X \times A_U$, and otherwise
  $K(x, u) = \hat{K}(\Gamma(x), \Gamma(u))$
  where 
  $\hat{K} : \Gamma(A_{X}) \times \Gamma(A_{U})$ $\rightarrow$ $\B$.
%
%
%
\end{definition}

\fussy

Note that a ${\cal Q}$ \gls{QFC} solution to a \gls{DTLHS} control problem does
not work outside the admissible region defined by ${\cal Q}$. This models the
fact that controllers for real-world systems must maintain the plant inside
given bounds (such requirements are part of the safety specifications). In the
following, we will define ${\cal Q}$ \gls{QFC} solutions by only specifying their
behavior inside the admissible region.

\begin{example}
\label{ex:gamma-qfc-sol} 
Let ${\cal P}$ be the \gls{DTLHS} control problem 
defined in Ex.~\ref{ex:ctr-dths} and ${\cal Q}=(A,\Gamma)$
be the quantization defined in Ex.~\ref{ex:quantization}.
Let $\hat{K}$ be defined by
$\hat{K}(\hat{x},\hat{u})\equiv [\hat{x}\not=0 \to \hat{u}=0]$.  
For any $\varepsilon>0$, the quantized controller $K(x, u) = 
\hat{K}(\Gamma(x),\Gamma(u))$ is an $\varepsilon$-solution
to ${\cal P}$,  
and hence it is a ${\cal Q}$ \gls{QFC} solution.
\end{example}

\sloppy

Along the same lines of similar undecidability
proofs~\cite{decidability-hybrid-automata-jcss98,AgrawalSTY06}, it is possible to show that
existence of a ${\cal Q}$ \gls{QFC} solution to a \gls{DTLHS} control problem
({\em \gls{DTLHS} quantized control problem}) is undecidable, as shown
in~\cite{ictac2012}.

\fussy

\begin{theorem} \label{th:undecidability}
The \gls{DTLHS} quantized control problem is undecidable.
\end{theorem}



\subsection{Proof of Uniqueness of the Most General Optimal Controller}
\label{proof.mgo}

In this section, we prove properties on mgo (see Sect.~\ref{lts-cts.tex-lts}).
This section can be skipped at a first reading.
We begin by giving the formal definition of strong and weak mgo.

\sloppy
\begin{definition}[Most general optimal solution to control problem]
\label{def:mgo}
The  \emph{most general optimal strong [weak] solution} 
to ${\cal P}$ 
is an optimal strong
[weak] solution 
$\tilde{K}$ to ${\cal P}$ such that  
for all other optimal strong [weak] solutions $K$ to ${\cal P}$,
for all $s \in S$,
for all $a \in A$
we have that $K(s, a)$ $\to$ $\tilde{K}(s, a)$. 
\end{definition}
\fussy

\begin{proposition}\label{prop:mgo}
  An \gls{LTS} control problem $({\cal S}, \varnothing, G)$ 
  has always an {\em unique} strong mgo $K^*$.
  Moreover, for all $I\subseteq S$, we have:

\begin{itemize}  
  \item if $I\subseteq {\rm Dom}(K^*)$, then $K^*$ is the unique strong mgo for
  the control problem $({\cal S}, I, G)$;
  
  \item if $I\not\subseteq {\rm Dom}(K^*)$, then the control problem $({\cal S}, I, G)$ has 	
no strong solution.
\end{itemize}  
\end{proposition}
\begin{proof}
Let ${\cal S}$ $=$ $(S, A, T)$ be an \gls{LTS}, and let $({\cal S}, I,G)$
be an \gls{LTS} control problem.
We define the sequences of sets $D_n$ and $F_n$ as follows:

\begin{itemize}
\item $D_0  =  \varnothing$
\item $F_1  =  \{ s\in S ~ | ~ \exists a\in A: a \in {\rm Adm}({\cal S}, s) \land {\rm Img}({\cal S}, s, a)\subseteq G\}$
\item $F_{n+1}  =  \{ s\in S\setminus D_n ~ | ~ \exists a\in A: a \in {\rm Adm}({\cal S}, s) \land {\rm Img}({\cal S}, s, a)\subseteq D_n\}$
\item $D_{n+1}  =  D_n\cup F_{n+1}$
\end{itemize}

Intuitively, $D_n$ is the set of states which can be driven inside $G$ in at most
$n$ steps, notwithstanding nondeterminism. $F_n$ is the subset of $D_n$
containing only those states for which at least a path to $G$ of length exactly $n$
exists.

The following properties hold for $D_n$ and $F_n$:
\begin{enumerate}

\item \label{proof.mgo.f_empty.item}  If $F_n=\varnothing$ for some $n\geq 1$,  then
for all $m\geq n$, $F_m=\varnothing$.  In fact, if $F_n=\varnothing$, then
$D_{n}=D_{n - 1}$, and hence $F_{n+1}=F_n=\varnothing$.

\item \label{proof.mgo.f_empty.itembis}  If $D_{n+1}=D_{n}$ for some $n\geq 0$,  then
for all $m\geq n$, $D_{m}=D_{n}$.  This immediately follows from the previous
point~\ref{proof.mgo.f_empty.item}.

%

\item \label{proof.mgo.d_cup_f.item}  $D_n=\bigcup_{1\leq j\leq n} F_j$ for $n \geq 1$
(also for $n \geq 0$ if we take the union of no sets to be $\varnothing$). We
prove this property by induction on $n$. As for the induction base, we have that
$D_1 = F_1$. As for the inductive step,  $D_{n+1}=D_n\cup F_{n+1} =
\bigcup_{1\leq j\leq n} F_j\cup F_{n+1}=\bigcup_{1\leq j\leq n+1} F_j$.

\item \label{proof.mgo.f_disj.item}  $F_i\cap F_j=\varnothing$ for all $i\not=j$. We
have that if $s\in F_{n+1}$ then $s\not\in D_n$. By previous
point~\ref{proof.mgo.d_cup_f.item}, we have that $s\not\in D_n$ implies $s\not\in F_j$
for $1 \leq j \leq n$. Hence, $s\in F_{n+1}$  implies that $s\not\in F_j$ for
all $1\leq j\leq n$. If by absurd a state $s$ exists s.t. $s \in F_i\cap F_j$
for some $i > j$, then $s \in F_i$ would imply $s \notin F_j$.
\end{enumerate}

For all $s\in S$ and $a\in A$, we define the controller  $\tilde{K} : S \times A
\to \B$ as follows: $ \tilde{K}(s,a)  \Leftrightarrow   (\exists n>1: s\in F_n
\wedge a \in {\rm Adm}({\cal S}, s) \land {\rm Img}({\cal S},s,a)\subseteq
D_{n-1}) \lor (s\in F_1 \wedge a \in {\rm Adm}({\cal S}, s) \land {\rm
Img}({\cal S},s,a)\subseteq G)  $.

Note that ${\rm Dom}(\tilde{K}) = \overline{D}=\bigsqcup_{n\in \N} D_n$, i.e.
the domain of $\tilde{K}$ is the least upper bound for sets $D_n$ (we
are not supposing $S$ to be finite, thus there may be a nonempty $D_n$ for any
$n \in \N$).

\sloppy

$\tilde{K}$ is a strong solution to $({\cal S}, \varnothing, G)$. To prove this,
we show that,  if $t\in F_n$, then $J_{\rm strong}({\cal S}^{(\tilde{K})}, G,
t)=n$ (note that $t \in {\rm Dom}(\tilde{K})$ implies $t\in F_n$ for some $n\geq
1$). In fact, if $t \in F_1$ then $J_{\rm strong}({\cal
S}^{(\tilde{K})},G,t)=\sup \{ J^{(S)}({\cal S}^{(\tilde{K})},G,t, a)$ $|$ $a \in
{\rm Adm}({\cal S}^{(\tilde{K})},t)\}=\sup \{ J^{(S)}({\cal
S}^{(\tilde{K})},G,t, a)$ $|$ $a$ is s.t. $\varnothing \neq {\rm Img}({\cal
S},t,a)\subseteq G\}=\sup \{\sup \{ J({\cal S}^{(\tilde{K})},G,\pi)$ $|$ $\pi
\in{\rm Path}({\cal S}^{(\tilde{K})},t,a)\}$ $|$ $a$ is s.t. $\varnothing \neq
{\rm Img}({\cal S},t,a)\subseteq G\}$ $=$ $\sup \{ J({\cal
S}^{(\tilde{K})},G,\pi)$ $|$ $\pi \in\{\pi \in {\rm Path}({\cal
S}^{(\tilde{K})},t,a)$ $|$ $a$ is s.t. $\varnothing \neq {\rm Img}({\cal
S},t,a)\subseteq G\}\}$ $=$ $\sup \{ \min\{n$ $|$ $n > 0 \land \pi^{(S)}(n)\in
G\}$ $|$ $\pi \in\{\pi \in {\rm Path}({\cal S}^{(\tilde{K})},t,a)$ $|$ $a$ is
s.t. $\varnothing \neq {\rm Img}({\cal S},t,a)\subseteq G\}\}$. Since for all
$\pi \in\{\pi \in {\rm Path}({\cal S}^{(\tilde{K})},t,a)$ $|$ $a$ is s.t.
$\varnothing \neq {\rm Img}({\cal S},t,a)\subseteq G\}$ we have that
$\pi^{(S)}(1) \in G$, we finally have that $J_{\rm strong}({\cal
S}^{(\tilde{K})},G,t) = \sup \{1\} = 1$. On the other hand, if $t \in F_n$ then
$J_{\rm strong}({\cal S}^{(\tilde{K})},G,t)=\sup \{ \min\{n$ $|$ $n > 0 \land
\pi^{(S)}(n)\in G\}$ $|$ $\pi \in\{\pi \in {\rm Path}({\cal
S}^{(\tilde{K})},t,a)$ $|$ $a$ is s.t. $\varnothing \neq {\rm Img}({\cal
S},t,a)\subseteq D_{n - 1}\}\} = \sup\{n_1, \ldots, n_j, \ldots\}$. We have
that, for all $j$, $n_j \leq n$. In fact, being $t \in F_n$ and $a$ s.t.
$\varnothing \neq {\rm Img}({\cal S},t,a)\subseteq D_{n - 1}$, we have that
$\pi^{(S)}(1) \in D_{n - 1}$ for all paths $\pi \in {\rm Path}({\cal
S}^{(\tilde{K})},t,a)$. This implies that $\pi^{(S)}(1) \in D_{n - 2} \lor
\pi^{(S)}(1) \in F_{n - 1}$. By property~\ref{proof.mgo.d_cup_f.item} above,
this implies that there exists $1 \leq i \leq n - 1$ s.t. $\pi^{(S)}(1) \in
F_i$. By iterating $n - 1$ times such a reasoning, we obtain that there exists
$1 \leq i \leq n$ s.t. $\pi^{(S)}(i) \in G$, which implies $n_j \leq n$ for all
$j$. Moreover, there exists a path $\pi \in {\rm Path}({\cal
S}^{(\tilde{K})},t,a)$ s.t. $\pi^{(S)}(n) \in G$ and for all $0 < i < n$ we have
that $\pi^{(S)}(i) \notin G$. Suppose by absurd that for all paths $\pi \in {\rm
Path}({\cal S}^{(\tilde{K})},t,a)$ we have that, if for all $0 < i < n$
$\pi^{(S)}(i) \notin G$, then $\pi^{(S)}(n) \notin G$. By using an iterative
reasoning as above, it is possible to show that this contradicts $t$ being in
$F_n$ and $a$ being s.t. $\varnothing \neq {\rm Img}({\cal S},t,a)\subseteq D_{n
- 1}$. Thus, being $n_j \leq n$ for all $j$ and existing a $j$ s.t. $n_j = n$,
we have that $J_{\rm strong}({\cal S}^{(\tilde{K})},G,t)=\sup\{n_1, \ldots, n_j,
\ldots\} = n$.


Note that also the converse holds, i.e. $J_{\rm strong}({\cal S}^{(\tilde{K})},
G, t)=n$ implies $t\in F_n$. This can be proved analogously to the reasoning
above.

To prove that $\tilde{K}$ is optimal, let us suppose that there exists another
solution $K$ and that there exists a nonempty set $Z$ of states,  such that for
all $z\in Z$,    $J_{\rm strong}({\cal S}^{(\tilde{K})}, G, z)>  J_{\rm
strong}({\cal S}^{(K)}, G, z)$. Let $z_0\in Z$ be a state for which   $J_{\rm
strong}({\cal S}^{(K)}, G, z_0)=n$ is minimal in $Z$, and let $a\in A$ be such
that $K(z_0,a)$.

We have that $n = 1$ implies that  ${\rm Img}({\cal S},z_0,a)\subseteq G$.  But
in such a case, $z_0$ would belong to $F_1$, and hence $J_{\rm strong}({\cal
S}^{(\tilde{K})}, G, z_0)=1=J_{\rm strong}({\cal S}^{(K)}, G, z_0)$. 

If $n>1$, for all $s\in {\rm Img}({\cal S},z_0,a)$, we have that $J_{\rm
strong}({\cal S}^{(K)}, G, s)\leq n-1$.  Since $n$ is the minimal distance for
which $J_{\rm strong}({\cal S}^{(\tilde{K})}, G, z)>  J_{\rm strong}({\cal
S}^{(K)}, G, z) = n$, we have that for all $s\in {\rm Img}({\cal S},z_0,a)$, $J_{\rm
strong}({\cal S}^{(\tilde{K})}, G, s)\leq  J_{\rm strong}({\cal S}^{(K)}, G, s)
\leq n - 1$.  This implies that, $J_{\rm strong}({\cal S}^{(\tilde{K})}, G,
z_0)\leq n$, which is absurd. 

\fussy

To prove that $\tilde{K}$ is the most general optimal solution,  we proceed in a
similar way.  Let us suppose that there exists another optimal solution $K$ and
that  there exists a nonempty set $Z$ of states,  such that for all $z\in Z$
there exists an action   $a$ s.t. $K(z,a)$ and $\neg \tilde{K}(z,a)$ holds. Let
$z_0\in Z$ be a state for which   $J_{\rm strong}({\cal S}^{(K)}, G, z_0)=n$ is
minimal in $Z$.

If $n=1$ we have that ${\rm Img}({\cal S},z_0,a)\subseteq G$ and thus  $z_0\in
F_1$ and $\tilde{K}(z_0, a)$, which leads to a contradiction.  

If $n>1$, by minimality of $J_{\rm strong}({\cal S}^{(K)}, G, z_0)$ in $Z$ we
have that, for all $s\in{\rm Img}({\cal S},z_0,a)$, $K(s,u)$ implies
$\tilde{K}(s,u)$. This implies that ${\rm Img}({\cal S},z_0,a)\in D_{n-1}$
and thus $\tilde{K}(z_0, a)$ holds.
\end{proof}


\section{Control Abstraction}
\label{ctr-abs.tex}

\sloppy

A quantization naturally induces an abstraction of a \gls{DTLHS}. 
Motivated by finding \gls{QFC} solutions in the abstract model, 
in this paper we introduce a novel notion of abstraction, 
namely \emph{control abstraction}.
In what follows we introduce the notion of control abstraction. 
In Sect.~\ref{sec:minMaxCtrAbs} we discuss on minimum and maximum control abstractions.
In Sect.~\ref{corol:min-ctr-abs.proof} we give some properties on control abstractions.

\fussy

Control abstraction (Def. \ref{def:ctr-abs}) models how a \gls{DTLHS} 
${\cal H}$ is \emph{seen} from the control software after \gls{AD}
conversions.  Since \gls{QFC} control rests on \gls{AD} conversion we must be
careful not to drive the plant outside the bounds in which \gls{AD}
conversion works correctly. This leads to the definition of
\emph{admissible action} (Def. \ref{def:safe-action}). 
Intuitively, an action is admissible in a state 
if it never drives the system outside of its 
admissible region.

\begin{definition}[Admissible actions]
\label{def:safe-action}
Let ${\cal H} = (X, U, Y, N)$ be a \gls{DTLHS} 
and ${\cal Q}=(A, \Gamma)$ be a quantization 
for ${\cal H}$.
An action $u \in A_U$ is $A$-\emph{admissible} in
$s \in A_X$ 
if for all $s'$, 
$(\exists y \in A_Y: N(s, u, y, s'))$
implies
$s' \in A_X$.
An action $\hat{u} \in \Gamma(A_U)$
is ${\cal Q}$-{\em admissible} in $\hat{s} \in \Gamma(A_X)$
if
for all 
$s \in \Gamma^{-1}(\hat{s})$, 
$u \in \Gamma^{-1}(\hat{u})$, $u$ is $A$-admissible for $s$ in ${\cal H}$.
%
\end{definition}


\begin{example}
\label{ex:safe}
%
%
Let ${\cal H}$ be as in  Ex.~\ref{ex:dths} and ${\cal Q}$ as in
Ex.~\ref{ex:quantization}. We have that action $u=1$ is not $A$-admissible in
the state $s=2$, thus $\hat{u}=1$ is not ${\cal Q}$-admissible in the state
$\hat{s}=1$. Analogously, $\hat{u}=1$ is not ${\cal Q}$-admissible in
$\hat{s}=-1$. It is easy to see that no other
$\hat{u}, \hat{s}$ exist s.t. $\hat{u}$ is not
${\cal Q}$-admissible in $\hat{s}$.
\end{example}


\begin{definition}[Control abstraction]
  \label{def:ctr-abs}
  Let ${\cal H} = (X, U, Y, N)$ be a \gls{DTLHS} and ${\cal Q}=(A, \Gamma)$ 
  be a quantization for ${\cal H}$. 
  We say that the \gls{LTS} $\hat{\cal H} =
  (\Gamma(A_X)$, $\Gamma(A_U)$, $\hat{N})$ 
  is a ${\cal Q}$ \emph{control abstraction} of ${\cal H}$ 
  if its transition relation $\hat{N}$ satisfies 
  the following conditions:

  \begin{enumerate}

  \item
    \label{item:witness}
    Each abstract transition stems from a concrete transition.
    Formally:
    for all 
    $\hat{s}, \hat{s}' \in \Gamma(A_X)$, $\hat{u} \in \Gamma(A_U)$,
    if $\hat{N}(\hat{s}, \hat{u}, \hat{s}')$
    then 
    there exist 
    $s \in \Gamma^{-1}(\hat{s})$,
    $u \in \Gamma^{-1}(\hat{u})$,
    $s' \in \Gamma^{-1}(\hat{s}')$, 
    $y \in A_Y$
    such that 
    $N(s, u, y, s')$.


%

  \item
    \label{item:nonloop}
    Each concrete transition 
    is faithfully represented by an abstract transition, 
    whenever it is not a self loop and its corresponding abstract action is ${\cal Q}$-admissible.
    Formally:
    for all $s,s'\in A_X$, $u\in A_U$ 
    such that $\exists y: N(s, u, y, s')$, 
    if $\Gamma(u)$ is ${\cal Q}$-admissible in $\Gamma(s)$ 
    and $\Gamma(s)\not=\Gamma(s')$ 
    then $\hat{N}(\Gamma(s), \Gamma(u), \Gamma(s'))$.
    %

  \item
    \label{item:loop}
    If there is no upper bound to the length of concrete paths inside the 
    counter-image of an abstract state then there is an abstract self loop.
    Formally:
    for all 
    $\hat{s} \in \Gamma(A_X)$, $\hat{u} \in \Gamma(A_U)$,
    if it exists an infinite run $\pi$ in ${\cal H}$ such that $\forall t\in\N$
    $\pi^{(S)}(t)\in\Gamma^{-1}(\hat{s})$ and $\pi^{(A)}(t)\in\Gamma^{-1}(\hat{u})$
    then
    $\hat{N}(\hat{s}, \hat{u}, \hat{s})$. A self loop $(\hat{s}, \hat{u}, \hat{s})$
    of $\hat{N}$ satisfying the above property is said to be a {\em non-eliminable self
    loop}, and {\em eliminable self
    loop} otherwise.
\end{enumerate}
%
\end{definition}

\begin{example}
\label{ex:ctr-abs}
%
%
\sloppy
Let  ${\cal H}$ be as in Ex.~\ref{ex:dths}  and ${\cal Q}$ be as in 
Ex.~\ref{ex:quantization}. Any ${\cal Q}$ control
abstraction $\hat{\cal H}$ of ${\cal H}$  has the form $(\{-1,0,1\},
\{0,1\},\hat{N})$ where  $\hat{N}$ always contains at
least all continuous arrows  in the automaton depicted in 
Fig.~\ref{strongWeakSolutions.fig} and some dotted arrows. Note that the only
non-eliminable self loops are 
$(0,0,0)$ and  $(0,1,0)$.
\fussy
\end{example}

\sloppy

%

Along the same lines of the proof for Theor.~\ref{th:undecidability},
in \cite{ictac2012} we proved that we cannot algorithmically decide 
if a self loop is eliminable or non-eliminable.

\fussy

\begin{proposition}\label{self.loop.undec.prop}
Given a \gls{DTLHS} ${\cal H}$ and a quantization ${\cal Q}$, it is undecidable to determine if a self loop is non-eliminable.
\end{proposition}

Note that if in Def.~\ref{def:ctr-abs} we drop condition~\ref{item:loop} and the guard $\Gamma(s) \not= \Gamma(s')$ in
condition~\ref{item:nonloop}, then we essentially get the usual definition of
\emph{abstraction} (e.g. see \cite{pred-abs-tecs06} and
citations thereof).  
As a result, any abstraction is also
a control abstraction whereas a control abstraction 
in general is not an abstraction since some self loops 
or some non admissible actions may be missing.

In the following, we will deal with two types of control abstractions, namely
{\em full} and {\em admissible} control abstractions, which are defined as
follows.

\begin{definition}[Admissible and full control abstractions]
\label{def:ctr-abs-full-adm}
Let ${\cal H} = (X, U, Y, N)$ be a \gls{DTLHS} and ${\cal Q}=(A, \Gamma)$  be a
quantization for ${\cal H}$.  A ${\cal Q}$ control abstraction $\hat{\cal H} =
  (\Gamma(A_X),$ $\Gamma(A_U),$ $\hat{N})$ of
${\cal H}$ is an \emph{admissible ${\cal Q}$ control abstraction} iff, for all
$\hat{s} \in \Gamma(A_X),\hat{u} \in \Gamma(A_U)$ s.t. $\hat{u} \in {\rm Adm}(\hat{\cal H}, \hat{s})$:
i) $\hat{u}$ is ${\cal Q}$-admissible in $\hat{s}$;
ii) $\forall s \in
\Gamma^{-1}(\hat{s})$ $\forall u \in \Gamma^{-1}(\hat{u})$ $\exists s' \in {\cal
D}_X$ $\exists y \in {\cal D}_Y$ $:$ $N(s, u, y, s')$, i.e. each concrete state in
$\Gamma^{-1}(\hat{s})$ has a successor for all  concrete actions in
$\Gamma^{-1}(\hat{u})$. 

%
%
%

We say that $\hat{\cal H}$ is a {\em full ${\cal Q}$ control
abstraction} if it 
satisfies properties~\ref{item:witness} and~\ref{item:loop} of Def.~\ref{def:ctr-abs}, plus the
following property (derived from property~\ref{item:nonloop}  of Def.~\ref{def:ctr-abs}): for all
$s,s'\in A_X$, $u\in A_U$  such that $\exists y: N(s, u, y, s')$,  if
$\Gamma(s)\not=\Gamma(s')$  then $\hat{N}(\Gamma(s),$ $\Gamma(u),$ $\Gamma(s'))$. 
%
%
\end{definition}


\begin{example}
\label{ex:ctr-abs-full-adm}
\sloppy
Let  ${\cal H}$ be as in Ex.~\ref{ex:dths}, ${\cal Q}$ be as in 
Ex.~\ref{ex:quantization}.
For all ${\cal Q}$ admissible control
abstractions of ${\cal H}$,  $\hat{N}(1,1,1) = \hat{N}(-1,1,-1) = 0$, since action $1$
is not ${\cal Q}$-admissible either in $-1$ or in $1$ (see Ex.~\ref{ex:safe}).
On the contrary, for all full ${\cal Q}$ control abstractions of ${\cal H}$,  $\hat{N}(1,1,1)
= \hat{N}(-1,1,-1) = 1$. Thus, a control abstraction s.t.
$\hat{N}(1,1,1) \oplus \hat{N}(-1,1,-1)$ (where $\oplus$ is the logical XOR) is
neither full nor admissible. 
\fussy
\end{example}


By  the definition of quantization, a control abstraction is a finite \gls{LTS}. It is
possible to show that two different admissible [full] ${\cal Q}$ control
abstractions only differ in the number of self loops. Moreover, the set of
admissible [full] ${\cal Q}$ control abstraction is a finite lattice with
respect to the \gls{LTS} refinement relation (Sect.~\ref{proof:poset.lattice}).
This implies that such lattices have minimum (and maximum).   Thus, it is easy to
prove that the minimum admissible [full] ${\cal Q}$ control abstraction is the
admissible [full] ${\cal Q}$ control abstraction with non-eliminable self loops
only. Thus, the following proposition is a 
corollary of
Prop.~\ref{self.loop.undec.prop}.

\begin{proposition}\label{minctrabs.undec.prop}
%
%
Given a \gls{DTLHS} ${\cal H}$ and a quantization ${\cal Q}$, it is undecidable to
state if an admissible [full] ${\cal Q}$ control abstraction for ${\cal H}$ is the minimum
admissible [full] ${\cal Q}$ control abstraction for ${\cal H}$.
\end{proposition}

\subsection{Maximum and Minimum Control Abstractions}
\label{subsec:maxMinCtrAbs}
\label{sec:minMaxCtrAbs}

\begin{figure}
  \begin{center}
    \scalebox{0.25}{\input{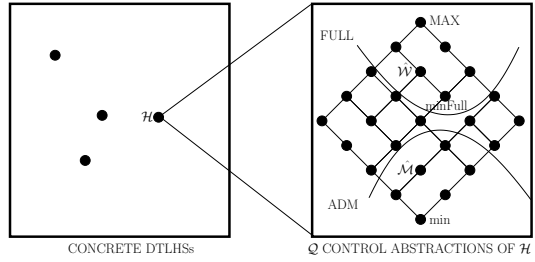}}
    \caption{Lattices on ${\cal Q}$ control abstractions.} 
    \label{lattice.eps}
  \end{center}
\end{figure}

By Theor.~\ref{th:undecidability},
we cannot hope for a constructive  sufficient
and necessary condition for the existence of a ${\cal Q}$ \gls{QFC} solution to a \gls{DTLHS} control problem, for a
given ${\cal Q}$. Accordingly, our approach is able to determine (via a
sufficient condition) if  a ${\cal Q}$ \gls{QFC} solution exists, and otherwise to
state (via a necessary condition) if  a ${\cal Q}$ \gls{QFC} solution cannot exist.
If both conditions are false, then our approach is not able to decide if  a ${\cal Q}$ \gls{QFC} solution
exists or not.
We base our sufficient [necessary] condition on computing a {\em (close to) minimum} admissible [full] ${\cal Q}$
control abstraction. Theor.~\ref{corol:min-ctr-abs} gives the foundations for
such an approach. The proof of Theor.~\ref{corol:min-ctr-abs} follows from the definitions
of admissible and full control abstractions and properties of strong and weak
solutions (Sect.~\ref{proof:structure}).
%
In the following theorem we use the {\em refinement} order relation
(denoted by $\sqsubseteq$) defined in Sect.~\ref{lts.tex}.

\begin{theorem}
  \label{corol:min-ctr-abs}\label{COROL:MIN-CTR-ABS}
  Let $\cal H$ be a \gls{DTLHS}, ${\cal Q}=(A, \Gamma)$ be a quantization for ${\cal H}$, 
  and $({\cal H}$, $I$, $G)$ be a control problem.

  \begin{enumerate}
  \item
    \label{item:minimal}
    If $\hat{\cal H}$ is an admissible ${\cal Q}$ control abstraction and  
    $\hat{K}$ is a strong solution to 
    $(\hat{\cal H},$ $\Gamma(I),$ $\Gamma(G))$ then, for any control law $k$ for $\hat{K}$,
    $K(x, u)$ $=$
    $(k(\Gamma(x)) = \Gamma(u))$ is a ${\cal Q}$ \gls{QFC} strong solution to
    $({\cal H}, I, G)$.

%
  \item \label{item:whyminimal} If $\hat{\cal H}_1 \sqsubseteq \hat{\cal H}_2$ are two admissible $\cal Q$
    control abstractions of ${\cal H}$ and $\hat{K}$ is a strong solution to 
    $(\hat{\cal H}_2, \Gamma(I), \Gamma(G))$, then $\hat{K}$ is a strong
    solution to 
    $(\hat{\cal H}_1, \Gamma(I),
    \Gamma(G))$.

  \item
    \label{item:maximal}
    If $\hat{\cal H}$ is a full ${\cal Q}$ control abstraction and  
    $(\hat{\cal H}$, $\Gamma(I)$,
    $\Gamma(G))$ does not have a weak solution then there exists no
    ${\cal Q}$ \gls{QFC} (weak as well as strong) 
    solution to 
    $({\cal H}, I, G)$. 

%
  \item \label{item:whymaximal} If $\hat{\cal H}_1 \sqsubseteq \hat{\cal H}_2$ are two full $\cal Q$
    control abstractions of ${\cal H}$ and $\hat{K}$ is a weak solution to 
    $(\hat{\cal H}_1, \Gamma(I), \Gamma(G))$, then $\hat{K}$ is a weak
    solution to 
    $(\hat{\cal H}_2, \Gamma(I),
    \Gamma(G))$.
  \end{enumerate}
\end{theorem}


Fig.~\ref{lattice.eps} graphically represents a sketch of the correspondence
between a concrete \gls{DTLHS} ${\cal H}$ and its control abstractions $\hat{\cal H}$
lattices. 
\begin{example}
\label{ex:minmax-ctr-abs}
Let ${\cal P}$ = $({\cal H},$ $I,$ $G)$ be as in Ex. \ref{ex:ctr-dths}
and ${\cal Q}\;=\;(A, \Gamma)$
be as in Ex. \ref{ex:quantization}.
For all admissible ${\cal Q}$ control abstractions $\hat{\cal H}$ (see Ex.~\ref{ex:ctr-abs-full-adm})
not containing the eliminable self loops $(-1, 0, -1)$ and $(1, 0, 1)$,
$\hat{K}(\hat{x},\hat{u})\equiv [\hat{x}\not=0 \to \hat{u}=0]$ (see Ex.~\ref{ex:gamma-qfc-sol}) 
is the strong mgo for $(\hat{\cal H},$ $\Gamma(I),$ $\Gamma(G))$.
Thus,
$K(x, u)=\hat{K}(\Gamma(x),\Gamma(u))$ 
is a ${\cal Q}$ \gls{QFC} solution to ${\cal P}$.
Let us consider the quantization ${\cal Q}'=(A, \Gamma')$, where 
$\Gamma'(w)$=$\left\lfloor w/2\right\rfloor$. 
A full ${\cal Q}'$ control abstraction of ${\cal H}$ is
${\cal L}$ = $(\{-2,-1,0,1\},$ $\{0,1\},$ $\hat{N})$,
where the transition $\hat{N}$ is depicted in Fig.~\ref{dtlhs.nosol.fig}.
(${\cal L}$, $\Gamma'(I)$, $\Gamma'(G)$)
has no weak solution, thus ${\cal P}$ has no ${\cal Q}'$ \gls{QFC} solution.
\end{example}

 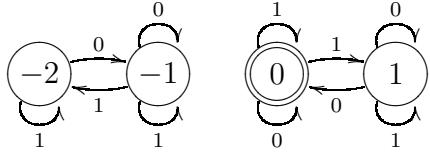
\begin{figure}
 \begin{center}
 $
 \xymatrix@C=7mm{
	&
       *+=<25pt>[o][F-]{-2} 
			\ar@(dl,dr)@[]_{1}
			\ar@/^/[r]^{0}
      &
	  *+=<25pt>[o][F-]{-1} 
			\ar@(dl,dr)@[]_{1}
			\ar@(ul,ur)@[]^{0}
			\ar@/^/[l]^{1}
        & *+=<25pt>[o][F=]{0} 
			\ar@/^/[r]^{1}
			\ar@(dl,dr)[]_{0}
			\ar@(ul,ur)[]^{1}
	& *+=<25pt>[o][F-]{1} 
			\ar@(dl,dr)@[]_{1}
			\ar@(ul,ur)@[]^{0}
 			\ar@/^/[l]^{0}
			&\\
	} 
 $
 \caption{${\cal Q}$ control abstraction without weak solutions.} 
 \label{dtlhs.eps.gt0.fig}
 \label{dtlhs.nosol.fig}
 \end{center}
 \end{figure}


	\subsection{Proof of Control Abstraction Properties}
\label{app:qfc-solution}
\label{proof:structure}
\label{proof:poset.lattice}
\label{corol:min-ctr-abs.proof}

\sloppy

In this section we give proofs about control abstraction properties.
This section can be skipped at a first reading.
In the following, we denote with $\ctrset$ the set of all ${\cal Q}$ control
abstractions of a \gls{DTLHS} ${\cal H}$.

\begin{fact}
  \label{lemma:structure}\label{LEMMA:STRUCTURE}
  Let ${\cal M}_1 = (S, B, T_1)$ and 
  ${\cal M}_2 = (S, B, T_2)$ be two admissible ${\cal Q}$ control
  abstractions of a \gls{DTLHS} ${\cal H}$, with ${\cal Q} = (A, \Gamma)$ quantization for ${\cal H}$. Then $ \forall \hat{x}, \hat{x}' \in S$
  s. t.  $\hat{x} \neq \hat{x}'$, $\forall \hat{a} \in B
  \, [
  T_1(\hat{x}, \hat{a}, \hat{x}')$ 
  $\Leftrightarrow$ $T_2(\hat{x}, \hat{a},
  \hat{x}')] $. The same holds if ${\cal M}_1, {\cal M}_2$ are full ${\cal Q}$
  control abstractions.
\end{fact}
\begin{proof}
Let $\hat{x} \neq \hat{x}' \in S, \hat{a} \in B$ be  such that $T_1(\hat{x},
\hat{a}, \hat{x}')$ holds. If ${\cal M}_1$ is an admissible ${\cal Q}$ control
  abstraction, this implies, by 
Def.~\ref{def:ctr-abs-full-adm}, that $\hat{a}$ is $A$-admissible in $\hat{x}$. From point~\ref{item:witness} of
Def.~\ref{def:ctr-abs} (for the admissible control abstraction case) or
Def.~\ref{def:ctr-abs-full-adm} of full control abstraction (for the full
control abstraction case), and from $T_1(\hat{x}, \hat{a}, \hat{x}')$ follows that $
\exists x \in \Gamma^{-1}(\hat{x}) \exists x' \in \Gamma^{-1}(\hat{x}'):
\exists a \in \Gamma^{-1}(\hat{a}) \exists y: N(x, a, y, x')$. By
point~\ref{item:nonloop} of Def.~\ref{def:ctr-abs} this implies that
$T_2(\hat{x}, \hat{a}, \hat{x}')$ holds.

The same reasoning may be applied to prove the other implication.
\end{proof}
%
%

\begin{fact}
  \label{prop:poset.lattice}\label{PROP:REFINEMENT.PO}
  Given a \gls{DTLHS} ${\cal H}$ and a quantization ${\cal Q}$,
  the set $(\ctrset, \sqsubseteq)$ of ${\cal Q}$ control abstractions of ${\cal H}$
  is a {\em lattice}. Moreover, the set of full ${\cal Q}$ control abstractions of ${\cal H}$
  is a {\em lattice}.
\end{fact}
\begin{proof}
By conditions~\ref{item:nonloop} and~\ref{item:loop} of Def.~\ref{def:ctr-abs} 
all control abstractions do contain all admissible actions 
that have a concrete witness and all non-eliminable self-loops.

As a consequence, if $S$ is the set of eliminable self-loops 
and $U$ is the set of non admissible actions, 
then $(\ctrset, \sqsubseteq)$ is isomorphic to the complete 
lattice $(2^{S\times U},\subseteq)$.

Analogously, the set of full ${\cal Q}$ control abstractions of ${\cal H}$ is isomorphic to the complete 
lattice $(2^{S},\subseteq)$.
%
\end{proof}

\begin{proof}[Theorem~\ref{corol:min-ctr-abs}]
The idea underlying the proof 
is that two
different admissible (as well as full) control abstractions, with the same quantization, 
have the same loop free structure, i.e. the same arcs except from self loops, 
as proved by Prop.~\ref{lemma:structure}.
For ease of notation, given a state $x$ (resp. an action $u$) 
we will often denote the corresponding abstract state $\Gamma(x)$ 
(resp. action $\Gamma(u)$) with $\hat{x}$ (resp. $\hat{u}$). Analogously, we
will often write $\hat{I}$ (resp. $\hat{G}$) for $\Gamma(I)$ (resp. $\Gamma(G)$).
In the following, ${\cal P} = ({\cal H}, I, G)$, $\hat{\cal P} = (\hat{\cal H},
\Gamma(I), \Gamma(G))$, and $\hat{\cal H} = (\Gamma(A_X), \Gamma(A_U), \hat{N})$.


\paragraph*{Proof of point~\ref{item:minimal}}
Applying the definition of solution to a \gls{DTLHS} control problem
(Def.~\ref{def:epsilon-solution}),  we have to show that if $\hat{K}$ is a
strong solution to the \gls{LTS} control problem  $(\hat{\cal H}, \hat{I},
\hat{G})$, then $K$ defined by  $K(x, u)$ = $(k(\hat{x}) = \hat{u})$  is a
strong solution to the \gls{LTS} control problem  (LTS(${\cal H}$), $I$, ${\cal
B}_{\|\Gamma\|}(G))$, being $k$ a control law for $\hat{K}$.

 
Note that, since $\hat{\cal H}$ is an admissible control abstraction, it
contains admissible actions only. This implies that all actions enabled by
$\hat{K}$ in $\hat{x}$ are ${\cal Q}$-admissible in $\hat{x}$. Hence, we have
that all actions enabled by $K$ in $x$ are $A$-admissible in $x$. Together
with point~\ref{item:nonloop} of Def.~\ref{def:ctr-abs}, this implies that, for
any transition $(x, u, x')$ of ${\rm LTS}({\cal H})^{(K)}$ such that
$\hat{x}\not=\hat{x}'$, $(\hat{x}, \hat{u}, \hat{x}')$ is a (abstract)
transition of $\hat{\cal H}^{(\hat{K})}$.

First of all, we prove that $I \subseteq \mbox{\rm Dom}(K)$.
Given a state $x\in I$, we have that $\hat{x}\in\hat{I}$.
Since $\hat{K}$ is a strong solution to $\hat{\cal P}$, 
we have that $\hat{I} \subseteq \mbox{\rm Dom}(\hat{K})$, thus $\hat{x}\in{\rm Dom}(\hat{K})$. 
Hence, there exists $\hat{u}\in\Gamma(A_U)$
such that $\hat{K}(\hat{x}, \hat{u})$ holds, which implies that $k(\hat{x})$ is
defined. 
By definition of $K$, we have that for all $u\in\Gamma^{-1}(k(\hat{x}))$
and for all $x\in\Gamma^{-1}(\hat{x})$ $K(x,u)$ holds, 
which means that $x\in{\rm Dom}(K)$.


Now, we prove that for all $x \in {\rm Dom}(K)$, 
$J_{\rm strong}(\mbox{\rm LTS}({\cal H})^{(K)}, {\cal B}_{\|\Gamma\|}(G), x)$ is finite.
Let us suppose by absurd that $J_{\rm strong}(\mbox{\rm LTS}({\cal H})^{(K)},
{\cal B}_{\|\Gamma\|}(G), x) = \infty$. This implies that one of the two
following holds: 

\begin{enumerate} 

	\item \label{finite.fullpath.item} there exists a finite fullpath $\pi =
x_0 u_0 x_1 u_1 \ldots x_n u_n$ in ${\rm LTS}({\cal H})^{(K)}$ such that $x_0 =
x$, ${\rm Adm}({\rm LTS}({\cal H})^{(K)}, x_n) = \varnothing$ and, for all
$i\in[n]$, $x_i\not\in {\cal B}_{\|\Gamma\|}(G)$;

	\item \label{infinite.fullpath.item} there exists an infinite fullpath
$\pi = x_0 u_0 x_1 u_1 \ldots x_n u_n \ldots$ in ${\rm LTS}({\cal H})^{(K)}$
such that $x_0 = x$ and, for all $i\in\N$, $x_i\not\in {\cal
B}_{\|\Gamma\|}(G)$. 

\end{enumerate} 


Let us deal with the finite fullpath case first
(point~\ref{finite.fullpath.item} above). Let $\hat{\pi} = \hat{x}_0
\hat{u}_0\ldots \hat{u}_{n-1} \hat{x}_n$, and let $\rho$ be defined from
$\hat{\pi}$ by collapsing all consecutive equal (abstract) states into one
(abstract) state. Formally, $|\rho| = \max_{i \in [n]} \alpha(i)$ and $\rho(i) =
\hat{\pi}^{(S)}(\alpha(i)) = \Gamma(\pi^{(S)}(\alpha(i)))$, where the function 
$\alpha:\N\rightarrow\N$ is recursively defined as follows:



\begin{itemize}  

\item let $Z_z = \{j \;|\; z < j \leq n \land \Gamma(x_j)\not=\Gamma(x_{z})\}$

\item $\alpha(0) = 0$

\item $\alpha(i + 1) = \left\{
\begin{array}{cc}
\alpha(i) & \mbox{if } Z_{\alpha(i)} = \varnothing\\
\min Z_{\alpha(i)} & \mbox{otherwise}
\end{array}
\right.$

\end{itemize}


By the fact (proved above) that if $(x, u, x')$ is a transition of ${\rm
LTS}({\cal H})^{(K)}$ with $\hat{x}\not=\hat{x}'$, then  $(\hat{x}, \hat{u},
\hat{x}')$ is a transition of $\hat{\cal H}^{(\hat{K})}$, we have that  $\rho$
is a run of $\hat{\cal H}^{(\hat{K})}$. Let $m = |\rho| = \max_{i \in [n]}
\alpha(i)$. Since $\hat{K}$ is a strong solution to $\hat{P}$ and $\hat{x} \in {\rm
Dom}(\hat{K})$, we have that $\hat{x}_m\in {\rm Dom}(\hat{K})$. This implies
that there exists $\hat{u} \in \Gamma(A_U)$ s.t. $\hat{K}(\hat{x}_m, \hat{u})$
and $k(\hat{x}_m) = \hat{u}$ ,
thus that there exists $\hat{u} \in {\rm Adm}(\hat{\cal H}^{(\hat{K})},
\hat{x}_m)$. Thus by 
Def.~\ref{def:ctr-abs-full-adm} (and since $x_n \in \Gamma^{-1}(\hat{x}_m)$) we
have that ${\rm Adm}({\rm LTS}({\cal H})^{(K)}, x_n) \supseteq
\Gamma^{-1}(\hat{u}) \neq \varnothing$, which implies that $\pi$ cannot be a
finite fullpath.

As for the infinite fullpath case (point~\ref{infinite.fullpath.item} above), we
observe that in $\pi$ we cannot have an infinite sequence $x_m u_m x_{m+1}
u_{m+1} \ldots$ such that for all $j\geq m$,  $\Gamma(x_j)=\Gamma(x_m)$ and
$\Gamma(u_j)=\Gamma(u_m)$.  In fact, suppose by absurd that this is true, and
let $\tilde{m}$ be the least $m$ for which this happens. Then $(\hat{x}_m,
\hat{u}_m, \hat{x}_m)$ is a non-eliminable self loop. Since $x_j \notin {\cal
B}_{\|\Gamma\|}(G)$ for all $j\geq m$, and thus $\hat{x}_j \notin \hat{G}$ for
all $j\geq m$, we also have that $J_{\rm strong}(\hat{\cal H}^{(\hat{K})},
\hat{G}, \hat{x}_m)=\infty$. By applying the same reasoning used for the finite
fullpath case, we have that there is a path in $\hat{\cal H}^{(\hat{K})}$
leading from $\hat{x}$ to $\hat{x}_m$, which implies that $J_{\rm
strong}(\hat{\cal H}^{(\hat{K})}, \hat{G}, \hat{x})=\infty$. Finally, this
contradicts the fact that $\hat{K}$ is a strong solution to ${\hat P}$ and
$\hat{x} \in {\rm Dom}(\hat{K})$.
Since the control law $k$ for $\hat{K}$ (and thus $K$, which is defined on $k$) 
only enables one action $\hat{u}$ for each abstract state,
we may conclude that we cannot have an infinite sequence $x_m u_m x_{m+1}
u_{m+1} \ldots$ such that for all $j\geq m$,  $\Gamma(x_j)=\Gamma(x_m)$.  


Thanks to this fact, from a given infinite fullpath $\pi = x_0 u_0 x_1 u_1
\ldots x_n u_n \ldots$ of ${\rm LTS}({\cal H})^{(\alpha)}$ with $x_0 = x$, we can 
extract an infinite abstract fullpath $\rho$ s.t. $\rho(i) =
\Gamma(\pi^{(S)}(\alpha(i)))$, where the function  $\alpha:\N\rightarrow\N$ is recursively
defined as follows:

\begin{itemize}  

\item $\alpha(0) = 0$

\item $\alpha(i+1) = \min\{j \;|\; \alpha(i) < j \land
\Gamma(x_j)\not=\Gamma(x_{\alpha(i)})\}$. 

\end{itemize}


By the fact (proved above) that if $(x, u, x')$ is a transition of ${\rm
LTS}({\cal H})^{(K)}$ with $\hat{x}\not=\hat{x}'$, then  $(\hat{x}, \hat{u},
\hat{x}')$ is a transition of $\hat{\cal H}^{(\hat{K})}$, we have that  $\rho$
is a run of $\hat{\cal H}^{(\hat{K})}$. Moreover, since for all
$i\in\N\;x_i\not\in {\cal B}_{\|\Gamma\|}(G)$, then we have that for all
$i\in\N$  $\hat{x}_{\alpha(i)}\not\in\hat{G}$. This contradicts the fact that 
$\hat{K}$ is a strong solution to $\hat{P}$ and $\hat{x} \in {\rm
Dom}(\hat{K})$. 


\paragraph*{Proof of point~\ref{item:whyminimal}}

Let $\hat{\cal H}_1 = (\Gamma(A_X),$ $\Gamma(A_U),$ $T_1)$ and
$\hat{\cal H}_2 = (\Gamma(A_X),$ $\Gamma(A_U),$ $T_2)$ be two
admissible ${\cal Q}$ control abstractions of ${\cal H}$, with $\hat{\cal H}_1
\sqsubseteq\hat{\cal H}_2$. If $\hat{\cal H}_1 =\hat{\cal H}_2$ the thesis is
proved, thus let us suppose that $\hat{\cal H}_1 \neq\hat{\cal H}_2$. By
Fact~\ref{lemma:structure}, the only difference between $\hat{\cal H}_1$ and
$\hat{\cal H}_2$ may be in a finite number of (eliminable) self loops which are
in $\hat{\cal H}_2$ only. That is, there exists a transitions set $B =
\{(\hat{x}_1, \hat{u}_1, \hat{x}_1), \ldots, (\hat{x}_m, \hat{u}_m,
\hat{x}_m)\}$ s.t. for all $(\hat{x}_i, \hat{u}_i, \hat{x}_i) \in B$ we have
that $T_1(\hat{x}_i, \hat{u}_i, \hat{x}_i) = 0 \land T_2(\hat{x}_i, \hat{u}_i,
\hat{x}_i) = 1$, and for all $(\hat{x}, \hat{u}, \hat{x}') \in \Gamma(A_X)
\times \Gamma(A_U) \times \Gamma(A_X)$ we have that if $(\hat{x}, \hat{u},
\hat{x}') \notin B$ then $T_1(\hat{x}, \hat{u}, \hat{x}) = T_2(\hat{x},
\hat{u}, \hat{x})$. Let $\hat{K}$ be the strong mgo to the \gls{LTS} control
problem $(\hat{\cal H}_2, \hat{I}, \hat{G})$ and let $(\hat{x}_i, \hat{u}_i,
\hat{x}_i)\in B$. 


Note that if $\hat{x}_i\notin\hat{G}$ and $\hat{K}(\hat{x}_i,\hat{u}_i)$  then
$J_{\rm strong}(\hat{\cal H}_2^{(\hat{K})},\hat{G},\hat{x}_i)=\infty$ since
there exists a $\pi \in {\rm Path}(\hat{\cal H}_2^{(\hat{K})}, \hat{x}_i,
\hat{u}_i)$ s.t. $\pi^{(S)}(t) = \hat{x}_i$ and $\pi^{(A)}(t) = \hat{u}_i$ for
all $t \in \N$. As a consequence, if $\hat{x}_i\notin \hat{G}$  then
$\hat{K}(\hat{x}_i,\hat{u}_i)$ does not hold. Moreover, suppose that
$\hat{x}_i\in\hat{G}$. Since $(\hat{x}_i, \hat{u}_i, \hat{x}_i)$ is an
eliminable self loop of $\hat{\cal H}_2$ and $\hat{\cal H}_2$ is an
admissible ${\cal Q}$ control abstraction, there exists a state
$\hat{x}'\not=\hat{x}_i$ such that $T_2(\hat{x}_i,\hat{u}_i,\hat{x}')$.  

We are now ready to prove the thesis. Since we already know that $\hat{I}
\subseteq {\rm Dom}(\hat{K})$, we only have to prove that i) $\hat{K}$ is a
controller for $\hat{\cal H}_1$ and that ii) $J_{\rm strong}(\hat{\cal
H}_1^{(\hat{K})},\hat{G},\hat{x}) < \infty$ for all $\hat{x} \in {\rm
Dom}(\hat{K})$. 

As for the first point, we have to show that $\hat{K}(\hat{x}, \hat{u})$ implies
$\hat{u} \in {\rm Adm}(\hat{\cal H}_1, \hat{x})$
(Def.~\ref{def:ctrproblem-lts}). Suppose by absurd that $\hat{u} \notin {\rm
Adm}(\hat{\cal H}_1, \hat{x})$ for some $\hat{x}, \hat{u}$.  Since
$\hat{K}(\hat{x}, \hat{u})$ implies $\hat{u} \in {\rm Adm}(\hat{\cal H}_2,
\hat{x})$, we have that $(\hat{x}, \hat{u}, \hat{x}) \in B$. If
$\hat{x}\notin\hat{G}$ then $\hat{K}(\hat{x}, \hat{u}) = 0$, which is false by
hypothesis. If $\hat{x}\in\hat{G}$, then there exists a state
$\hat{x}'\not=\hat{x}$ such that $T_2(\hat{x},\hat{u},\hat{x}')$.  Thus,
$T_1(\hat{x},\hat{u},\hat{x}')$ holds by Fact~\ref{lemma:structure} and we have
$\hat{u} \in {\rm Adm}(\hat{\cal H}_1, \hat{x})$, which is absurd.

As for the second one, it is sufficient to prove that $J_{\rm strong}(\hat{\cal
H}_1^{(\hat{K})},\hat{G},\hat{x}) = J_{\rm strong}(\hat{\cal
H}_2^{(\hat{K})},\hat{G},\hat{x})$. This can be proved by induction on the value
of $J_{\rm strong}(\hat{\cal H}_2^{(\hat{K})},\hat{G},\hat{x})$. 

Suppose $J_{\rm strong}(\hat{\cal H}_2^{(\hat{K})},\hat{G},\hat{x}) = 1$. Then,
$\varnothing \neq {\rm Img}(\hat{\cal H}_2^{(\hat{K})},\hat{x},\hat{u})
\subseteq \hat{G}$ for all $\hat{u}$ s.t. $\hat{K}(\hat{x}, \hat{u})$. If for
all $\hat{u}$ s.t. $\hat{K}(\hat{x}, \hat{u})$ there exists a state $\hat{x}'
\neq \hat{x}$ s.t. $\hat{x}' \in {\rm Img}(\hat{\cal
H}_2^{(\hat{K})},\hat{x},\hat{u})$, then we have that $\hat{x}' \in {\rm
Img}(\hat{\cal H}_1^{(\hat{K})},\hat{x},\hat{u})$ by Fact~\ref{lemma:structure},
and since $\varnothing \neq {\rm Img}(\hat{\cal
H}_1^{(\hat{K})},\hat{x},\hat{u}) \subseteq {\rm Img}(\hat{\cal
H}_2^{(\hat{K})},\hat{x},\hat{u}) \subseteq \hat{G}$ we have that $J_{\rm
strong}(\hat{\cal H}_1^{(\hat{K})},\hat{G},\hat{x}) = 1 = J_{\rm
strong}(\hat{\cal H}_2^{(\hat{K})},\hat{G},\hat{x})$. Otherwise, let $\hat{u}$
be s.t. $\hat{K}(\hat{x}, \hat{u})$ and $T_2(\hat{x}, \hat{u}, \hat{x}') \to
\hat{x}' = \hat{x}$. Note that this implies $\hat{x}\in\hat{G}$. If $(\hat{x},
\hat{u}, \hat{x}) \notin B$, then $T_1(\hat{x}, \hat{u}, \hat{x})$ thus $J_{\rm
strong}(\hat{\cal H}_1^{(\hat{K})},\hat{G},\hat{x}) = 1 = J_{\rm
strong}(\hat{\cal H}_2^{(\hat{K})},\hat{G},\hat{x})$. The other case, i.e.
$(\hat{x}, \hat{u}, \hat{x}) \in B$, is impossible since, by the reasoning above and
being $\hat{x}\in\hat{G}$, it would imply that
there exists a state $\hat{x}'\not=\hat{x}$ such that
$T_2(\hat{x},\hat{u},\hat{x}')$.


Suppose now that for all $\hat{x}$ s.t. $J_{\rm strong}(\hat{\cal
H}_2^{(\hat{K})},\hat{G},\hat{x}) = n$, $J_{\rm strong}(\hat{\cal
H}_1^{(\hat{K})},\hat{G},\hat{x}) = J_{\rm strong}(\hat{\cal
H}_2^{(\hat{K})},\hat{G},\hat{x})$. Let $\hat{x} \in {\rm Dom}(\hat{K})$ be s.t.
$J_{\rm strong}(\hat{\cal H}_2^{(\hat{K})},\hat{G},\hat{x}) = n + 1$. If
$(\hat{x}, \hat{u}, \hat{x}) \notin B$ for any $\hat{u}$, then ${\rm
Img}(\hat{\cal H}_2^{(\hat{K})}, \hat{x}, \hat{u}) = {\rm Img}(\hat{\cal
H}_1^{(\hat{K})}, \hat{x}, \hat{u})$ for all $\hat{u}$, thus $J_{\rm
strong}(\hat{\cal H}_1^{(\hat{K})},\hat{G},\hat{x}) = J_{\rm strong}(\hat{\cal
H}_2^{(\hat{K})},\hat{G},\hat{x})$ by induction hypothesis. Otherwise, let
$(\hat{x}, \hat{u}, \hat{x}) \in B$ for some $\hat{u}$. By the reasoning above,
if $\hat{x}\notin\hat{G}$ then $\hat{K}(\hat{x},\hat{u}) = 0$, and again $J_{\rm
strong}(\hat{\cal H}_1^{(\hat{K})},\hat{G},\hat{x}) = J_{\rm strong}(\hat{\cal
H}_2^{(\hat{K})},\hat{G},\hat{x})$ by induction hypothesis. If
$\hat{x}\in\hat{G}$, then there exists a state $\hat{x}'\not=\hat{x}$ such that
$T_2(\hat{x},\hat{u},\hat{x}')$ (and $T_1(\hat{x},\hat{u},\hat{x}')$).  Since
$J_{\rm strong}(\hat{\cal H}_2^{(\hat{K})},\hat{G},\hat{x}) = n + 1$, we must
have $J_{\rm strong}(\hat{\cal H}_2^{(\hat{K})},\hat{G},\hat{x}') \leq n$,
thus again $J_{\rm strong}(\hat{\cal H}_1^{(\hat{K})},\hat{G},\hat{x}) = J_{\rm
strong}(\hat{\cal H}_2^{(\hat{K})},\hat{G},\hat{x})$ by inductive hypothesis. 

Finally, note that in general $\hat{K}$ is not optimal for $({\cal H}_1, \hat{I},
\hat{G})$. As a counterexample, consider the control abstractions $\hat{\cal
H}_2 = (\{0,$ $1,$ $2\},$ $\{0,$ $1\},$ $\{(0,$ $0,$ $2),$ $(0,$ $0,$ $0),$
$(0,$ $1,$ $1),$ $(1,$ $1,$ $2),$ $(2, 0,$ $2)\})$ and $\hat{\cal H}_1 = (\{0,$
$1,$ $2\},$ $\{0,$ $1\},$ $\{(0,$ $0,$ $2),$ $(0,$ $1,$ $1), (1,$ $1,$ $2),$
$(2,$ $0,$ $2)\})$, with $\hat{I} = \{0, 1, 2\}$ and $\hat{G} = \{2\}$. We have
that the strong mgo for $\hat{\cal H}_2$ is $\hat{K}_2 = \{(0, 1), (1, 1), (2,
0)\}$, whilst the strong mgo for $\hat{\cal H}_1$ is $\hat{K}_1 = \{(0, 0), (1,
1), (2, 0)\}$, with $J_{\rm strong}(\hat{\cal H}_1^{(\hat{K}_1)},\hat{G},0) = 1$
and $J_{\rm strong}(\hat{\cal H}_1^{(\hat{K}_2)},\hat{G},0) = J_{\rm
strong}(\hat{\cal H}_2^{(\hat{K}_2)},\hat{G},0) = 2$.

\paragraph*{Proof of point~\ref{item:maximal}}

Applying the definition of \gls{DTLHS} control problem
(Def.~\ref{def:epsilon-solution}),  we will show that if 
$K$ is a weak solution to the \gls{LTS} control problem $(\mathrm{LTS}({\cal H})$,
$I$, ${\cal B}_{\|\Gamma\|}(G))$,  and $\hat{\cal H}$ is any full ${\cal Q}$
control abstraction of ${\cal H}$ then there exists a weak solution $\hat{K}$ to the control problem $(\hat{\cal
H}, \hat{I}, \hat{G})$.

Let us define, for $\hat{x} \in \Gamma(A_X)$ and $\hat{u} \in \Gamma(A_U)$,
$\hat{K}(\hat{x},\hat{u})$ $=$  $\exists x \in \Gamma^{-1}(\hat{x})$ $\exists u
\in \Gamma^{-1}(\hat{u}): K(x, u)$. We show that $\hat{K}$ is a weak solution
to any full ${\cal Q}$ control abstraction of ${\cal H}$.

Let $\hat{\cal H}$ be a full ${\cal Q}$ control abstraction of ${\cal H}$. First
of all, we show that $\hat{K}$ is a controller for $\hat{\cal H}$
(Def.~\ref{def:ctrproblem-lts}), i.e. that $\hat{K}(\hat{x},\hat{u})$ implies
$\hat{u} \in {\rm Adm}(\hat{\cal H}, \hat{x})$. Suppose
$\hat{K}(\hat{x},\hat{u})$ holds: this implies that there exist $x \in
\Gamma^{-1}(\hat{x}), u \in \Gamma^{-1}(\hat{u})$ s.t. $K(x, u)$ and $u \in {\rm Adm}({\cal
H}, x)$. If there exists $x' \in A_X$ s.t. $x' \in {\rm Img}({\cal H}, x, u)$
and $\hat{x}' \neq \hat{x}$, then, being $\hat{\cal H}$ a full ${\cal Q}$
control abstraction of ${\cal H}$, we have that $(\hat{x}, \hat{u}, \hat{x}')$
is a transition of $\hat{\cal H}$, thus $\hat{u} \in {\rm Adm}(\hat{\cal H},
\hat{x})$. Otherwise, one of the following must hold:

\begin{itemize}

	\item ${\rm Img}({\cal H}, x, u) = \varnothing$, which is impossible
	since $K(x, u)$;


	\item for all $x' \in A_X$ s.t. $x' \in {\rm Img}({\cal H}, x, u)$, we
have that either $x' \notin A_X$ or $\hat{x}' = \hat{x}$. Being $K$ a weak
controller for $\cal H$ defined only on $A_X \times A_U$ (i.e., $K(x, u)$
implies $x \in A_X$ and $u \in A_U$), and given that $K(x, u)$ holds, we must have that
there exists $x' \in A_X$ s.t. $x' \in {\rm Img}({\cal H}, x, u)$ and $\hat{x}'
= \hat{x}$. If $x = x'$, then there exists an infinite path inside
$\Gamma^{-1}(\hat{x})$ with actions in $\Gamma^{-1}(\hat{u})$, i.e. $(\hat{x},
\hat{u}, \hat{x})$ is a non-eliminable self loop. This implies that
$\hat{N}(\hat{x}, \hat{u}, \hat{x})$ holds, thus $\hat{u} \in {\rm Adm}(\hat{\cal H},
\hat{x})$. Otherwise, i.e. if $x \neq x'$, then we whole reasoning may be
applied to $x'$. Then, either we arrive to a state $t \notin
\Gamma^{-1}(\hat{x})$ starting from a state in $\Gamma^{-1}(\hat{x})$, and
$\hat{N}(\hat{x}, \hat{u}, \hat{t})$ implies $\hat{u} \in {\rm Adm}(\hat{\cal
H}, \hat{x})$, or we have an infinite path inside $\Gamma^{-1}(\hat{x})$ via
$\Gamma^{-1}(\hat{u})$ , thus $(\hat{x}, \hat{u}, \hat{x})$ is a non-eliminable
self loop and $\hat{N}(\hat{x}, \hat{u}, \hat{x})$ implies $\hat{u} \in {\rm
Adm}(\hat{\cal H}, \hat{x})$.


\end{itemize}

We now have to prove that $\hat{K}$ is a weak solution to $\hat{\cal H}$, where
$\hat{\cal H}$ is a full ${\cal Q}$ control abstraction of ${\cal H}$. First of all, we
show that $\hat{I}\subseteq{\rm Dom}(\hat{K})$. Given $\hat{x}\in\hat{I}$, we
have that there exists $x\in\Gamma^{-1}(\hat{x})$  such that $x\in I$. Since $K$
is a weak solution to ${\cal P}$, there exists  $u\in A_U$ s.t. $K(x,u)$, thus
by definition of ${\hat K}$, $\hat{K}(\hat{x}, \hat{u})$  holds, and hence
$\hat{x}\in{\rm Dom}({\hat K})$.


Now, we show that for all $\hat{x}\in{\rm Dom}(\hat{K})$,  $J_{\rm
weak}(\hat{\cal H}^{(\hat{K})}, \hat{G}, \hat{x})$ is finite.  By definition of
$\hat{K}$, and since $K$ is a weak solution to ${\cal P}$, there exists a finite
path  $\pi=x_0 u_0 x_1 u_1 \ldots u_{n-1} x_n$ such that $x_0 \in
\Gamma^{-1}(\hat{x})$, $x_i \in A_X$ for all $0 \leq i \leq n$ and $x_n\in {\cal
B}_{\|\Gamma\|}(G)$.

Let $\hat{\pi} = \hat{x}_0 \hat{u}_0\ldots \hat{u}_{n-1} \hat{x}_n$, and let
$\rho$ be defined from $\hat{\pi}$ by collapsing all consecutive equal
(abstract) states into one state. Formally, $|\rho| = \max_{i \in [n]} \alpha(i)$ and
$\rho(i) = \hat{\pi}^{(S)}(\alpha(i)) = \Gamma(\pi^{(S)}(\alpha(i)))$, where the function 
$\alpha:\N\rightarrow\N$ is recursively defined as follows:

\begin{itemize}  

\item let $Z_z = \{j \;|\; z < j \leq n \land \Gamma(x_j)\not=\Gamma(x_{z})\}$

\item $\alpha(0) = 0$

\item $\alpha(i + 1) = \left\{
\begin{array}{cc}
\alpha(i) & \mbox{if } Z_{\alpha(i)} = \varnothing\\
\min Z_{\alpha(i)} & \mbox{otherwise}
\end{array}
\right.$

\end{itemize}


In a full ${\cal Q}$ control abstraction $\hat{\cal H}$, if $(x, u, x')$ is 
transition of ${\rm LTS}({\cal H})$ and $\hat{x}\not=\hat{x}'$, then
$\hat{N}(\hat{x}, \hat{u}, \hat{x}')$. Then we have that $\rho$ is a
finite path in $\hat{\cal H}^{(\hat{K})}$ that leads from $\hat{x}_0 = \hat{x}$ to
the goal.  As a consequence, $\hat{K}$ is a weak solution to $\hat{\cal P}$.


\paragraph*{Proof of point~\ref{item:whymaximal}}

Analogously to the proof of point~\ref{item:whyminimal}, let $\hat{\cal H}_1 =
(\Gamma(A_X),$ $\Gamma(A_U),$ $T_1)$ and $\hat{\cal H}_2 =
(\Gamma(A_X),$ $\Gamma(A_U),$ $T_2)$ be two full ${\cal Q}$
control abstractions of ${\cal H}$, with $\hat{\cal H}_1 \sqsubseteq\hat{\cal
H}_2$. If $\hat{\cal H}_1 =\hat{\cal H}_2$ the thesis is proved, thus let us
suppose that $\hat{\cal H}_1 \neq\hat{\cal H}_2$. By Fact~\ref{lemma:structure},
the only difference between $\hat{\cal H}_1$ and $\hat{\cal H}_2$ may be in a
finite number of eliminable self loops which are in $\hat{\cal H}_2$ only. Let
$B = \{(\hat{x}_1, \hat{u}_1, \hat{x}_1), \ldots, (\hat{x}_m, \hat{u}_m,
\hat{x}_m)\}$ be the set of such self loops. Let $\hat{K}$ be the weak mgo to
the \gls{LTS} control problem $(\hat{\cal H}_1, \hat{I}, \hat{G})$ and let
$(\hat{x}_i, \hat{u}_i, \hat{x}_i)\in B$. 


Since we already know that $\hat{I} \subseteq {\rm Dom}(\hat{K})$, we only have
to prove that i) $\hat{K}$ is a controller for $\hat{\cal H}_2$ and that ii)
$J_{\rm weak}(\hat{\cal H}_2^{(\hat{K})},\hat{G},\hat{x}) < \infty$ for all
$\hat{x} \in {\rm Dom}(\hat{K})$. 

As for the first point, we have to show that $\hat{K}(\hat{x}, \hat{u})$ implies
$\hat{u} \in {\rm Adm}(\hat{\cal H}_2, \hat{x})$
(Def.~\ref{def:ctrproblem-lts}). Since $\hat{K}(\hat{x}, \hat{u})$ implies
$\hat{u} \in {\rm Adm}(\hat{\cal H}_1, \hat{x})$, and since $\hat{u} \in {\rm
Adm}(\hat{\cal H}_1, \hat{x})$ implies $\hat{u} \in {\rm Adm}(\hat{\cal H}_2,
\hat{x})$, this point is proved.

As for the second one, it is sufficient to prove that $J_{\rm weak}(\hat{\cal
H}_2^{(\hat{K})},\hat{G},\hat{x}) \leq J_{\rm weak}(\hat{\cal
H}_1^{(\hat{K})},\hat{G},\hat{x})$. This can be proved by induction on the value
of $J_{\rm weak}(\hat{\cal H}_1^{(\hat{K})},\hat{G},\hat{x})$. 

Suppose $J_{\rm weak}(\hat{\cal H}_1^{(\hat{K})},\hat{G},\hat{x}) = 1$. Then,
${\rm Img}(\hat{\cal H}_1^{(\hat{K})},\hat{x},\hat{u}) \cap \hat{G} \neq
\varnothing$ for all $\hat{u}$ s.t. $\hat{K}(\hat{x}, \hat{u})$. Since
$\hat{\cal H}_2$ only adds self loops to $\hat{\cal H}_1$, we have that ${\rm
Img}(\hat{\cal H}_2^{(\hat{K})},\hat{x},\hat{u}) \cap \hat{G} \neq \varnothing$
for all $\hat{u}$ s.t. $\hat{K}(\hat{x}, \hat{u})$, thus $J_{\rm weak}(\hat{\cal
H}_2^{(\hat{K})},\hat{G},\hat{x}) = 1 = J_{\rm weak}(\hat{\cal
H}_1^{(\hat{K})},\hat{G},\hat{x})$.


Suppose now that for all $\hat{x}$ s.t. $J_{\rm weak}(\hat{\cal
H}_1^{(\hat{K})},\hat{G},\hat{x}) = n$, $J_{\rm weak}(\hat{\cal
H}_2^{(\hat{K})},\hat{G},\hat{x}) \leq J_{\rm weak}(\hat{\cal
H}_1^{(\hat{K})}$, $\hat{G},$ $\hat{x})$. Let $\hat{x}$ be s.t. $J_{\rm
weak}(\hat{\cal H}_1^{(\hat{K})},\hat{G},\hat{x}) = n + 1$. If $(\hat{x},
\hat{u}, \hat{x}) \notin B$ for any $\hat{u}$, then ${\rm Img}(\hat{\cal
H}_1^{(\hat{K})}, \hat{x}, \hat{u}) = {\rm Img}(\hat{\cal H}_2^{(\hat{K})},
\hat{x}, \hat{u})$ for all $\hat{u}$, thus $J_{\rm weak}(\hat{\cal
H}_2^{(\hat{K})},\hat{G},\hat{x}) \leq J_{\rm weak}(\hat{\cal
H}_1^{(\hat{K})},\hat{G},\hat{x})$ by induction hypothesis. Otherwise, let
$(\hat{x}, \hat{u}, \hat{x}) \in B$ for some $\hat{u}$.  If
$\hat{x}\notin\hat{G}$ we simply have that $J_{\rm weak}(\hat{\cal
H}_2^{(\hat{K})},\hat{G},\hat{x}) \leq J_{\rm weak}(\hat{\cal
H}_1^{(\hat{K})},\hat{G},\hat{x})$ by induction hypothesis. Otherwise, if
$\hat{x}\in\hat{G}$, let $\hat{K}_1$ be s.t. $\hat{K}_1(\hat{x}, \hat{u}) = 0$
and $\hat{K}_1(\hat{s}, \hat{a}) = \hat{K}(\hat{s}, \hat{a})$ for $(\hat{s},
\hat{a}) \neq (\hat{x}, \hat{u})$. Then, $J_{\rm weak}(\hat{\cal
H}_2^{(\hat{K})},\hat{G},\hat{x}) = \max\{1, J_{\rm weak}(\hat{\cal
H}_1^{(\hat{K}_1)}$, $\hat{G},$ $\hat{x})\}$ $\leq J_{\rm weak}(\hat{\cal
H}_1^{(\hat{K})},\hat{G},\hat{x}')$, thus the thesis is proved.
%
%
%
%
%
\end{proof}

%
%
%
 

\section{Quantized Controller Synthesis}
\label{ctr-syn-algorithm.tex}

In this section, we present the quantized controller synthesis  algorithm
(function \fun{qCtrSyn} in  Alg.~\ref{qfc-syn-outline}). Function \fun{qCtrSyn}
takes as input a \gls{DTLHS} control problem ${\cal P} = ({\cal H}, I, G)$ and a
quantization ${\cal Q}$. Then, resting on Theor.~\ref{corol:min-ctr-abs},
\fun{qCtrSyn} computes  an admissible ${\cal Q}$ control  abstraction $\hat{\cal
M}$ in order to find a ${\cal Q}$ \gls{QFC} strong solution to ${\cal P}$, 
and a full ${\cal Q}$ control abstraction $\hat{\cal W}$
to determine if   such a solution does not exist. 

Sects.~\ref{proof:selfLoop},~\ref{lemma:minCtrAbs.proof},~\ref{sec:qks-correctness-proof},
and~\ref{admissibility.check.subsec} show theoretical and implementation
details that can be skipped at a first reading.


\sloppy

Namely, as for the sufficient condition, we compute the strong mgo $\hat{K}$ for
the \gls{LTS} control problem $(\hat{\cal M}, \Gamma(I), \Gamma(G))$. If $\hat{K}$
exists, then a ${\cal Q}$ \gls{QFC} strong solution to ${\cal P}$ may be built from
$\hat{K}$. Note that, if $\hat{K}$  does not exist, a strong solution may exist for
some other admissible ${\cal Q}$ control abstraction $\hat{\cal H}$. However, by
point~\ref{item:whyminimal} of Theor.~\ref{corol:min-ctr-abs}, $\hat{\cal H}$
must be lower than $\hat{\cal M}$ in the hierarchy lattice (see
Fig.~\ref{lattice.eps}). 
This suggests to compute $\hat{\cal M}$ as the minimum
(admissible) ${\cal Q}$ control abstraction of ${\cal H}$.
Since by Prop.~\ref{minctrabs.undec.prop} we are not able to compute the
minimum ${\cal Q}$ control abstraction, we compute $\hat{\cal M}$ as a {\em close to}
minimum admissible ${\cal Q}$ control abstraction, i.e. an admissible ${\cal
Q}$ control abstraction containing as few eliminable self loops as possible (see
Ex.~\ref{example-strong.tex}).

\fussy

As for the necessary condition, we compute the weak mgo $\hat{K}$ for the \gls{LTS}
control problem $(\hat{\cal W}, \Gamma(I), \Gamma(G))$. If $\hat{K}$ does not
exists, then a ${\cal Q}$ \gls{QFC} (weak as well as strong) solution to ${\cal P}$
cannot exist.  Note that, if $\hat{K}$ exists, a weak solution may not exist for some
other full ${\cal Q}$ control abstraction $\hat{\cal H}$. However, by
point~\ref{item:whymaximal} of Theor.~\ref{corol:min-ctr-abs}, $\hat{\cal H}$
must be lower than $\hat{\cal W}$ in the hierarchy lattice (see
Fig.~\ref{lattice.eps}). Hence, again by Prop.~\ref{minctrabs.undec.prop}, we
compute $\hat{\cal W}$ as the close to minimum full ${\cal Q}$ control
abstraction.



\begin{algorithm}
  \caption[\gls{QFC} Synthesis]{\gls{QFC} synthesis} 
  \label{qfc-syn-outline}
  \begin{algorithmic}[1]
    \REQUIRE
    \gls{DTLHS} control problem $({\cal H}, I, G)$, quantization ${\cal Q}=(A, \Gamma)$
    \ENSURE {\fun{qCtrSyn}$({\cal H}$, ${\cal Q}$, $I$, $G)$}
    \STATE $\hat{I}\gets\Gamma(I)$, $\hat{G}\gets\Gamma(G)$\label{item:IGquant} 
    \STATE $\hat{\cal M}$ $\gets$\fun{minCtrAbs}$({\cal H}$, ${\cal Q})$
    		\label{item:minCtrAbs}
    \STATE $(b,\hat{D}, \hat{K})$ $\gets$ 
    \fun{strongCtr}$(\hat{\cal M}$, $\hat{I}$, $\hat{G})$
    			\label{item:StrongCtr}
    \STATE {\bf if} $b$ 
    {\bf then return }$(${\sc Sol}, $\hat{D}, \hat{K})$
    			\label{item:qfcFound} 
    \STATE $\hat{\cal W}$ $\gets$ 
    \fun{minFullCtrAbs} $({\cal H}, {\cal Q})$
    			\label{item:minFullCtrAbs}
    \STATE {\bf if} 
    \fun{existsWeakCtr}$(\hat{\cal W}, \hat{I}, \hat{G})$ 
    {\bf then return}~$(${\sc Unk}, $\hat{D}, \hat{K})$
    			\label{item:fail}
    \STATE {\bf else return} 
    $(${\sc NoSol}, $\hat{D}, \hat{K})$
    			\label{item:noSolution}
  \end{algorithmic}
\end{algorithm}  

\subsection{\gls{QFC} Synthesis Algorithm}
\label{sec:QFCsynthesis}
\label{sec:qfc-syn-outline}

\sloppy

Our \gls{QFC} synthesis algorithm (function \fun{qCtrSyn}  
outlined in Alg.~\ref{qfc-syn-outline})  
takes as input a \gls{DTLHS} 
${\cal H}$ = ($X$, $U$, $Y$, $N$),
a quantization ${\cal Q} = (A, \Gamma$), and 
two predicates $I$ and $G$ over $X$, such that
(${\cal H}$, $I$, $G$) is a \gls{DTLHS} control problem.
Function \fun{qCtrSyn} returns a tuple
($\mu$, $\hat{D}$, $\hat{K}$), where:
$\mu \in \{\mbox{\sc Sol}, \mbox{\sc NoSol}, \mbox{\sc Unk}\}$, 
$\hat{D}$ = $\mbox{\rm Dom}(\hat{K})$ and
$\hat{K}$ is such that the controller 
$K$, defined by $K(x,u)=\hat{K}(\Gamma(x), \Gamma(u))$ 
is a ${\cal Q}$ \gls{QFC} (strong) solution to the control problem
$({\cal H}, \Gamma^{-1}(\hat{D}), G)$.

\fussy

We represent boolean functions (e.g. the transition relation
of $\hat{\cal H}$) and sets (by using their
characteristic functions) using 
\newacronym{OBDD}{OBDD}{Ordered Binary Decision Diagram}%
\glspl{OBDD}~\cite{Bry86}. 
For the sake of clarity, however, we will present
our algorithms 
using a set theoretic notation for sets and predicates
over sets.

Alg.~\ref{qfc-syn-outline} starts (line~\ref{item:IGquant}) by computing 
a quantization $\hat{I}$ of the initial region $I$ 
and a quantization $\hat{G}$ of the goal region $G$ 
(further details are given in Sect.~\ref{dths-ctr-abs.tex}).

Function \fun{minCtrAbs} in line~\ref{item:minCtrAbs} computes the
close to minimum ${\cal Q}$ control abstraction $\hat{\cal M}$ of ${\cal H}$ 
(see Sect.~\ref{sec:minFullCtrAbs} for further details about \fun{minFullCtrAbs}).

Line~\ref{item:StrongCtr} determines if a strong mgo to the \gls{LTS} control problem 
$\hat{\cal P} = (\hat{\cal M}, \hat{I}, \hat{G})$ exists by  
calling function \fun{strongCtr} \rimandotekrep{sec:lts-alg} that implements a variant of the
algorithm in \cite{strong-planning-98}.
Given $\hat{\cal M}, \hat{I}, \hat{G}$, function \fun{strongCtr} returns a triple 
$(b, \hat{D}, \hat{K})$ 
such that 
$\hat{K}$ is the strong mgo to 
$(\hat{\cal M}, \varnothing, \hat{G})$ and $\hat{D} = \mbox{\rm Dom}(\hat{K})$.
If 
$b$ is {\sc True} then $\hat{K}$ is a strong mgo for $\hat{\cal P}$ (i.e. $\hat{I} \subseteq \hat{D}$),
and 
\fun{qCtrSyn} returns the tuple
$(\mbox{\sc Sol}, \hat{D}, \hat{K})$ (line~\ref{item:qfcFound}). 
By 
Theor.~\ref{corol:min-ctr-abs} (point~\ref{item:minimal}),  
$K(x, u)$ = $\hat{K}(\Gamma(x), \Gamma(u))$ is a ${\cal Q}$ \gls{QFC} solution to the
\gls{DTLHS} control problem $({\cal H}, I, G$). Otherwise, 
in lines~\ref{item:minFullCtrAbs}--\ref{item:noSolution}
\fun{qCtrSyn} tries to establish if such a solution may exist or not.

\sloppy

Function \fun{minFullCtrAbs} in line~\ref{item:minFullCtrAbs} computes the
close to minimum full ${\cal Q}$ control abstraction $\hat{\cal W}$ of ${\cal H}$ 
(see Sect.~\ref{sec:minFullCtrAbs} for further details about \fun{minFullCtrAbs}).
Line~\ref{item:fail} checks if the
weak mgo to 
$\hat{\cal P}' = (\hat{\cal W}, \hat{I}, \hat{G})$ exists  
by calling function \fun{existsWeakCtr} \rimandotekrep{sec:lts-alg}, which is based on the algorithm
in~\cite{Tro98}.

\fussy

If function \fun{existsWeakCtr} returns {\sc False}, 
then a weak mgo 
to 
$\hat{\cal P}'$ 
does not exist, and since the weak mgo is unique 
no weak solution exists to $\hat{\cal P}'$. 
By 
Theor.~\ref{corol:min-ctr-abs} (point~\ref{item:maximal}),
no ${\cal Q}$ \gls{QFC} solution exists for the
\gls{DTLHS} control problem $({\cal H}, I, G)$ and accordingly \fun{qCtrSyn}
returns {\sc NoSol} (line~\ref{item:noSolution}). Otherwise no
conclusion can be drawn and accordingly {\sc Unk} is returned (line~\ref{item:fail}). 
%
In any case, the 
strong mgo $\hat{K}$ for $\hat{\cal P}$ for the
(close to) minimum control abstraction is returned, together with its
controlled region $\hat{D}$.
\subsection{Synthesis Algorithm Correctness}
\label{sec:qks-correctness}
The above considerations imply correctness
of function \fun{qCtrSyn} (and thus of our approach), as stated by the following theorem.

\begin{theorem}
  \label{theor:synth-correct}
  Let $\cal H$ be a \gls{DTLHS}, 
  ${\cal Q}=(A, \Gamma)$  be a quantization, and  
  (${\cal H}$, $I$, $G$) be a \gls{DTLHS} control problem. 
  Then \fun{qCtrSyn}(${\cal H}$, ${\cal Q}$, $I$, $G$)
  returns a triple
  ($\mu$, $\hat{D}$, $\hat{K}$) such that:
$\mu \in \{\mbox{\sc Sol},$ $\mbox{\sc NoSol},$ $\mbox{\sc Unk}\}$,
$\hat{D}$ = $\mbox{\rm Dom}(\hat{K})$ and,
for all control laws $k$ for $\hat{K}$, $K(x,u)=(k(\Gamma(x)) = \Gamma(u))$ 
is a ${\cal Q}$ \gls{QFC} solution to the control problem
$({\cal H}, \Gamma^{-1}(\hat{D}), G)$.
Furthermore, the following holds:
i) if $\mu$ = $\mbox{\sc Sol}$
then 
$I \subseteq \Gamma^{-1}(\hat{D})$ and 
$K$ is a 
${\cal Q}$ \gls{QFC} solution to the control problem
$({\cal H}, I, G)$;
%
ii) if $\mu$ = $\mbox{\sc NoSol}$ then
there is no 
${\cal Q}$ \gls{QFC} solution to the control problem
$({\cal H}, I, G)$.
%
\end{theorem}

\sloppy

\begin{remark}
  \label{rem:timeoptimality}
\cite{MazoTabuada11} describes a method for the automatic control software synthesis for continuous time linear systems.
Function \fun{strongCtr}, as well as the approach in \cite{MazoTabuada11}, returns $\hat{K}$ as a (worst case) \emph{time optimal} controller, i.e. in each state $\hat{K}$ enables the actions leading to a goal state in the least number of transitions.
This stems from the fact that in both cases (\fun{strongCtr} and \cite{MazoTabuada11}) the \gls{OBDD} representation for the controller is computed using the approach in \cite{strong-planning-98} where symbolic control synthesis algorithms for finite state \glspl{LTS} have been studied in a universal planning setting.
\end{remark}

\begin{remark}
\label{rem:emsoft12}
Instead of computing the controller (function \fun{strongCtr}) with \cite{strong-planning-98}, it is possible to trade the size of the synthesized controller with time optimality while preserving closed loop performances. Such an issue has been investigated in \cite{emsoft12}.
\end{remark}

\fussy

\begin{remark}
\label{rem:near-optimal}
%
Note however that $\hat{K}$ may not be time optimal for
the real plant. In fact, self loops elimination shrinks all concrete sequences 
of the form $x_n$, $ u_n$, $ \ldots$, $ x_m$ in every path of
LTS$({\cal H})$ into a single abstract transition $(\hat{x}_n,$ $\hat{u}_n,$ $\hat{x}_m)$
of $\hat{\cal M}$ whenever  $\hat{x}_n=\ldots
=\hat{x}_{m-1}$ and $\hat{u}_n=\ldots
=\hat{u}_{m-1}$. Thus, the length of paths
in the plant model and those in the control abstraction used for the synthesis may not coincide.
Moreover, nondeterminism added by quantization might lead to prefer an action
$\hat{u}_1$ to an action $\hat{u}_2$ for an abstract state $\hat{x}$, whilst
actions in $\hat{u}_2$ might be better for some real states inside $\hat{x}$.
Finally, since we are not able to compute the minimum control abstraction, we
may discard
a possibly optimal action $\hat{u}$ on a state $\hat{x}$ if the following holds: $(\hat{x}, \hat{u}, \hat{x})$ is an eliminable self loop, but function \fun{minCtrAbs} decides that it is non-eliminable.
For these
reasons we refer to our controller as a \emph{near time optimal} controller.
\end{remark}





\subsection{Quantization}
\label{dths-ctr-abs.tex}

In the following let $\cal H$ = ($X$, $U$, $Y$, $N$) be a \gls{DTLHS}, 
${\cal Q}=(A, \Gamma)$ be a quantization for ${\cal H}$,
and (${\cal H}$, $I$, $G$) be a \gls{DTLHS} control problem.

\sloppy

In our approach we consider $\Gamma$ only in problems of type
$P(W) \equiv (\max, J(W), L(W) \land (\Gamma(W) =
\hat{v}))$, where $W$ is either $X, X'$ or $U$, $J(W)$ is a linear expression,
$L(W)$ a conjunctive predicate and $(\Gamma(W) = \hat{v})\equiv\bigwedge_{i\in
[|W|]} (\gamma_{w_i}(w_i) = \hat{v}_{i})$, with $w_i \in W$. In order to be able
to solve $P(W)$ via a \gls{MILP} solver, 
%
we restrict ourselves to quantization functions
$\gamma_{w_i}$ for which equality tests can be
represented by using conjunctive predicates.  
Namely, for $w \in X \cup U$, we employ the uniform quantization $\gamma_w : A_w
\rightarrow [0, \Delta_w - 1]$, defined for a given $\Delta_w$ as follows. Let $\delta_w =
(\sup A_w - \inf A_w)/\Delta_w$. We have that $\gamma_w(w)=\hat{z}$ if and only if the conjunctive predicate
$P_{\gamma_w}(w, \hat{z}) \equiv \inf A_w + \delta_w\hat{z} \leq w \leq
\inf A_w + \delta_w(\hat{z}+1)$ holds.

\fussy


We may now explain how $\hat{I}, \hat{G}$ are effectively computed in
line~\ref{item:IGquant} of Alg.~\ref{qfc-syn-outline}.
Since the initial region $I$ is represented as a conjunctive predicate, 
its quantization $\hat{I}$ 
is computed by solving $|\Gamma(A_X)|$ feasibility
problems. More precisely, $\hat{I}$ = $\{\hat{x} \; | \;$ 
\fun{feasible}($I(X)\;\land \;\Gamma(X)=\hat{x}$)$\}$. 
Similarly, the quantization $\hat{G}$ of the goal region $G$ 
is $\hat{G}=\{\hat{x}\; |\; \mbox{\fun{feasible}}(G(X)\land \Gamma(X)=\hat{x})\}$.

%
%
%
%
%
%
%

\begin{algorithm}
  \caption[Synthesis: Building control abstractions]{Building control abstractions}
  \label{ctr-abs.alg}\label{symb.alg.rat}
  \begin{algorithmic}[1]
    \REQUIRE
    \gls{DTLHS} ${\cal H} = (X, U, Y, N)$, quantization ${\cal Q}=(A, \Gamma)$.
    \ENSURE {\fun{minCtrAbs}
    $({\cal H}$, ${\cal Q})$}
			\label{name.alg.step}
    \STATE $\hat{N} \gets \varnothing$
			\label{init.alg.step}
    \FORALL {$\hat{x} \in \Gamma(A_{X})$} 
\label{forall_s.alg.step}
     \FORALL {$\hat{u} \in \Gamma(A_{U})$} \label{forall_u.alg.step} 
      \STATE \colorbox{light-gray}{{\bf if} $\neg$ \fun{${\cal Q}$-admissible}$({\cal H}, {\cal Q}, \hat{x}, \hat{u})$ {\bf then continue}} 
		\label{check_s_u.alg.step} 
      \STATE 
      {{\bf if} \fun{selfLoop}(${\cal H}$,${\cal Q}$, $\hat{x}$,$\hat{u}$) 
{\bf then} $\hat{N} \gets \hat{N} \cup \{(\hat{x}, \hat{u}, 
\hat{x})\}$}  
      \label{self_loop.alg.step}
	  \STATE ${\cal O}$ $\gets$ \fun{overImg}$({\cal H}, {\cal Q}, \hat{x}, \hat{u})$
	  \label{overimg.alg.step}	
      \FORALL {$\hat{x}' \in \Gamma({\cal O})$}
        \label{forall_s_prime.alg.step}
	\IF {
	{$\hat{x} \neq \hat{x}'\land$}$\mbox{\fun{existsTrans}}({\cal H}, {\cal Q}, \hat{x},\hat{u},\hat{x}')$}
\label{check_s_s_prime.alg.step}
        \STATE $\hat{N}\!\gets\!\hat{N} \cup \{(\hat{x}, \hat{u}, \hat{x}')\}$\label{update_t.alg.step}
        \ENDIF
      \ENDFOR
     \ENDFOR
    \ENDFOR
    \STATE {\bf return} $\hat{N}$\label{return.alg.step}
  \end{algorithmic}
\end{algorithm}  
\subsection{Computing Minimum Control Abstractions}
\label{minmax-ctrabs.tex}

In this section, we present in Alg.~\ref{ctr-abs.alg} function 
\fun{minCtrAbs}, which effectively computes a close to minimum 
${\cal Q}$ control abstraction $\hat{\cal M} = 
(\Gamma(A_X), \Gamma(A_U), \hat{N})$ for a given ${\cal H}$. 

Starting from the empty transition relation (line \ref{init.alg.step})
function \fun{minCtrAbs} checks for every triple
$(\hat{x},\hat{u},\hat{x}') \in \Gamma(A_X)\times\Gamma(A_U)\times\Gamma(A_X)$ 
if the transition $(\hat{x},\hat{u},\hat{x}')$ 
belongs to 
$\hat{\cal M}$
and accordingly adds it 
to $\hat{N}$ or not.


\sloppy

For any pair $(\hat{x},\hat{u})$ in $\Gamma(A_X)\times\Gamma(A_U)$ 
line~\ref{check_s_u.alg.step} 
checks if $\hat{u}$ is ${\cal Q}$-admissible in $\hat{x}$.
This check is carried out by determining if
the 
predicate
$P(X, U, Y, X', \hat{x}, \hat{u}) \equiv N(X, U, Y, X')\land \Gamma(X)=\hat{x} \land \Gamma(U)=\hat{u}     
\land X'\not\in A_X$ is not feasible. 
If $\hat{u}$ is not ${\cal Q}$-admissible in $\hat{x}$ (i.e., if $P(X, U, Y, X', \hat{x}, \hat{u})$ is feasible),
no transition of the form $(\hat{x},\hat{u},\hat{x}')$
is added to $\hat{N}$.
Note that $P(X, U, Y, X', \hat{x}, \hat{u})$ is not a conjunctive predicate,
however it is possible to check its feasibility by properly calling function
\fun{feasible} $2|X|$ times (Sect.~\ref{admissibility.check.subsec}).


\fussy


If $\hat{u}$ is ${\cal Q}$-admissible in $\hat{x}$, line~\ref{self_loop.alg.step} checks if the self loop
$(\hat{x},\hat{u},\hat{x})$ has to be added to $\hat{N}$.
To this aim, we employ a function \fun{selfLoop} (see Sect.~\ref{self_loop.subsubsec}) 
which takes a (state, action) pair $(\hat{x},\hat{u})$ 
and returns {\sc False} if 
the self loop $(\hat{x},\hat{u},\hat{x})$ is eliminable. 

\sloppy

Function \fun{overImg} (line~\ref{overimg.alg.step}) computes a  
rectangular region ${\cal O}$, that is a \emph{quite tight} 
overapproximation of the set of one step reachable states 
from $\hat{x}$ via $\hat{u}$.
${\cal O} $ is obtained by computing for each state variable $x_i$ 
the minimum and maximum possible values for the corresponding next state
variable. Namely,
${\cal O} = 
\prod_{i = 1, \ldots, |X|}[\gamma_{x_i}(m_i), \gamma_{x_i}(M_i)]
$
where 
$m_i = \fun{optimalValue}(\min,$ $x'_i,$ $N(X, U, Y, X') \land A(X')\land
\Gamma(X)=\hat{x} \land \Gamma(U)=\hat{u})$ 
and 
$M_i = \fun{optimalValue}(\max,$ $x'_i,$ $N(X, U, Y, X') \land A(X')\land
\Gamma(X)=\hat{x} \land \Gamma(U)=\hat{u})$.

\fussy

Finally, for each abstract state $\hat{x}'\in\Gamma({\cal O})$
line~\ref{check_s_s_prime.alg.step} checks if 
there exists a concrete transition realizing the abstract transition
($\hat{x}$, $\hat{u}$, $\hat{x}'$) when $\hat{x}$ $\not=$ $\hat{x}'$.
To this end, function \fun{existsTrans} solves the \gls{MILP} problem
$N(X, U, Y, X')\land \Gamma(X)=\hat{x} \land \Gamma(U)=\hat{u}\land \Gamma(X')=\hat{x}'$.

\begin{remark}
\label{explicit-loops-remark}
%
%
From the nested loops in lines 
\ref{forall_s.alg.step}, 
\ref{forall_u.alg.step},
\ref{forall_s_prime.alg.step} 
we have that \fun{minCtrAbs} worst case runtime 
is $O(|\Gamma(A_X)|^{2}|\Gamma(A_U)|)$. However,
thanks to the heuristic implemented in function \fun{overImg}, 
\fun{minCtrAbs} typical runtime is about
$O(|\Gamma(A_X)||\Gamma(A_U)|)$
as confirmed by our experimental results
(see Sect.~\ref{expres.tex}, Fig.~\ref{bits-time.eps}). The same holds for
function \fun{minFullCtrAbs} (see Sect.~\ref{sec:minFullCtrAbs}).
%
%
\end{remark}

\begin{remark}
\label{symbolic-modes-remark}
Function \fun{minCtrAbs} is explicit in the (abstract) states and actions 
of $\hat{\cal H}$ and symbolic with respect to 
the auxiliary variables (\emph{modes}) in the 
transition relation $N$ of ${\cal H}$.
As a result our approach will work well with systems with just a few
state variables and many modes, our target here.
\end{remark}

\subsubsection{Computing Minimum Full Control Abstraction}
\label{sec:minFullCtrAbs}

Function \fun{minCtrAbs} can be easily modified in order to compute
the close to minimum full ${\cal Q}$ control abstraction, thus obtaining function 
\fun{minFullCtrAbs} called in Alg.~\ref{qfc-syn-outline},
line~\ref{item:minFullCtrAbs}.
Function \fun{minFullCtrAbs}  
is obtained by removing  
the highlighted code (on grey background) from Alg.~\ref{ctr-abs.alg}, namely
the admissibility check in line~\ref{check_s_u.alg.step}.
%
%
%

\subsection{Self Loop Elimination}
\label{self_loop.subsubsec}

\sloppy

In order to exactly get the minimum control abstraction, 
function \fun{selfLoop} should return {\sc True} 
iff the given self loop is
non-eliminable.
This is undecidable by Prop.~\ref{self.loop.undec.prop}.
Function \fun{selfLoop}, outlined in Alg.~\ref{alg:selfLoop}, 
checks a sufficient {\em gradient based} condition for self loop elimination that in practice
turns out to be very effective (see Tabs.~\ref{expres-table.tex:1} and~\ref{expres-table.tex:2} in Sect.~\ref{expres.tex}). That is, 
function \fun{selfLoop} returns {\sc False} when a self loop is
eliminable (or there is not a concrete witness for it). On the other hand, if function \fun{selfLoop} returns {\sc True},
then the self loop under consideration may be non-eliminable as well as
eliminable. In a conservative way, we assume self loops for which function
\fun{selfLoop} returns {\sc True} to be non-eliminable (i.e. they are added to
$\hat{\cal M}$, see line~\ref{self_loop.alg.step} of
Alg.~\ref{ctr-abs.alg}).

\fussy

Function \fun{selfLoop} in Alg.~\ref{alg:selfLoop},
which correctness is proved in Sect.~\ref{proof:selfLoop},
works as follows. First of all it checks if there is a 
concrete witness for the self loop under consideration.
If it is not the case, \fun{selfLoop} returns {\sc False}
(line~\ref{selfLoop:exTrans}). 
Otherwise, for each real variable $x_i$, it tries to establish if
$x_i$ is either always increasing (line~\ref{line:positive})
or always decreasing (line~\ref{line:negative}) inside $\Gamma^{-1}(\hat{x})$ by performing 
actions in $\Gamma^{-1}(\hat{u})$. If this is the case, we have that, being 
$\Gamma^{-1}(\hat{x})$ a compact set, no Zeno-phenomena may arise,
thus executing actions in $\Gamma^{-1}(\hat{u})$ it is
guaranteed that ${\cal H}$ will eventually leave the region
$\Gamma^{-1}(\hat{x})$. Otherwise, {\sc True} is returned in
line~\ref{line:selfloopgivesup}.

\algsetup{indent=0.4em}

\begin{algorithm}
  \caption[selfLoop]{Self loop elimination}
  \label{alg:selfLoop}
  \begin{algorithmic}[1]
    \REQUIRE \gls{DTLHS} ${\cal H} = (X, U, Y, N)$, quantization ${\cal Q}=(A,
    \Gamma)$, abstract state $\hat{x}$, abstract action $\hat{u}$.
    \ENSURE {\fun{selfLoop}$({\cal H}, {\cal Q}, \hat{x}$, $\hat{u})$}
    \STATE {\bf if} $\neg$\fun{existsTrans}$(\hat{x},\hat{u},\hat{x})$ {\bf then return} {\sc False}
	\label{selfLoop:exTrans}
 	\FOR{$x_i$ {\bf in} $X^r$}
    \STATE $w_i \gets$ \fun{optimalValue}($\min, x'_i-x_i, N(X, U, Y, X')
             \land \Gamma(X) = \hat{x} \land \Gamma(U) = \hat{u} \land \Gamma(X') = \hat{x}$)\label{selfLoop:milp1}
    \STATE {\bf if} $w_i>0$ {\bf then return} {\sc False}\label{line:positive}
    \STATE $W_i \gets$ \fun{optimalValue}($\max, x'_i-x_i, N(X, U, Y, X')\land
             \Gamma(X) = \hat{x} \land \Gamma(U) = \hat{u} \land \Gamma(X') = \hat{x}$)\label{selfLoop:milp2}
    \STATE {\bf if} $W_i<0$ {\bf then return} {\sc False}\label{line:negative}             
    \ENDFOR\label{selfLoop:loopend}
    \STATE {\bf return} {\sc True}\label{line:selfloopgivesup}
  \end{algorithmic}
\end{algorithm}  

\algsetup{indent=1em} 

%


\subsection{Proof of Function \fun{selfLoop} Correctness}
\label{proof:selfLoop}
In this section we prove correctness of Alg.~\ref{alg:selfLoop}.
This section can be skipped at a first reading.

\begin{proposition}
\label{prop:selfLoop}
\sloppy

Let ${\cal H} = (X, U, Y, N)$ be a \gls{DTLHS},  
${\cal Q} = (A,\Gamma)$ be a quantization for ${\cal H}$, 
$\hat{x}\in\Gamma(A_X)$, and $\hat{u}\in\Gamma(A_U)$.
If the abstract self loop $(\hat{x}, \hat{u}, \hat{x})$ has a concrete witness
and \fun{selfLoop}$({\cal H}, {\cal Q}, \hat{x}, \hat{u})$ returns {\sc False}, then 
$(\hat{x}, \hat{u}, \hat{x})$ is an eliminable self loop.
%
\end{proposition}

\begin{proof}

\sloppy

Suppose by absurd that the abstract self loop $(\hat{x}, \hat{u}, \hat{x})$ has
a concrete witness, \fun{selfLoop}$({\cal H}, {\cal Q}, \hat{x}, \hat{u})$ returns {\sc False}, and
$(\hat{x}, \hat{u}, \hat{x})$ is a non-eliminable self loop. Then there exists
an infinite run $\pi = x_0 u_0 x_1 u_1\ldots$ such that for all $t\in\N$ $x_t\in
\Gamma^{-1}(\hat{x})$ and $u_t\in \Gamma^{-1}(\hat{u})$.

\fussy

For $i \in [|X^r|]$, let $w_i \leq W_i$ be the values computed in
lines~\ref{selfLoop:milp1} and~\ref{selfLoop:milp2} of Alg.~\ref{alg:selfLoop},
i.e. $w_i$ = \fun{optimalValue}%
($\min$, $x'_i-x_i$, $N(X, U, Y, X')\land
\Gamma(X) = \hat{x} \land \Gamma(U) = \hat{u} \land \Gamma(X') = \hat{x}$) and 
$W_i$ = \fun{optimalValue}%
($\max$, $x'_i-x_i$, $N(X, U, Y, X') \land \Gamma(X) =
\hat{x} \land \Gamma(U) = \hat{u} \land \Gamma(X') = \hat{x}$).

Since \fun{selfLoop}$({\cal H}, {\cal Q}, \hat{x}, \hat{u})$ returns {\sc False},  there exists at
least an index $j\in [|X^r|]$ such that  $w_j>0$ or $W_j<0$ (see
lines~\ref{line:positive} and~\ref{line:negative} of Alg.~\ref{alg:selfLoop}
resp.). Let us consider the former case (note that $w_j > 0$ implies $W_j > 0$).

For all $k\in\N$, we have that $|(x_k)_j-(x_0)_j| = (x_k)_j-(x_0)_j \geq k w_j$.
If we take $\tilde{k}>\frac{\|\gamma_{x_j}\|}{w_j}$, we have that 
$|(x_{\tilde{k}})_j-(x_0)_j|>\|\gamma_{x_j}\|$ and hence $x_{\tilde{k}}$ cannot belong 
to $\Gamma^{-1}(\hat{x})$.

Analogously, if $w_j \leq W_j < 0$ then we have that $|(x_k)_j-(x_0)_j| =
(x_0)_j-(x_k)_j \geq k w_j$. If we take
$\tilde{k}>\frac{\|\gamma_{x_j}\|}{w_j}$, we have that 
$|(x_{\tilde{k}})_j-(x_0)_j|>\|\gamma_{x_j}\|$ and hence $x_{\tilde{k}}$ cannot
belong  to $\Gamma^{-1}(\hat{x})$.

In both cases we have a contradiction, thus the thesis is proved.
\end{proof}

\subsection{Proof of Functions \fun{minCtrAbs} and \fun{minFullCtrAbs} Correctness}
\label{lemma:minCtrAbs.proof}
In this section we prove correctness of functions \fun{minCtrAbs} (Alg.~\ref{ctr-abs.alg}) 
and \fun{minFullCtrAbs} used in Alg.~\ref{qfc-syn-outline}.
This section can be skipped at a first reading.

\begin{proposition}
\label{lemma:CtrAbs}
Let ${\cal H} = (X, U, Y, N)$ be a DTLHS and 
${\cal Q} = (A,\Gamma)$ be a quantization for ${\cal H}$.

If $\hat{N}$ is the transition relation computed by \fun{minCtrAbs}(${\cal H}$, ${\cal Q}$) 
then $\hat{\cal H} = (\Gamma(A_X)$, $\Gamma(A_U)$, $\hat{N})$ is an admissible 
${\cal Q}$ control abstraction of ${\cal H}$. 

If $\hat{N}$ is the transition relation computed by \fun{minFullCtrAbs}(${\cal H}$, ${\cal Q}$) 
then $\hat{\cal H} = (\Gamma(A_X)$, $\Gamma(A_U)$, $\hat{N})$ is a full ${\cal Q}$
control abstraction of ${\cal H}$. 
\end{proposition}


\begin{proof}
\sloppy
Here we prove only the part regarding function \fun{minCtrAbs}, since the other
part may be proved analogously. We first show that the control abstraction 
$\hat{\cal H} = (\Gamma(A_X), \Gamma(A_U), \hat{N})$  satisfies
conditions~\ref{item:witness}--\ref{item:loop}
of Def.~\ref{def:ctr-abs}.

\begin{enumerate}
\item

Each transition $(\hat{x}, \hat{u}, \hat{x}')$ is added to $\hat{N}$ 
in line~\ref{self_loop.alg.step} or in line~\ref{update_t.alg.step} 
of Alg.~\ref{symb.alg.rat}. In both cases, it has been checked by function 
\fun{existsTrans} 
that $\exists x \in \Gamma^{-1}(\hat{x})$,
    $u \in \Gamma^{-1}(\hat{u})$,
    $x' \in \Gamma^{-1}(\hat{x}')$, 
    $y \in A_Y$
    such that 
    $N(x, u, y, x')$
(in the latter case the check is inside function \fun{selfLoop}).
 
\item
Let $x,s'\in A_X$ and $u\in A_U$ be such that $\exists y: N(x, u, y, x')$ and
$\Gamma(x) \neq \Gamma(x')$.
Since \fun{minCtrAbs} examines all tuples in $\Gamma(A_X)\times\Gamma(A_U)\times\Gamma(A_X)$, it will eventually examine
the tuple $(\hat{x},\hat{u},\hat{x}')$ s.t. $\hat{x}=\Gamma(x)$, $\hat{u}=\Gamma(u)$, and $\hat{x}'=\Gamma(x')$.
If $\hat{u}$ is not ${\cal Q}$-admissible in $\hat{x}$ no transition is added to
$\hat{N}$ because of the check in line~\ref{check_s_u.alg.step}.
Otherwise, since $\exists y: N(x, u, y, x')$ holds, \fun{existsTrans}($\hat{x}$, $\hat{u}$, $\hat{x}'$)
returns {\sc True}
and the transition $(\hat{x}, \hat{u}, \hat{x}')$ is added to $\hat{N}$ in line~\ref{update_t.alg.step} 
of Alg.~\ref{symb.alg.rat}.

\item Note that condition~\ref{item:loop} of Def.~\ref{def:ctr-abs} may be
rephrased as follows: if $(\hat{x}, \hat{u}, \hat{x})$ is a non-eliminable self
loop, then $\hat{N}(\hat{x}, \hat{u}, \hat{x})$ must hold. That is, if
$\hat{N}(\hat{x}, \hat{u}, \hat{x}) = 0$ then either there is not a concrete witness for
the self loop $(\hat{x}, \hat{u}, \hat{x})$, or $(\hat{x}, \hat{u}, \hat{x})$ is
an eliminable self loop. This is exactly the case for which function
\fun{selfLoop}(${\cal H}, {\cal Q}, \hat{x}, \hat{u}$)  returns {\sc False} (resp. by
line~\ref{selfLoop:exTrans} of Alg.~\ref{alg:selfLoop} and by
Prop.~\ref{prop:selfLoop}). Since a self loop $(\hat{x}, \hat{u}, \hat{x})$ is
not added to $\hat{N}$ only if \fun{selfLoop}(${\cal H}, {\cal Q}, \hat{x}, \hat{u}$)  returns
{\sc False} in line~\ref{self_loop.alg.step} of Alg.~\ref{symb.alg.rat}, and
since function \fun{selfLoop}(${\cal H}, {\cal Q}, \hat{x}, \hat{u}$) is
eventually invoked for all $\hat{x}\in\Gamma(A_X)$ and 
$\hat{u}\in\Gamma(A_U)$, the thesis is proved.
%
\end{enumerate}
\fussy
\end{proof}

\subsection{Proof of Synthesis Algorithm Correctness}
\label{sec:qks-correctness-proof}

In this section we prove Theor.~\ref{theor:synth-correct}.
This section can be skipped at a first reading.

\begin{proof}[Theorem~\ref{theor:synth-correct}]
If function \fun{qCtrSyn} returns ({\sc Sol}, $\hat{D}$, $\hat{K}$), then
function \fun{minCtrAbs} has found an admissible  ${\cal Q}$ control abstraction
$\hat{\cal M}$ of ${\cal H}$ (see Prop.~\ref{lemma:CtrAbs}) and function
\fun{strongCtr}  has found the strong mgo $\hat{K}$ to the control problem 
($\hat{\cal M}$, $\Gamma(I)$, $\Gamma(G)$). By Theor.~\ref{corol:min-ctr-abs}
(point~\ref{item:minimal}) the controller $K$, defined by
$K(x,u)=(k(\Gamma(x)) = \Gamma(u))$ with $k$ control law for $\hat{K}$, is a  ${\cal Q}$ \gls{QFC} strong solution to
the control problem $({\cal H},I,G)$.

If function \fun{qCtrSyn} returns ({\sc NoSol}, $\hat{D}$, $\hat{K}$), there is
no weak solution to the control problem ($\hat{\cal W}$, $\Gamma(I)$,
$\Gamma(G)$),  where $\hat{\cal W}$ is the close to minimum full control
abstraction of ${\cal H}$ computed  by function \fun{minFullCtrAbs}
(Prop.~\ref{lemma:CtrAbs}). Therefore, by Theor.~\ref{corol:min-ctr-abs}
(point~\ref{item:maximal}) there is no  ${\cal Q}$ \gls{QFC} solution to the control
problem $({\cal H},I,G)$.
\end{proof}

\subsection{Details on Actions Admissibility Check}\label{admissibility.check.subsec}

\sloppy

In this section we show how we can check for action admissibility.
This section can be skipped at a first reading.

In Sect.~\ref{minmax-ctrabs.tex}, for any pair $(\hat{x},\hat{u})$ in
$\Gamma(A_X)\times\Gamma(A_U)$  line~\ref{check_s_u.alg.step} of
Alg.~\ref{ctr-abs.alg} checks if $\hat{u}$ is ${\cal Q}$-admissible in
$\hat{x}$. This check is carried out by determining if the predicate $P(X, U, Y,
X', \hat{x}, \hat{u}) \equiv N(X, U, Y, X')\land \Gamma(X)=\hat{x} \land
\Gamma(U)=\hat{u}      \land X'\not\in A_X$ is not feasible. 

\fussy

Note that $X'\not\in A_X$ is not a conjunctive predicate, thus
feasibility of predicate $P(X$, $U$, $Y$, $X'$, $\hat{x}$, $\hat{u})$ cannot be
directly checked via function \fun{feasible}. We implement such a
check by calling $2|X|$ times function \fun{feasible} in the following
way. For each $x' \in X'$, let $P_{x'}^-(X, U, Y, X', \hat{x},
\hat{u}) \equiv N(X, U, Y, X')\land \Gamma(X)=\hat{x} \land
\Gamma(U)=\hat{u} \land x' \leq \inf A_x$ and $P_{x'}^+(X, U, Y, X',
\hat{x}, \hat{u}) \equiv N(X, U, Y, X')\land \Gamma(X)=\hat{x} \land
\Gamma(U)=\hat{u} \land x' \geq \sup A_x$. For each $x' \in X'$, we
call function \fun{feasible} on $P_{x'}^+$ and $P_{x'}^-$ separately. If all
such $2|X|$ calls return {\sc False}, then $P$ is not feasible,
otherwise $P$ is feasible.

Note that by Def.~\ref{def:ctr-abs} we should also check that $\forall x \in
\Gamma^{-1}(\hat{x})$ $\forall u \in \Gamma^{-1}(\hat{u})$ $\exists x' \in {\cal
D}_X$ $\exists y \in {\cal D}_Y$ $:$ $N(x, u, y, x')$. This cannot be checked
via function \fun{feasible}. We therefore perform such a check by using a tool
for quantifier elimination, namely Mjollnir~\cite{quantifier-elimination-cav10}.
More in detail, we call Mjollnir only once, as a precomputation of
Alg.~\ref{ctr-abs.alg}, on the formula $\Phi(\hat{x}, \hat{u}) \equiv
\exists x \in {\cal D}_X$ $\exists u \in {\cal D}_U$ $\Gamma(X)=\hat{x} \land
\Gamma(U) = \hat{u} \land \neg[\exists x' \in {\cal D}_X$ $\exists y \in {\cal
D}_Y$ $:$ $N(x, u, y, x')]$. The output of Mjollnir is a formula
$\tilde{\Phi}(\hat{x}, \hat{u})$ s.t. $\tilde{\Phi}(\hat{x}, \hat{u}) \equiv
\Phi(\hat{x}, \hat{u})$ and $\tilde{\Phi}(\hat{x}, \hat{u})$ does not contain
quantifiers (i.e., the only variables in $\tilde{\Phi}(\hat{x}, \hat{u})$ are
$\hat{x}$ and $\hat{u}$). $\tilde{\Phi}(\hat{x}, \hat{u})$ is true if $\hat{u}$
is not safe in $\hat{x}$. Since $\tilde{\Phi}(\hat{x}, \hat{u})$ only depends on
bounded discrete variables, we may turn it into an OBDD $\hat{L}$. This is the
last step of the precomputation. Then, we use $\hat{L}$ as follows. Each time
that function \fun{${\cal Q}$-admissible} (line~\ref{check_s_u.alg.step} of
Alg.~\ref{ctr-abs.alg}) is invoked, it first checks if $(\hat{x}, \hat{u})
\in \hat{L}$. If this holds, then function \fun{${\cal Q}$-admissible} directly
returns {\sc False}. Otherwise, the above described check (involving at most $2|X|$
calls to function \fun{feasible}) is performed.



\section{Control Software Generation}
\label{sec:controlSoftware}
\label{obdd2c.tex}

\sloppy

In this section we describe how 
we synthesize the actual control software
(C functions \texttt{Control\_Law} 
and
\texttt{Controllable\_Region} in Sect.~\ref{intro.tex})
and show how we compute its \gls{WCET}. More details are given
in~\cite{icsea2011}.

\fussy

First, we note that given an \gls{OBDD} $B$, we can easily generate
a C function implementation $\mbox{\rm \fun{obdd2c}}(B)$
for the boolean function (defined by) $B$ 
by implementing in C the semantics of \gls{OBDD} $B$.
We do this
by replacing each \gls{OBDD} node with an \texttt{if-then-else} block
and each \gls{OBDD} edge with a \texttt{goto} instruction. 
Let $(\mu,$ $\hat{D},$ $\hat{K})$ be 
the  output of function \fun{qCtrSyn} in Alg.~\ref{qfc-syn-outline}.
We synthesize function \texttt{Controllable\_Region} by computing 
$\mbox{\rm \fun{obdd2c}}(\hat{D})$.
As for function \texttt{Control\_Law},
let $r$ (resp.~$n$) be the number of bits used to
represent plant actions (resp.~states).
We compute~\cite{Tro98} a boolean function $F$ : $\B^n$ $\rightarrow$ $\B^r$
that,
for each quantized state $\hat{x}$ in the controllable region
$\hat{D}$,
returns a quantized action $\hat{u}$ such that $\hat{K}(\hat{x},\hat{u})$ holds. 
Let $F_i$ : $\B^n$ $\rightarrow$ $\B$ be the boolean function computing the $i$-th bit
of $F$. That is, $F(\hat{x})$ = [$F_1(\hat{x}), \ldots, F_r(\hat{x})$].
We take function \texttt{Control\_Law} to be
(the C implementation of) [$\mbox{\rm \fun{obdd2c}}(F_1), \ldots, \mbox{\rm \fun{obdd2c}}(F_r)$]. 


\subsection{Control Software \gls{WCET}}
\label{wcet.tex}

\sloppy

We can easily compute the \gls{WCET}
for our control software.
In fact all \glspl{OBDD} we are considering have at most $n$
variables. Accordingly, the execution of the resulting C code
will go through at most $n$ instruction blocks consisting essentially
of an \texttt{if-then-else} and a \texttt{goto} statement.
Let $T_B$ be the time needed to compute one such a block on the 
microcontroller hosting the control software.
Then we have that 
the \gls{WCET} of \texttt{Controllable\_Region} [\texttt{Control\_Law}]
is less than or equal to $n \cdot T_B$ [$r \cdot n \cdot T_B$].
Thus, neglecting I/O times,
each iteration of the control loop (see Fig. \ref{control-loop-figure.tex})
takes time (control software \gls{WCET}) at most $(r + 1) \cdot n \cdot T_B$.
Note that a more strict upper bound for the \gls{WCET} may be obtained by taking into
account \glspl{OBDD} heights (which are by construction at most $n$). 
The control loop (Fig. \ref{control-loop-figure.tex})
poses the hard real time requirement that the  
control software \gls{WCET} be less than or equal to the
sampling time $T$.
This is the case when \gls{WCET} $\leq$ $T$ holds.
Such an equation allows us to know,
before hand, the realizability of the foreseen control schema.

\fussy








\section{Experimental Results}\label{case_study_synthesis.tex}\label{expres.tex}

We implemented our \gls{QFC} synthesis algorithm 
in C programming language, using
GLPK 
to solve \gls{MILP} problems and the CUDD 
package for \gls{OBDD} based computations.
We name the resulting tool \glsfirst{QKS} (publicly available at \cite{QKS}).



Our methods focus on centralized control software synthesis problems.
Therefore we focus our experimental results on such cases.
Distributed control problems (such as TCAS \cite{PlatzerC09}),
widely studied in a verification setting,
are outside our scopes.

In this section we present our experiments that aim at
evaluating effectiveness of:  
the control abstraction 
generation, 
the synthesis of \gls{OBDD} representation of control law,  
and the control software 
size, performance, and guaranteed operational ranges (i.e. controllable region).
In Sects.~\ref{expres.tex.setting},~\ref{qks-perf.subsec}, and~\ref{sec:controller-performances}
we present results for the {\em buck DC-DC converter} case study.
In Sects.~\ref{expres.tex.setting.invpend},~\ref{qks-perf.subsec.invpend}, and~\ref{sec:controller-performances.invpend}
we shortly outline results for the {\em inverted pendulum} case study.
Note that control software reaction time (\gls{WCET}) is known a priori
from Sect.~\ref{wcet.tex} and its robustness to parameter
variations in the controlled system 
as well as 
enforcement of safety bounds on state variables are an input
to our synthesis algorithm 
(see Ex.~\ref{ex:dths} and Sect.~\ref{expres.tex.setting}).





\begin{table}
  \centering
  \small
  \caption{Buck DC-DC converter (Sect.~\ref{dths.tex}): control abstraction \& controller synthesis results. Part I.\label{expres-table.tex:1}}{%
  \begin{tabular}{cccccccc}
    \toprule
    & \multicolumn{5}{c}{Control Abstraction} & \multicolumn{2}{c}{Controller Synthesis}\\
    \cmidrule(r){2-6}\cmidrule(r){7-8} $b$ & CPU & MEM & Arcs & MaxLoops & LoopFrac & CPU & $|K|$\\
    \midrule
8 & 1.95e+03 & 4.41e+07 & 6.87e+05 & 2.55e+04 & 0.00333 & 2.10e-01 & 1.39e+02\\
9 & 9.55e+03 & 5.67e+07 & 3.91e+06 & 1.87e+04 & 0.00440 & 2.64e+01 & 3.24e+03\\
10 & 1.42e+05 & 8.47e+07 & 2.61e+07 & 2.09e+04 & 0.00781 & 7.36e+01 & 1.05e+04\\
11 & 8.76e+05 & 1.11e+08 & 2.15e+08 & 2.26e+04 & 0.01435 & 2.94e+02 & 2.88e+04\\
    \bottomrule
  \end{tabular}}
\end{table}

\begin{table}
  \centering
  \caption{Buck DC-DC converter (Sect.~\ref{dths.tex}): control abstraction \& controller synthesis results. Part II.\label{expres-table.tex:2}}{%
  \begin{tabular}{cccc}
    \toprule
    & \multicolumn{3}{c}{Total}\\
    \cmidrule(r){2-4} $b$ & CPU & MEM & $\mu$\\
    \midrule
8 & 1.96e+03 & 4.46e+07 & {\sc Unk}\\
9 & 9.58e+03 & 7.19e+07 & {\sc Sol}\\
10 & 1.42e+05 & 1.06e+08 & {\sc Sol}\\
11 & 8.76e+05 & 2.47e+08 & {\sc Sol}\\
    \bottomrule
  \end{tabular}}
\end{table}

\subsection{Buck DC-DC Converter: Experimental Settings}
\label{expres.tex.setting}

In this section (and in Sects.~\ref{qks-perf.subsec},~\ref{sec:controller-performances}) 
we present experimental results obtained by using \gls{QKS} on a version of the buck
DC-DC converter described in Sect.~\ref{example-buck.tex}.
Further case studies (namely, the inverted pendulum and the multi-input buck DC-DC converter) 
can be found in~\cite{cdc12} and~\cite{emsoft12}.
We denote with ${\cal H} = (X, U, \tilde{Y}, \tilde{N})$ the \gls{DTLHS} modeling such a
converter, where $X, U$ are as in Sect.~\ref{example-buck.tex}.
We set
the parameters of ${\cal H}$
as
follows:
 $T = 10^{-6}$ secs, 
 $L = 2 \cdot 10^{-4}$ H,
 $r_L = 0.1$ ${\rm \Omega}$,
 $r_C = 0.1$ ${\rm \Omega}$, 
 $R = 5\pm25\%$ ${\rm \Omega}$, 
 $R_{off} = 10^4$ ${\rm \Omega}$, 
 $C = 5 \cdot 10^{-5}$ F, 
 $V_i = 15\pm25\%$ V. Thus, we
 require our controller to be robust
to foreseen variations (25\%) in the load ($R$) and 
in the power supply ($V_i$). To this aim, $\tilde{N}$ is obtained by extending
$N$ of Sect.~\ref{example-buck.tex} as follows. As for variations in the power
supply $V_i$, they are modeled analogously to Ex.~\ref{ex:dths}. As for
variations in the load $R$, much more work is needed~\cite{buck-tekrep-art-2011} since ${\cal H}$ dynamics is not linear in $R$.
For the sake of brevity, we simply point out that modeling variations in the load $R$ 
requires 11 auxiliary boolean variables to be added 
to $Y$, thus obtaining $\tilde{Y}$, and 15 (guarded) constraints to be added to
$\tilde{N}$~\cite{buck-tekrep-art-2011}.
%


%
%

For converters, \emph{safety} (as well as physical) considerations
set requirements on admissible values for state variables (admissible regions). 
We set $A_{i_L} = [-4, 4]$ and $A_{v_O} = [-1, 7]$.
We define $A = A_{i_L} \times A_{v_O} \times A_u$.
As for auxiliary variables, we use the following safety bounds:
$A_{i_u} = A_{i_D} = [-10^3, 10^3]$ and  $A_{v_u} = A_{v_D}
= [-10^7, 10^7]$. As a result, we add 12 further constraints to $\tilde{N}$ stating
that $\bigwedge_{w \in \{i_L, v_O, i_u, i_D, v_u, v_D\}} w \in
A_w$, thus obtaining a bounded \gls{DTLHS}~\cite{buck-tekrep-art-2011}. 

Finally, the initial region
$I$ 
and goal region $G$ are as in Ex. \ref{example-goal}, thus
the \gls{DTLHS} control problem we consider is 
$P$ = (${\cal H}$, $I$, $G$). 
%
Note that no (formally proved) robust control software 
is available for buck DC-DC converters. 
%
%
%
%
%
%



We use a uniform quantization dividing the domain
of each state variable ($i_L, v_O$) into $2^b$ equal
intervals, where $b$ is the number of bits used by \gls{AD} conversion, thus w.r.t.
Sect.~\ref{dths-ctr-abs.tex} we have that $\Delta_{i_L} = \Delta_{v_O} = 2^b$. 
The resulting quantization is ${\cal Q}_b = (A, \Gamma_b)$,
with $\|\Gamma_b\| = 
2^{3-b}$. Since we have two
quantized variables ($i_L, v_O$) each one with $b$ bits,  the number of states
in the control abstraction is exactly $2^{2b}$.  

For each value of interest for $b$, we run \gls{QKS}, and thus
Alg.~\ref{qfc-syn-outline}, on the control problem $({\cal H}, I, G)$ with
quantization ${\cal Q}_b$. In the following, we will call $\hat{\cal M}_{b}$ 
the close to minimum (admissible) ${\cal Q}_b$ control abstraction for ${\cal H}$, $\hat{\cal
H}_{b}$  the maximum (full) ${\cal Q}_b$ control abstraction for ${\cal H}$ (which
we compute for statistical reasons also when Alg.~\ref{qfc-syn-outline} returns {\sc Sol}),
$\hat{K}_{b}$ the strong mgo for  $\hat{\cal P}_{b} = (\hat{\cal M}_{b}$,
$\varnothing$, $\Gamma_b(G))$, $\hat{D}_{b} = \mbox{\rm Dom}(\hat{K}_b)$ the
controllable region of $\hat{K}_{b}$, and $K_b(s, u) = \hat{K}_b(\Gamma_b(s),
\Gamma_b(u))$ the ${\cal Q}_b$  \gls{QFC} solution  to ${\cal
P}_{b} = ({\cal H}$, $\Gamma_b^{-1}(\hat{D}_{b})$, $G)$.  
All our experiments have been
carried out on a 3.0 GHz Intel hyperthreaded Quad Core Linux PC with 8 GB of
RAM.







\subsection{Buck DC-DC Converter: \gls{QKS} Performance}\label{qks-perf.subsec}

\begin{table}
  \centering
  \caption{Buck DC-DC converter: number of \glspl{MILP} and time to solve them (secs). Part I.\label{expres-table.tras.milps.tex:1}}{%
  \begin{tabular}{*{7}{c}}
    \toprule
    & \multicolumn{3}{c}{$b=8$} & \multicolumn{3}{c}{$b=9$}\\
    \cmidrule(r){2-4}\cmidrule(r){5-7}
    \gls{MILP} & Num & Avg & Time & Num & Avg & Time\\
    \midrule
1 & 6.6e+04 & 7.0e-05 & 4.6e+00 & 2.6e+05 & 7.0e-05 & 1.8e+01\\
2 & 4.0e+05 & 1.5e-03 & 3.3e+02 & 1.6e+06 & 1.4e-03 & 1.1e+03\\
3 & 2.3e+05 & 9.1e-04 & 2.1e+02 & 9.2e+05 & 9.2e-04 & 8.4e+02\\
4 & 7.8e+05 & 9.9e-04 & 7.7e+02 & 4.4e+06 & 1.0e-03 & 4.5e+03\\
5 & 4.3e+05 & 2.8e-04 & 1.2e+02 & 1.7e+06 & 2.8e-04 & 4.9e+02\\
    \bottomrule
  \end{tabular}}
\end{table}

\begin{table}
  \centering
  \caption{Buck DC-DC converter: number of \glspl{MILP} and time to solve them (secs). Part II.\label{expres-table.tras.milps.tex:2}}{%
  \begin{tabular}{*{7}{c}}
    \toprule
    & \multicolumn{3}{c}{$b=10$} & \multicolumn{3}{c}{$b=11$}\\
    \cmidrule(r){2-4}\cmidrule(r){5-7}
    \gls{MILP} & Num & Avg & Time & Num & Avg & Time \\
    \midrule
1 & 1.0e+06 & 2.7e-04 & 2.8e+02 & 4.2e+06 & 2.3e-04 & 9.7e+02\\
2 & 6.4e+06 & 3.8e-03 & 1.3e+04 & 2.5e+07 & 3.3e-03 & 4.6e+04\\
3 & 3.7e+06 & 3.0e-03 & 1.1e+04 & 1.5e+07 & 2.6e-03 & 3.8e+04\\
4 & 3.0e+07 & 2.6e-03 & 7.8e+04 & 2.6e+08 & 2.2e-03 & 5.7e+05\\
5 & 6.8e+06 & 1.8e-03 & 1.3e+04 & 2.7e+07 & 1.6e-03 & 4.2e+04\\
    \bottomrule
  \end{tabular}}
\end{table}

In this section we will show the performance (in terms of computation time and
memory) of algorithms discussed in Sect.~\ref{ctr-syn-algorithm.tex}. 

Tabs.~\ref{expres-table.tex:1}, \ref{expres-table.tex:2}, \ref{expres-table.tras.milps.tex:1} and~\ref{expres-table.tras.milps.tex:2}
show our
experimental results for \gls{QKS}  (and thus for Alg.~\ref{qfc-syn-outline}).
Columns in Tab.~\ref{expres-table.tex:1} have the following meaning. 
Column $b$ shows the number of \gls{AD} bits. 
Columns labeled {\em Control Abstraction} show
performance for Alg.~\ref{ctr-abs.alg} (computation of $\hat{\cal M}_b$)
and they show running time (column {\em CPU}, in secs), 
memory usage ({\em MEM}, in bytes), 
the number of transitions in $\hat{\cal M}_b$ ({\em Arcs}),
the number of self loops in $\hat{\cal H}_b$ ({\em MaxLoops}), 
and the fraction of self loops that are kept in $\hat{\cal M}_b$
w.r.t. the number of self loops in $\hat{\cal H}_b$ ({\em LoopFrac}).
%
Columns labeled {\em Controller Synthesis}
show the computation time (column {\em CPU}, in secs) for
the generation of $\hat{K}_b$,
and the size of its \gls{OBDD} representation ($|K|$, number of nodes).
The latter is also 
the size (number of lines) of ${\hat{K}_b}$ C code 
synthesized implementation. 
%
Columns in Tab.~\ref{expres-table.tex:2} have the following meaning.
Column $b$ shows the number of \gls{AD} bits. 
Columns labeled {\em Total} show the total
computation time (column {\em CPU}, in secs) 
and the memory ({\em MEM}, in bytes) 
for the whole process (i.e., control abstraction plus controller source code
generation), as well as the final outcome $\mu\in\{${\sc Sol}, {\sc NoSol}, {\sc
Unk}$\}$ of Alg.~\ref{qfc-syn-outline}. 

From Tabs.~\ref{expres-table.tex:1} and~\ref{expres-table.tex:2} we see that computing control
abstractions (i.e. Alg.~\ref{ctr-abs.alg}) 
is the most expensive operation in 
\gls{QKS} 
and that
thanks to function \fun{SelfLoop} $\hat{\cal M}_b$ contains no more than
2\% of the loops in $\hat{\cal H}_b$.





\subsubsection{\gls{MILP} problems Analysis}
\label{sec:milp-analysis}

\begin{figure*}
  \centering
  \begin{tabular}{ccc}
  \begin{minipage}{0.3\textwidth}
  \includegraphics[width=0.9\textwidth]{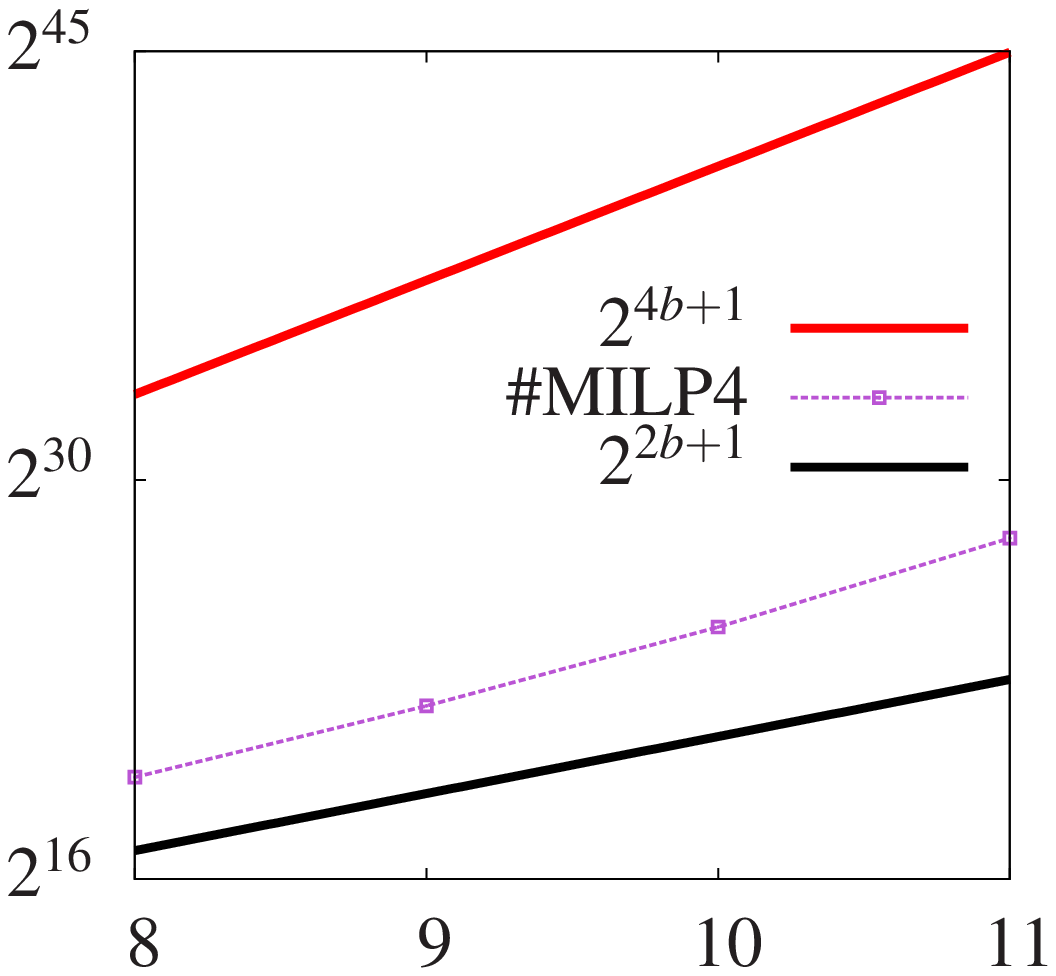}
  \caption{Buck DC-DC Converter: Number of MILP4 calls.}
  \label{bits-time.eps}
  \end{minipage}
  &
  \begin{minipage}{0.3\textwidth}
  \includegraphics[width=1\textwidth]{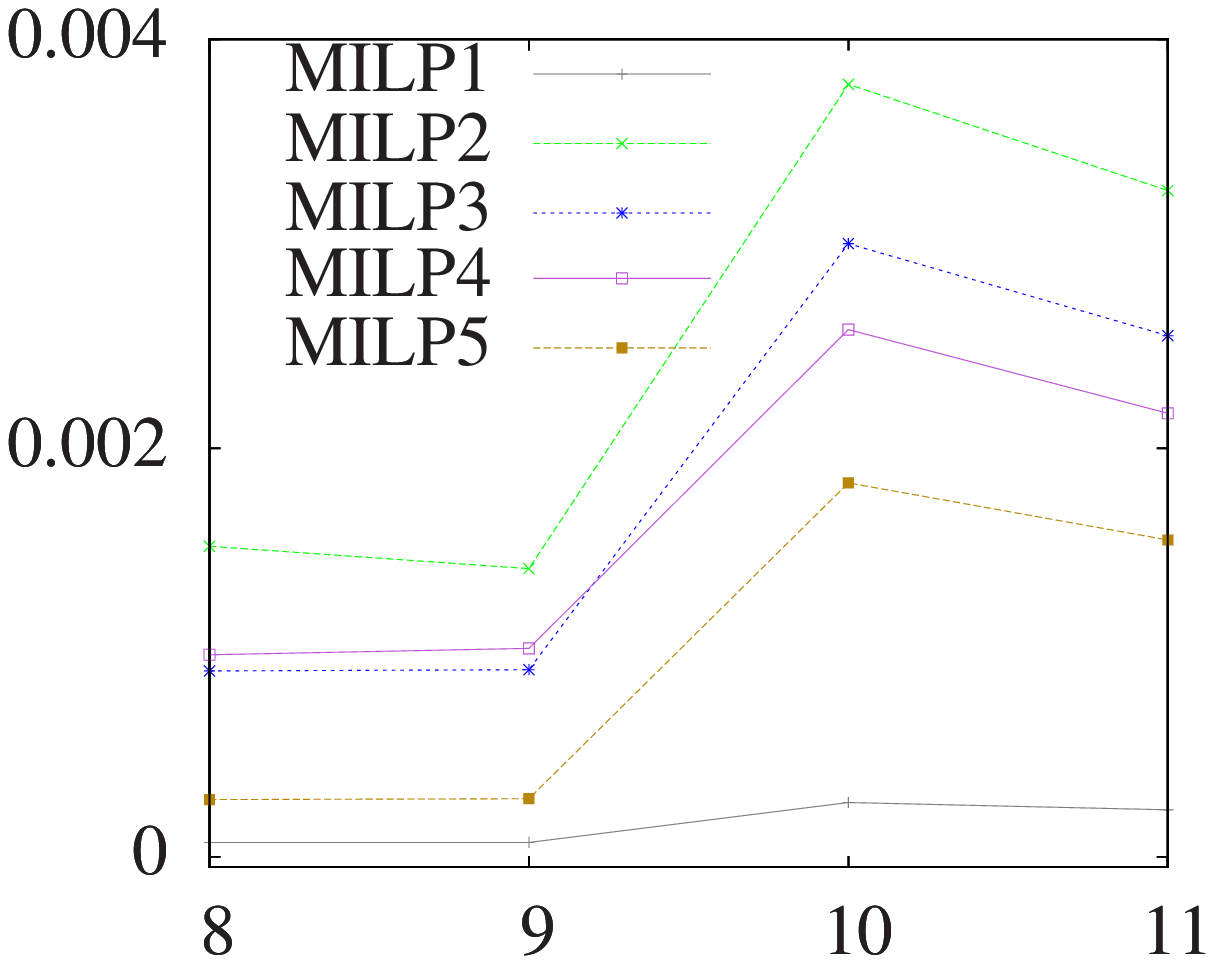}
  \caption{Buck DC-DC Converter: Average of \gls{MILP} calls.}
  \label{bits-avg.eps}
  \end{minipage}
  &
  \begin{minipage}{0.35\textwidth}
  \includegraphics[width=1\textwidth]{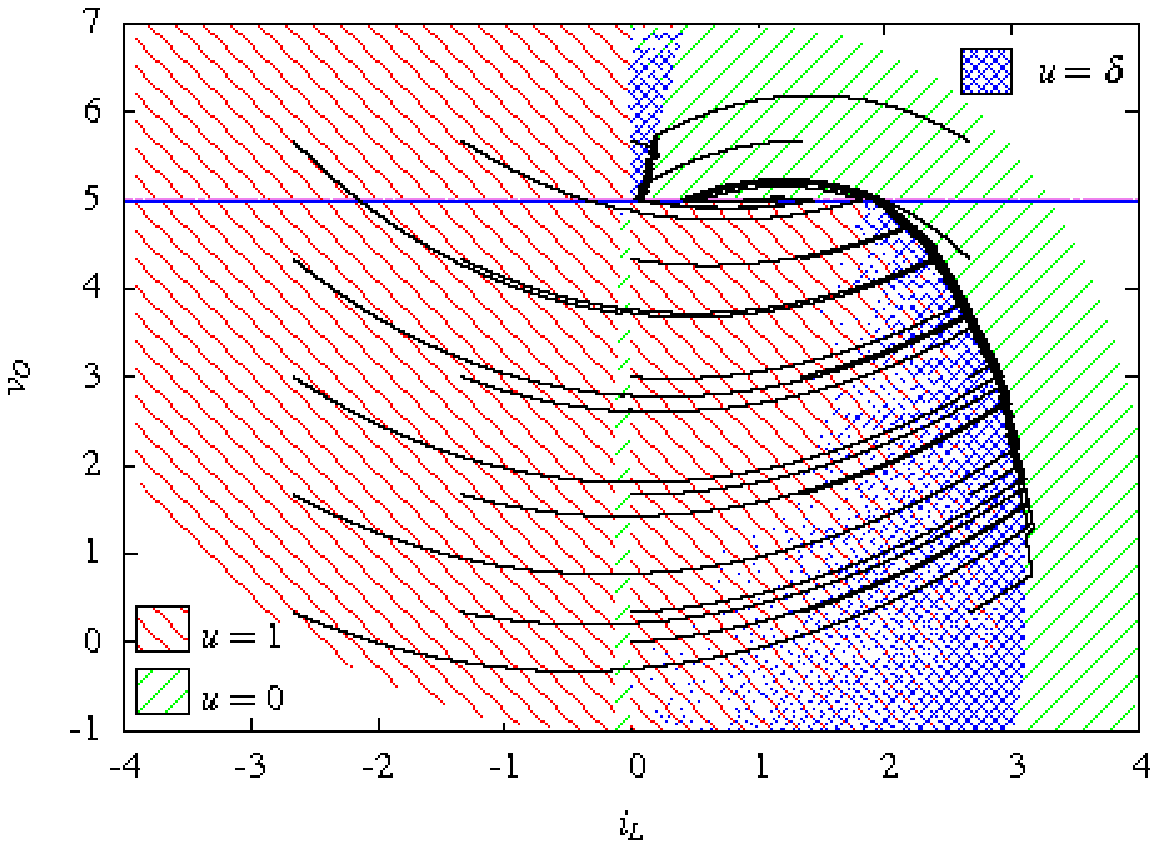}
  \caption{Buck DC-DC Converter: Controlled region with $b=10$ bits.}
  \label{controlled.region.10.eps}
  \end{minipage}
  \end{tabular}
\end{figure*}

%
%
%
%
%
%
%
%
%
%
%
%
%
%
%
%
%

For each \gls{MILP} problem solved in \gls{QKS}, 
Tabs.~\ref{expres-table.tras.milps.tex:1} and~\ref{expres-table.tras.milps.tex:2}
show (as a function of $b$)
the total and the average CPU time (in seconds) spent 
solving \gls{MILP} problems, together with the number of \gls{MILP} problems solved, divided
by different kinds of \gls{MILP} problems as follows.
MILP1 refers to the \gls{MILP} problems described in Sect.~\ref{dths-ctr-abs.tex}, i.e. those
computing the quantization for $I$ and $G$, 
MILP2 refers to \gls{MILP} problems in function \fun{SelfLoop} (see Alg.~\ref{alg:selfLoop}), 
MILP3 refers to the \gls{MILP} problems 
used in function \fun{overImg} (line~\ref{overimg.alg.step} of
Alg.~\ref{ctr-abs.alg}), 
MILP4 refers to \gls{MILP} problems used to check actions admissibility (line~\ref{check_s_s_prime.alg.step} of
Alg.~\ref{ctr-abs.alg}), and 
MILP5 refers to \gls{MILP} problems used to check transitions witnesses (line~\ref{check_s_u.alg.step} of
Alg.~\ref{ctr-abs.alg}).
Columns in Tabs.~\ref{expres-table.tras.milps.tex:1} and~\ref{expres-table.tras.milps.tex:2} 
have the following meaning: 
{\em Num} is the number of times that the \gls{MILP} problem of the given
type is called, {\em Time} is the total CPU time (in secs) needed to solve all the {\em
Num} instances of the \gls{MILP} problem of the given type, and {\em Avg} is the average
CPU time (in secs), i.e. the ratio between columns {\em Time} and {\em Num}
 \rimandotekrep{expres-details.tex}.

CPU time standard deviation is always less than 0.003 \rimandotekrep{expres-details.tex}.

Fig.~\ref{bits-time.eps} graphically shows 
(as a function of $b$)
the number of MILP4 instances solved (column {\em Num} of columns group 
MILP4 in Tabs.~\ref{expres-table.tras.milps.tex:1} and~\ref{expres-table.tras.milps.tex:2}).

From Tabs.~\ref{expres-table.tras.milps.tex:1} and~\ref{expres-table.tras.milps.tex:2}, column {\em Avg}, 
we see that the average time spent solving each \gls{MILP} instance
is small. Fig.~\ref{bits-avg.eps} graphically shows that \gls{MILP} average computation time 
does not heavily depend on $b$.
%
%
As observed in Remark \ref{explicit-loops-remark},  
Fig.~\ref{bits-time.eps} shows that the number of MILP4 invocations is 
much closer to 
$|\Gamma(A_X)||\Gamma(A_U)|$ = $2^{2b + 1}$, 
rather than the theoretical worst case running time $|\Gamma(A_X)|^2|\Gamma(A_U)|$ = $2^{4b + 1}$
of Alg.~\ref{ctr-abs.alg}.
This shows effectiveness of function \fun{overImg} heuristic.
\subsection{Buck DC-DC Converter: Control Software Performance}
\label{sec:controller-performances}

In this section we discuss the performance of the generated controller. 
Fig.
~\ref{generatedcodesnapshot.c.tex}
shows a snapshot of the \gls{QKS} synthesized 
control software for the Buck DC-DC converter 
when 10 bits
($b = 10$) are used for \gls{AD} conversion.


\begin{figure}[htbp]
  \framebox[1.0\hsize][c]{
    \begin{minipage}{0.9\hsize}
      \begin{center}
        \small
        \input{generatedcodesnapshot.c.tex}
      \end{center}
    \end{minipage}
  }
  \caption{A snapshot of the synthesized control software for the Buck DC-DC converter with 10 bit \gls{AD} conversion.} \label{generatedcodesnapshot.c.tex}
\end{figure}



\subsubsection{Controllable Region}\label{controllable-region.subsec}

One of the most important features of our approach is that it returns 
the guaranteed operational range (precondition) 
of the synthesized software (Theor. \ref{theor:synth-correct}). 
This is the controllable region $\hat{D}$ returned by 
Alg.~\ref{qfc-syn-outline}.
In our case study, 9 bit turns out to be enough to have 
a controllable region that covers the initial region~\cite{buck-tekrep-art-2011}.
Increasing the number of bits, we obtain even larger controllable regions.
Fig.~\ref{controlled.region.10.eps} shows the
controllable region $D_{10} = \Gamma^{-1}_{10}(\hat{D}_{10})$ for ${K}_{10}$ along with some trajectories
(with time increasing counterclockwise) for the closed loop system.
We see that the initial region $I \subseteq D_{10}$. Thus we know (on a
formal ground) that 10 bit \gls{AD} 
conversion suffices for our purposes. 
More details on controllable region visualization can be found in \cite{INFOCOMP2012}.
\begin{center}
  \begin{figure}[t]
    \centering
    \subfigure[$v_O$ from $i_L = 0, v_O = 0$]
    {
      \centering
      \includegraphics[width=0.32\textwidth, angle=-90]{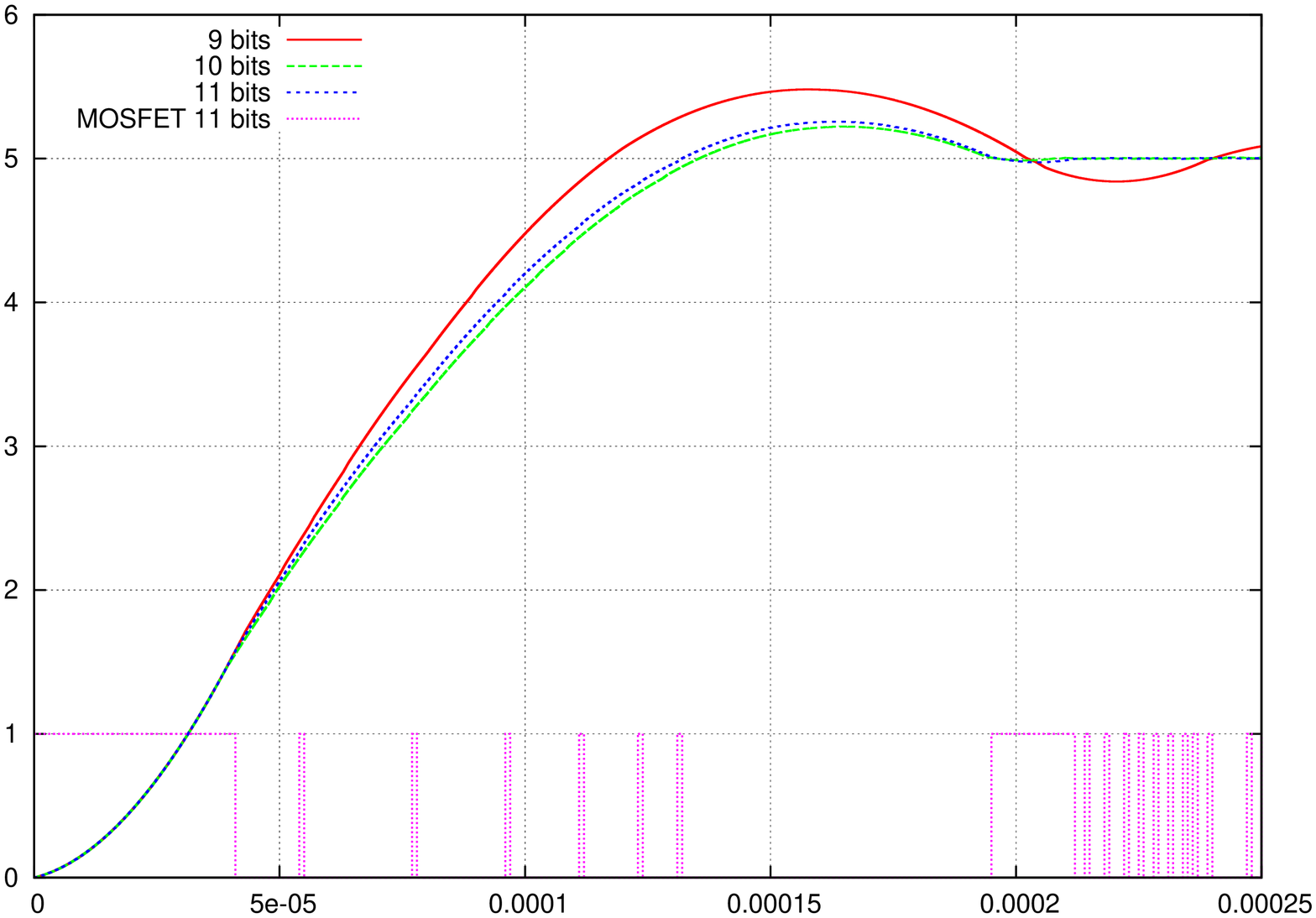}
      \label{load.eps}
    }
    \subfigure[{\em Ripple} for $v_O$ ($b = 11$)] 
    {
      \centering
      \includegraphics[width=0.32\textwidth, angle=-90]{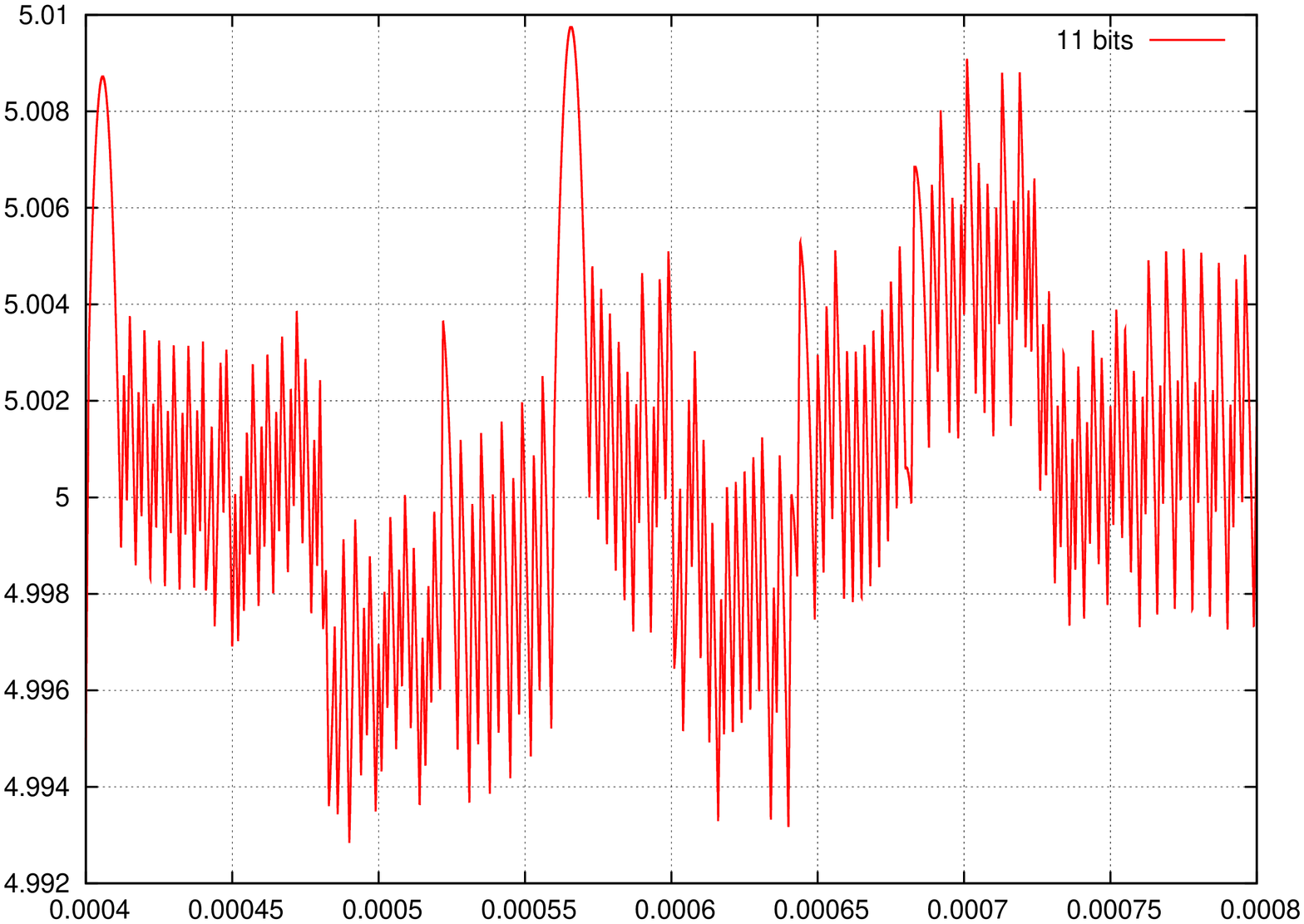}
      \label{ripple.11.eps}
    }
    \caption{Controller performances for the Buck DC-DC Converter: setup time and ripple.}
    \label{fig:performances}
  \end{figure}
\end{center}


\subsubsection{Setup Time and Ripple}
\label{startup-ripple.subsec}

Our model based control software synthesis approach
presently does not handle \emph{quantitative liveness specifications}.
Accordingly, quantitative system level formal specifications have to be
verified a posteriori. This can be done using a classical
\newacronym{HIL}{HIL}{Hardware-In-the-Loop}%
\glspl{HIL}
simulation approach or, 
even better, following a formal approach, as discussed in
\cite{henzingder-popl2010,kim-larsen-hscc2010}.
In our context \gls{HIL} simulation is quite easy since we already have
a \gls{DTLHS} model for the plant and the control software is 
generated automatically.

To illustrate such a point in this section we highlight \gls{HIL} simulation results
for two quantitative specifications typically considered in control
systems: \emph{Setup Time} and \emph{Ripple}.

The setup time measures the time it takes to reach the goal (steady state)
when the system is turned on. 
Fig.~\ref{load.eps} shows trajectories starting from point $(0, 0)$ for
$K_{9}$, $K_{10}$ and $K_{11}$ as well as the 
control command sent to the MOSFET (square wave in Fig.~\ref{load.eps}) for $K_{11}$. Note
that  all trajectories stabilize (steady state) after only $0.0003$ secs (setup time).

The ripple measures the wideness of the oscillations around the goal (steady state)
once this has been reached. 
Fig.~\ref{ripple.11.eps} shows the ripple for the output voltage after
stabilization.  
For $K_{11}$ we see that
the ripple is about $0.01$ V, that is $0.2\%$ of the reference value $V_{\mathrm{ref}}
= 5$ V. 

\sloppy

It is worth noticing that both setup time and ripple
compare well with typical figures of commercial high-end buck
DC-DC converters  
(e.g. see 
\cite{texas-instruments-buck-dc-dc}) and with the
results available from the literature (e.g.  
\cite{fuzzy-dc-dc-1996,time-optimal-dc-dc-2008}).

\fussy


%
%


\sloppy
\subsection{Inverted Pendulum: Experimental Settings}\label{expres.tex.setting.invpend}

\begin{figure*}[!b]
  \centering
  \begin{tabular}{cc}
  \begin{minipage}{0.27\textwidth}
  \includegraphics[width=\textwidth]{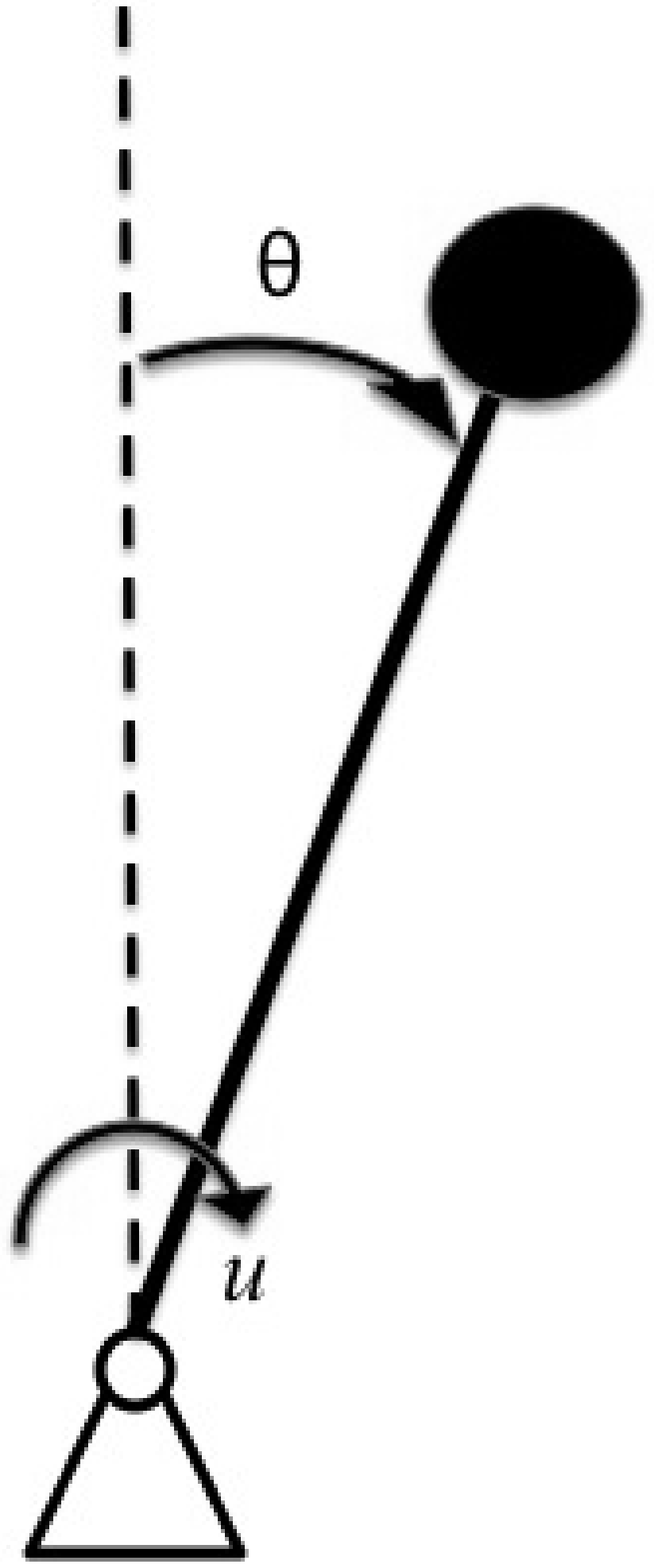}
  \caption{Inverted Pendulum with Stationary Pivot Point.}
  \label{fig:invpend}
  \end{minipage}
  &
  \begin{minipage}{0.5\textwidth}
  \includegraphics[width=\textwidth]{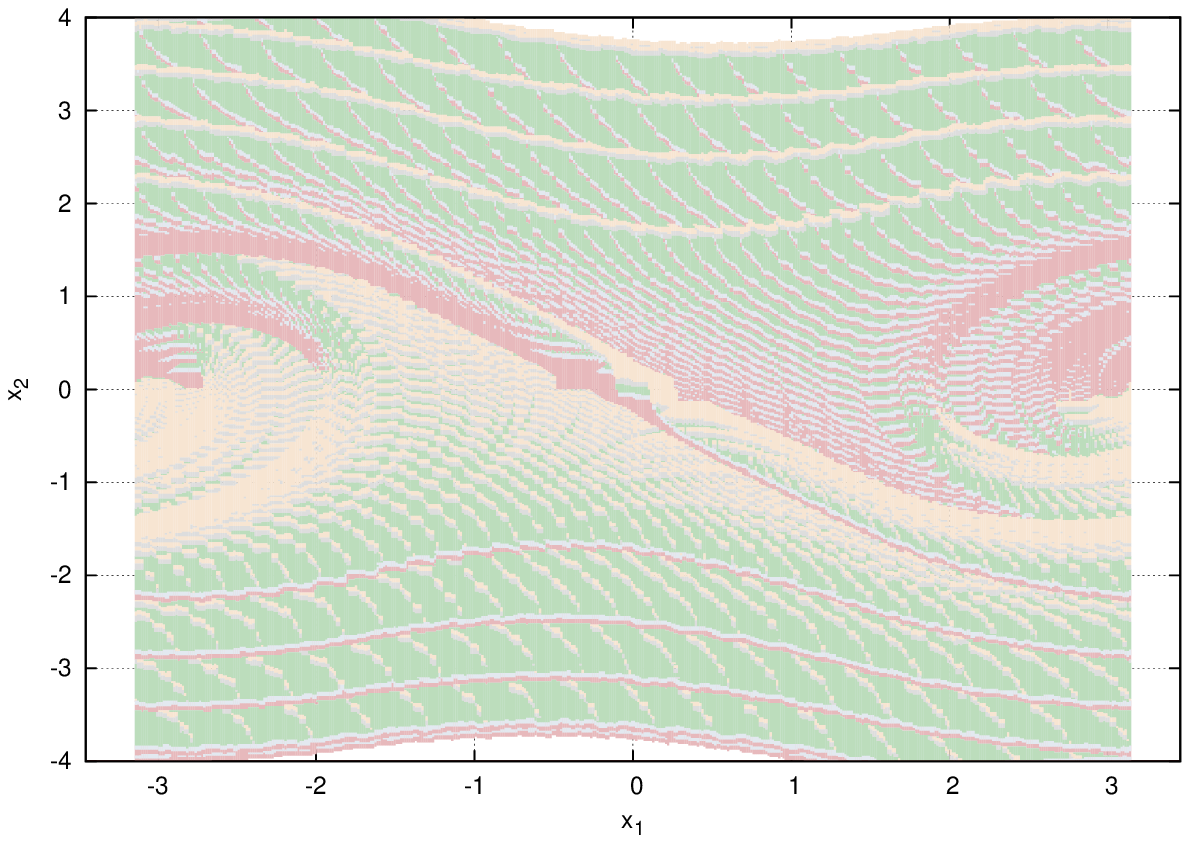}
  \caption{Inverted Pendulum: Controlled region with $b=9$ bits, $F=0.5$, $T=0.1$.}
  \label{fig:controllable-region-05}
  \end{minipage}
  \end{tabular}
\end{figure*}


In this section (and in Sects.~\ref{qks-perf.subsec.invpend},~\ref{sec:controller-performances.invpend}) 
we present experiment results obtained by using \gls{QKS} 
on the inverted pendulum described in \cite{KB94}, as shown in Fig. \ref{fig:invpend}.
The system is modeled by taking the angle $\theta$ and the angular velocity $\dot{\theta}$ as
state variables. The input of the system is the torquing force $u$, that can influence
the velocity in both directions.
Moreover, the behaviour of the system depends on the pendulum mass $m$, the length of the
pendulum $l$ and the gravitational acceleration $g$. Given such parameters,
the motion of the system is described by the differential equation
$\ddot{\theta} = {g \over l} \sin \theta + {1 \over m l^2} u$.

In order to obtain a state space representation, we consider the following normalized system,
where $x_1$ is the angle $\theta$ and $x_2$ is the angular speed $\dot{\theta}$.
\begin{eqnarray}
 \dot{x}_1 & = & x_2  \label{eq:pendmotion2.1} \\
 \dot{x}_2 & = & {g \over l} \sin x_1 + {1 \over m l^2} u \label{eq:pendmotion2.2} 
\end{eqnarray}

Differently from~\cite{KB94}, we consider the problem of finding a discrete controller, whose decisions may be ``apply the force clockwise'' ($u=1$), ``apply the force counterclockwise'' ($u=-1$), or ``do nothing'' ($u=0$).
The intensity of the force will be given as a constant $F$. Finally, the discrete time
transition relation $N$ is obtained from the equations (\ref{eq:pendmotion2.1}-\ref{eq:pendmotion2.2})
as the Euler approximation with sampling time $T$, i.e. the
predicate $(x'_1 = x_1 + T x_2)\,\land\,(x'_2 = x_2 + T {g \over l} \sin x_1 + T {1 \over m l^2} F u)$.

Since the system whose dynamics are in equations (\ref{eq:pendmotion2.1}-\ref{eq:pendmotion2.2}) is not linear,
we build a linear over-approximation of it as shown in \cite{cdc12}.
The result is the DTLHS ${\cal H}$ defined in Ex.~5 of \cite{cdc12}.
From now on we use ${\cal H}$ to denote the inverted pendulum system.

In all our experiments, as in~\cite{KB94},
we set parameters $l$ and $m$ in such a way that $\frac{g}{l}=1$ (i.e. $l=g$) and
$\frac{1}{ml^2}=1$ (i.e. $m=\frac{1}{l^2}$). Moreover, we set the force intensity $F=0.5$.
More experiments can be found in \cite{cdc12}.

As we have done for the buck DC-DC converter, we use uniform quantization functions dividing the domain of each state
variable ${\cal D}_{x_1}=[-1.1\pi, 1.1\pi]$ (we write $\pi$ for a rational approximation of it) and  ${\cal D}_{x_2}=[-4, 4]$
into $2^b$ equal intervals, where $b$ is the number of
bits used by AD conversion.
Since we have two quantized
variables, each one with $b$ bits,  the number of quantized
states 
is exactly $2^{2b}$.

The typical goal for the inverted pendulum is to turn the pendulum steady
to the upright position, starting from any possible initial position, within a given speed interval.
In our experiments, the goal region is defined by the predicate
$G(X)\equiv (-\rho\leq x_1\leq\rho)\,\land\,(-\rho\leq x_2\leq\rho)$,
where $\rho\in\{0.05, 0.1\}$, and the initial
region is defined by the predicate
$I(X)\equiv(-\pi\leq x_1\leq \pi)\,\land\,(-4\leq x_2\leq 4$).

We run \gls{QKS} on the control problem $({\cal H}, I, G)$
for different values of the remaining parameters, i.e. $\rho$ (goal tolerance),
$T$ (sampling time), and $b$ (number of bits of AD).
For each of such experiments, \gls{QKS} outputs a control software $K$
in C language. In the following, we sometimes make explicit the dependence on
$b$ by writing $K_{b}$. In order to evaluate performance of $K$,
we use an  {\em inverted pendulum  simulator} written in C. The simulator
computes the next state by using Eqs.~(\ref{eq:pendmotion2.1}-\ref{eq:pendmotion2.2}),
thus simulating a path of ${\cal H}^{(K)}$.
Such simulator also
introduces
random disturbances (up to 4\%) in the next state computation
to assess $K$ robustness w.r.t. non-modeled disturbances.
Finally, in the simulator
Eqs.~(\ref{eq:pendmotion2.1}-\ref{eq:pendmotion2.2}) are translated into the discrete time
version by means of a simulation time step $T_s$ much smaller than the sampling
time $T$ used in ${\cal H}$. Namely, $T_s = 10^{-6}$
seconds, whilst $T = 0.01$ or $T = 0.1$ seconds. This allows us to have a more
accurate simulation. Accordingly, $K$ is called each $10^4$ (or $10^5$)
simulation steps of ${\cal H}$. When $K$ is not called, the last chosen action
is selected again ({\em sampling and holding}).

All experiments for the inverted pendulum have been carried out on an Intel(R) Xeon(R) CPU @ 2.27GHz,
with 23GiB of RAM, 
Debian
GNU/Linux 6.0.3 (squeeze).

\begin{table}[!tb]
  \small
  \centering
  \caption{Inverted Pendulum: control abstraction \& controller synthesis results with $F=0.5$.\label{expres.invpend.table.tex}}{%
  \begin{tabular}{ccc|ccc}
    \toprule
    $b$ & $T$ & $\rho$ & $|K|$ & CPU & MEM\\
    \midrule

8 & 0.1 & 0.1 & 2.73e+04 & 2.56e+03 & 7.72e+04 \\
9 & 0.1 & 0.1 & 5.94e+04 & 1.13e+04 & 1.10e+05 \\
10 & 0.1 & 0.1 & 1.27e+05 & 5.39e+04 & 1.97e+05 \\
11 & 0.01 & 0.05 & 4.12e+05 & 1.47e+05 & 2.94e+05 \\

    \bottomrule
  \end{tabular}}
\end{table}






\subsection{Inverted Pendulum: \gls{QKS} Performance}\label{qks-perf.subsec.invpend}
To stabilize an {\em underactuated} inverted pendulum (i.e. $F<
1$) from the hanging position to the upright position, a controller  needs to
find a non obvious strategy that consists of swinging the pendulum once or more
times  to gain enough momentum. 
\gls{QKS} is able to synthesize such a
controller taking as input ${\cal H}$ with $F=0.5$ (note that
in~\cite{KB94} $F=0.7$). Results are in Tab.~\ref{expres.invpend.table.tex}, where each
row corresponds to a \gls{QKS} run,
columns $b$, $T$ and $\rho$ show the corresponding inverted
pendulum parameters, column $|K|$ shows the size of the C code for
$K_{b}$, and columns CPU and MEM show the computation
time (in seconds) and RAM usage (in KB) needed by \gls{QKS} to
synthesize $K_{b}$.

\begin{center}
  \begin{figure}
    \centering
    \subfigure[$x_1$ from $x_1=\pi$, $x_2=0$]
    {
      \centering
      \includegraphics[width=0.47\textwidth]{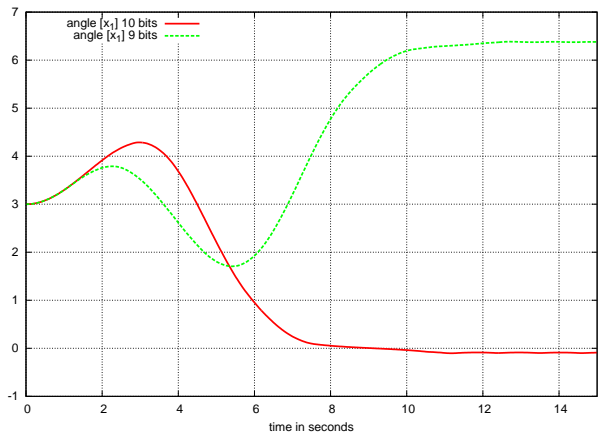}
      \label{fig:trajectories-05}
    }
    \subfigure[Ripple for $x_1$ ($b=10$)]
    {
      \centering
      \includegraphics[width=0.47\textwidth]{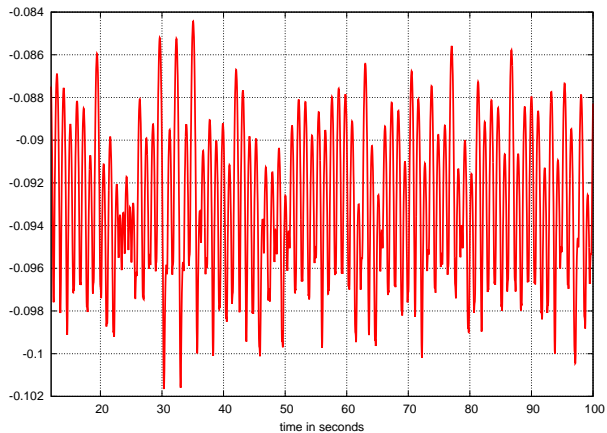}
      \label{fig:ripple-05}
    }
    \caption{Controller performances for the Inverted Pendulum with $F=0.5$, $b=9, 10$: setup time and ripple.}
    \label{fig:performances.invpend}
  \end{figure}
\end{center}

\subsection{Inverted Pendulum: Control Software Performance}\label{sec:controller-performances.invpend}

As for $K_{b}$ performance, it is easy to show that by reducing the
sampling time $T$ and the quantization step (i.e. increasing $b$), we increase
the quality of $K_{b}$ in terms of ripple and set-up time. 
Fig.~\ref{fig:trajectories-05} shows the simulations of ${\cal
H}^{(K_{9})}$ and ${\cal H}^{(K_{10})}$.  As we can see,
$K_{10}$ drives the system to the goal with a smarter trajectory,  with
one swing only. This have a significant impact on the set-up time  (the system
stabilizes after about $8$ seconds when controlled by $K_{10}$ instead of about  $10$
seconds required when controlled by $K_{9}$).
Fig.~\ref{fig:controllable-region-05} shows that
the controllable region of $K_{9}$
covers almost all
states in the admissible region that we consider. Different colors mean different
set of actions enabled by the controller.
Finally,
Fig.~\ref{fig:ripple-05} shows the ripple of $x_1$ for ${\cal H}^{(K_{10})}$
inside the goal. Note that such ripple is very low (0.018 radiants).

\fussy


\sloppy

\section{Related Work}
\label{related-works.tex}


\newacronym{LHA}{LHA}{Linear Hybrid Automaton}%
\newacronym{TA}{TA}{Timed Automaton}%
\newacronym{PWA-DTHS}{PWA-DTHS}{Piecewise Affine Discrete Time Hybrid Systems}%

This paper is a journal version of \cite{qks-cav10} which is extended here
by providing omitted proofs and algorithms.

Sect.~\ref{sect:controlsoftware} compares our contribution with related work on control software synthesis from system level formal specifications.
For the sake of completeness, Sect.~\ref{sect:software} expands such a comparison to recent results on (non control) software synthesis from formal specifications, focusing on papers using techniques related to ours (constraint solving, \gls{OBDD}, supervisory control \cite{RamadgeWonham1987}).
Sect.~\ref{sect:summary} describes Tab.~\ref{tab:related-works.tab.tex}, which summarizes the novelty of our contribution with respect to automatic methods for control software synthesis.



\subsection{Control Software Synthesis from System Level Formal Specifications}
\label{sect:controlsoftware}

Control Engineering has been studying
control law design (e.g., optimal control, robust control, etc.)
for more than half a century
(e.g., see~\cite{modern-control-theory-1990}).
As explained in Sect.~\ref{separation-of-concerns} such results
cannot be directly used in our (formal) software synthesis context.
On the other hand we note that there are many control systems that
are not software based (e.g., in analog circuit design).
In such cases, of course, our approach cannot be used.

\subsubsection{Control of Linear and Switched Hybrid Systems} 

%
%
%
%

The paper closer to our is \cite{KB94} which studies the problem of
control synthesis for discrete time hybrid systems.
However, while we present an automatic method, the approach in \cite{KB94} is not automatic since it requires
the user to provide a suitable Lyapunov function (a far from trivial task even for linear hybrid systems).

\emph{Quantized Feedback Control} 
has been widely studied in control engineering
(e.g. see \cite{quantized-ctr-tac05}).
However such research addresses linear systems
(not the general case of hybrid systems) and focuses on control law design rather
than on control software synthesis (our goal).
Furthermore, all control engineering approaches
model \emph{quantization errors} as statistical \emph{noise}. 
As a result, correctness of the control law holds in a probabilistic sense.
Here instead, we model quantization errors as nondeterministic 
(\emph{malicious})
\emph{disturbances}. This guarantees system level correctness of the generated
control software (not just that of the control law) 
with respect to \emph{any} possible sequence of quantization errors.

Control software synthesis for continuous time linear systems
has been widely studied 
(e.g. see \cite{modern-control-theory-1990}).
However such research does not account for quantization.
Control software synthesis for continuous time linear systems
with quantization has been investigated in \cite{MazoTabuada11}.
This paper presents an automatic method which, taking as
input a continuous time linear system and a goal specification,
produces a control law (represented as an \gls{OBDD}) 
through {\sc Pessoa} \cite{pessoa-cav10}.
While \cite{MazoTabuada11} applies to (continuous time) linear systems, our contribution focuses on (discrete time) linear {\em hybrid} systems (\glspl{DTLHS}).
%
Furthermore, although taking into account the quantization process, \cite{MazoTabuada11} does not supply an effective method to generate control software (as we do in Sect.~\ref{sec:QFCsynthesis}). As a consequence \cite{MazoTabuada11} gives no guarantee on \gls{WCET}, an important issue since an \gls{SBCS} is a hard real-time system.

%
\cite{GirardPolaTabuada10} presents a method to find an over-approximation of {\em switched systems}, under certain stability hypotheses. A switched system is a hybrid system whose mode transitions only depend on control inputs. Such a line of research goes back to \cite{PGT2007} which presents a method to compute symbolic models for nonlinear control systems. In combination with \cite{MazoTabuada11}, such results provide a semi-automatic method for the construction of a control law for switched and nonlinear systems.
However, we note that nonlinear systems in \cite{PGT2007} are not hybrid systems, since they cannot handle discrete variables.
Moreover, while a switched system as in  \cite{GirardPolaTabuada10} is a linear hybrid system the converse is false since in a linear hybrid system mode transitions can be triggered by state changes (without any change in the input).
For example, our approach can synthesize controllers both for the buck DC-DC converter of Fig.~\ref{buck.eps} (a linear hybrid system) and for the boost DC-DC converter in \cite{GirardPolaTabuada10} (a switched system). However, the approach in \cite{GirardPolaTabuada10} cannot handle the buck DC-DC converter of Fig.~\ref{buck.eps} because of the presence of the diode which triggers state dependent mode changes.
Moreover, \cite{MazoTabuada11} combined with \cite{GirardPolaTabuada10} and \cite{PGT2007} provide semi-automatic methods since they rely on a Lyapunov function provided by the user, much in the spirit of \cite{KB94}.

\subsubsection{Control of Timed Automata and Linear Hybrid Automata}

When the plant model is a 
\gls{TA}
\cite{alur-sfm04,MalerMP91} the reachability and control law synthesis
problems have both been widely studied. Examples are in
\cite{lpw_cav97,tiga-concur05,ctr-synthesis-cav07,optctr-hscc01,PEM11} 
and citations thereof.
When the plant model is a
\gls{LHA}
\cite{alg-hs-tcs95,AHH96}  
reachability and existence of a control law are both undecidable problems
\cite{dt-ctr-rect-aut-icalp97,decidability-hybrid-automata-jcss98}.
This, of course, has not prevented devising effective 
(semi) algorithms for such problems. Examples are in
\cite{AHH96,HHT96,phaver-sttt08,ctr-lha-cdc97,Benerecetti:2011yq}.
Much in the same spirit here we give 
necessary and sufficient constructive conditions for control software existence.
Note that none of the above mentioned papers address control software synthesis
since they all assume \emph{exact} (i.e. real valued) state measures
(that is, state feedback quantization is not considered).

\newacronym{LLHA}{LLHA}{Lazy Linear Hybrid Automaton}%
Continuous-time linear hybrid systems with time delays and given precision of state measurement,
\gls{LLHA}, have been studied in \cite{AgrawalT05}.
The reachability problem is shown to be undecidable for \glspl{LLHA} in \cite{AgrawalSTY06}.
Note that such results do not directly apply to our context (Theorem~\ref{th:undecidability})
since we are addressing discrete-time systems.

\subsubsection{Control of Piecewise Affine and Nonlinear Hybrid Systems}
%
%
%
%
%
%
%
%

Finite horizon control of
\gls{PWA-DTHS}
has been studied using
a \gls{MILP} based approach.
See, for example, 
\cite{sat-opt-ctr-hscc04}.
\glspl{PWA-DTHS} form a strict subclass of \glspl{DTLHS} since \gls{PWA-DTHS} cannot handle linear
constraints consisting of discrete state variables whereas \glspl{DTLHS} can.
Such approaches cannot be directly used in our context since they
address synthesis of finite horizon controllers and do not account
for quantization.

%
Much in the spirit of \cite{KB94}, \cite{icar08} presents an explicit control synthesis algorithm for discrete time (possibly nonlinear) hybrid systems, by avoiding the needs of providing Lyapunov functions. Moreover, \cite{upmurphi-icaps09} presents control synthesis algorithms for discrete time hybrid systems cast as universal planning problems.
Such approaches cannot be directly used in our context since they do not account for quantization.

Hybrid Toolbox \cite{HybTBX} considers continuous time piecewise affine systems.
Such a tool outputs a feedback control law that is
then passed to Matlab in order to generate control software.
We note that such an approach does not account for state feedback
quantization and thus,
as explained in Sect.~\ref{separation-of-concerns}, does not
offer any formal guarantee about system level correctness
of the generated software, which is instead our focus here.

Using the engine proposed in this paper and computing suitable over-approximations of nonlinear functions, it is possible to address synthesis for nonlinear hybrid systems as done in \cite{cdc12}.

\subsubsection{Software Synthesis in a Finite Setting}

Correct-by-construction software synthesis
in a finite state setting
has been studied, for example, in
\cite{Tro97,Tro98,Tro99,Tro99a}.
An automatic method for the generation of supervisory controllers for finite state systems is presented in \cite{Tro96}.
Control software synthesis in non-deterministic finite domains is studied in \cite{strong-planning-98} (cast as a universal planning problem).
%
Such approaches cannot be directly used in our context since they
cannot handle
continuous state variables.

In Sect.~\ref{sec:qfc-syn-outline} we presented our \gls{QFC} synthesis algorithm (Alg.~\ref{qfc-syn-outline}).
Line~\ref{item:StrongCtr} of Alg.~\ref{qfc-syn-outline} calls function \fun{strongCtr} (implementing a variant of the algorithm in \cite{strong-planning-98}) in order to compute a time optimal controller for the finite state quantized system.
\cite{emsoft12} presents a method to obtain a compressed non time optimal controller for a finite state system. This is done by trading the size of the synthesized controller with time optimality while preserving closed loop performances (Remark~\ref{rem:emsoft12}).
Such a method can be implemented in function \fun{strongCtr}.
Thus, \cite{emsoft12} is not an improvement to the present paper but it is a contribution on controller synthesis for finite state systems.

\subsubsection{Switching Logic}

Optimal switching logic for hybrid systems has been also widely
investigated. For example, see \cite{TGT09,JGST10,jha-emsoft11} and citations
thereof. Such approaches, by ignoring the quantization process,
indeed focus on the control law design (see Sect.~\ref{separation-of-concerns}).
However we note that \cite{JGST10,jha-emsoft11} address dwell-time and optimality
issues which are not covered by our approach.

\subsubsection{Abstraction}

Quantization can be seen as a sort of abstraction (the reason for the name {\em control abstraction}),
which has been widely studied in a
hybrid system formal verification context
(e.g., see \cite{AHLP-ieee00,pred-abs-tecs06,Tiwari08,ST11}).
Note however that in a verification context abstractions
are designed so as to ease the verification task whereas
in our setting quantization is a design requirement since it models a hardware component
(\gls{AD} converter) which is part of the specification of the control software synthesis problem.
Indeed, in our setting, we have to design a controller \emph{notwithstanding} the nondeterminism stemming
from the quantization process. As a result, the techniques used to devise clever abstractions
in a verification setting cannot be directly used in our synthesis setting where the quantization to be used is given.

\subsection{Software Synthesis from Formal Specifications}
\label{sect:software}

Much as control software synthesis, also software synthesis has been widely studied since a long time in many contexts.
For examples, see \cite{PR89b,PR89a,SF06,GR09}.
We give a glimpse of recent results on (non control) software synthesis approaches using techniques related to ours (constraint solving, \gls{OBDD}, supervisory control).


\cite{emerson-toplas-04} shows how to mechanically synthesize fault-tolerant concurrent programs for various fault classes.
\cite{SGF2010} presents a method that synthesizes a program, if there exists one, that meets the input/output specification and uses only the given resources.
\cite{GJTV2011} addresses the problem of synthesizing loop-free programs starting from logical relations between input and output variables.
\cite{SGCF2011} proposes a synthesis technique and applies it to the inversion of imperative programs (e.g., such as insert/delete operations, compressors/decompressors).
\cite{CCHRS11} presents a method for the quantitative, performance-aware synthesis of concurrent programs.
Procedures and tools for the automated synthesis of code fragments are also proposed in \cite{KMPS12,GKP11,KMPS10}.

Such approaches build on techniques (constraint solving, \gls{OBDD}, supervisory control) related to ours, but do not address control software synthesis from system level formal specifications.


\newcommand{\Y}{$\bullet$}

\begin{table}
  \centering
  \small
  \caption{Summary of Related Work. `\Y' stands for `Yes'. An empty cell means that feature is not supported.\label{tab:related-works.tab.tex}}{%
  \begin{tabular}{!{\vrule width 1.1pt}p{0.32\textwidth}!{\vrule width 1.1pt}p{0.05cm}|p{0.05cm}!{\vrule width 1.1pt}p{0.05cm}|p{0.05cm}|p{0.05cm}|p{0.05cm}|p{0.05cm}|p{0.05cm}|p{0.05cm}|p{0.05cm}!{\vrule width 1.1pt}p{0.05cm}!{\vrule width 1.1pt}p{0.05cm}|p{0.05cm}|p{0.05cm}!{\vrule width 1.1pt}p{0.05cm}|p{0.05cm}|p{0.05cm}!{\vrule width 1.1pt}}
    \toprule
    {\em Citation} & 
    \multicolumn{2}{c!{\vrule width 1.1pt}}{\em T} & 
    \multicolumn{8}{c!{\vrule width 1.1pt}}{\em Input System} & 
     &
    \multicolumn{3}{c!{\vrule width 1.1pt}}{\em K} & 
    \multicolumn{3}{c!{\vrule width 1.1pt}}{\em Impl}\\
    \hline
    & 
    \begin{sideways}Continuous time\end{sideways} & 
    \begin{sideways}Discrete time\end{sideways} & 
    \begin{sideways}Finite state\end{sideways} & 
    \begin{sideways}Linear\end{sideways} & 
    \begin{sideways}Switched\end{sideways} & 
    \begin{sideways}\glsname{TA}, \glsname{LHA}\end{sideways} &
    \begin{sideways}Piecewise affine\end{sideways} & 
    \begin{sideways}Linear hybrid sys.\end{sideways} &
    \begin{sideways}Nonlinear\end{sideways} & 
    \begin{sideways}Nonlinear hybrid sys.\end{sideways} & 
    \begin{sideways}Quantization\end{sideways} &
    \begin{sideways}Formally verified\end{sideways} & 
    \begin{sideways}Control Software\end{sideways} & 
    \begin{sideways}Guaranteed \gls{WCET}\end{sideways} & 
    \begin{sideways}Fully automatic\end{sideways} & 
    \begin{sideways}Semi automatic\end{sideways} & 
    \begin{sideways}Tool available\end{sideways}\\
    \midrule
    \cite{qks-cav10} and this paper &   &\Y &   &\Y &\Y &\Y &\Y &\Y &   &   &\Y &\Y &\Y &\Y &\Y &   &\Y \\
    \midrule
    \cite{cdc12}                    &   &\Y &   &   &   &   &   &   &   &\Y &   &\Y &\Y &\Y &\Y &   &\Y \\
    \hline
    \cite{emsoft12}                 &   &\Y &\Y &   &   &   &   &   &   &   &   &\Y &\Y &\Y &\Y &   &\Y \\
    \hline
    \cite{optctr-hscc01}            &\Y &   &   &   &   &\Y &   &   &   &   &   &\Y &   &   &\Y &   &   \\
    \hline
    \cite{HybTBX}                   &\Y &   &   &   &   &   &\Y &   &   &   &   &   &\Y &   &\Y &   &\Y \\
    \hline
    \cite{sat-opt-ctr-hscc04}       &   &\Y &   &   &   &   &\Y &   &   &   &   &   &   &   &\Y &   &   \\
    \hline
    \cite{Benerecetti:2011yq}       &\Y &   &   &   &   &\Y &   &   &   &   &   &\Y &   &   &\Y &   &\Y \\ 
    \hline
    \cite{tiga-concur05}            &\Y &   &   &   &   &\Y &   &   &   &   &   &\Y &   &   &\Y &   &   \\
    \hline
    \cite{strong-planning-98}       &   &\Y &\Y &   &   &   &   &   &   &   &   &\Y &\Y &   &\Y &   &\Y \\
    \hline
    \cite{icar08}                   &   &\Y &   &\Y &\Y &\Y &\Y &\Y &\Y &\Y &   &   &   &   &\Y &   &\Y \\
    \hline
    \cite{upmurphi-icaps09}         &   &\Y &   &\Y &\Y &\Y &\Y &\Y &   &   &   &   &   &   &\Y &   &\Y \\
    \hline
    \cite{quantized-ctr-tac05}      &\Y &   &   &\Y &   &   &   &   &   &   &\Y &   &   &   &   &   &   \\
    \hline
    \cite{GirardPolaTabuada10}      &\Y &   &   &   &\Y &   &   &   &   &   &\Y &\Y &   &   &   &\Y &   \\
    \hline
    \cite{JGST10}                   &\Y &   &   &   &\Y &   &   &   &   &   &   &\Y &   &   &\Y &   &   \\
    \hline
    \cite{jha-emsoft11}             &\Y &   &   &   &\Y &   &   &   &   &   &   &\Y &   &   &\Y &   &   \\
    \hline
    \cite{KB94}                     &   &\Y &   &\Y &   &   &   &   &\Y &   &\Y &\Y &   &   &   &\Y &   \\
    \hline
    \cite{lpw_cav97}                &\Y &   &   &   &   &\Y &   &   &   &   &   &\Y &   &   &\Y &   &   \\
    \hline
    \cite{ctr-synthesis-cav07}      &\Y &   &   &   &   &\Y &   &   &   &   &   &\Y &   &   &\Y &   &   \\
    \hline
    \cite{MazoTabuada11}            &\Y &   &   &\Y &   &   &   &   &   &   &\Y &\Y &   &   &\Y &   &\Y \\
    \hline
    \cite{PEM11}                    &\Y &   &   &   &   &\Y &   &   &   &   &   &\Y &\Y &   &\Y &   &\Y \\
    \hline
    \cite{PGT2007}                  &\Y &   &   &   &   &   &   &   &\Y &   &\Y &\Y &   &   &   &\Y &   \\
    \hline
    \cite{TGT09}                    &\Y &   &   &   &\Y &   &   &   &   &   &   &\Y &   &   &\Y &   &   \\
    \hline
    \cite{Tro96}                    &   &\Y &\Y &   &   &   &   &   &   &   &   &\Y &\Y &   &\Y &   &   \\
    \hline
    \cite{Tro97}                    &   &\Y &\Y &   &   &   &   &   &   &   &   &\Y &\Y &   &\Y &   &   \\
    \hline
    \cite{Tro98}                    &   &\Y &\Y &   &   &   &   &   &   &   &   &\Y &\Y &   &\Y &   &   \\
    \hline
    \cite{Tro99}                    &   &\Y &\Y &   &   &   &   &   &   &   &   &\Y &\Y &   &\Y &   &   \\
    \hline
    \cite{Tro99a}                   &   &\Y &\Y &   &   &   &   &   &   &   &   &\Y &\Y &   &\Y &   &   \\
    \hline
    \cite{ctr-lha-cdc97}            &\Y &   &   &   &   &\Y &   &   &   &   &   &\Y &   &   &\Y &   &   \\ 
    \bottomrule
  \end{tabular}}
\end{table}

\subsection{Summary}
\label{sect:summary}

Tab.~\ref{tab:related-works.tab.tex} summarizes the novelty of our contribution with respect to automatic methods for control software synthesis (our focus here).
For this reason, it only considers papers addressing control software synthesis. Namely, those in Sect.~\ref{sect:controlsoftware} but the ones focusing on abstraction (since Sect.~\ref{sect:software} results do not address control software synthesis).

Tab.~\ref{tab:related-works.tab.tex} is organized as follows.
Each row refers to a citation. Each column represents a feature of a cited work. A bullet in a cell means that the citation in the cell row has the feature in the cell column. Where the feature is missing, the cell is empty.
The group of columns labeled {\em T} denotes whether the input model is expressed in {\em continuous time} or {\em discrete time}.
The group of columns labeled {\em Input System} lists the kind of input models we are interested in. Namely: {\em finite state}, {\em linear}, {\em switched}, {\em piecewise affine}, {\em \gls{TA} or \gls{LHA}}, {\em linear hybrid sys.}, {\em nonlinear}, {\em nonlinear hybrid sys}.
Note that the combination of columns {\em linear hybrid sys.} and {\em discrete time} denotes our class of \glspl{DTLHS}.
The column labeled {\em Quantization} denotes that the row supplies the quantization process.
The group of columns labeled {\em K} lists the output controller characteristics we are interested in. In particular: {\em Formally verified} denotes if the output controller is guaranteed to satisfy the given input specification; {\em Control software} indicates if the presented method outputs a control software implementation; {\em Guaranteed \gls{WCET}} denotes if the output controller has a guaranteed \gls{WCET}.
Finally, the group of columns labeled {\em Impl} considers implementation issues, namely if a method is {\em fully automatic} or {\em semi automatic}, and if there exists a {\em tool available} implementing the presented method.
Note that \cite{GirardPolaTabuada10} and \cite{PGT2007} in Tab.~\ref{tab:related-works.tab.tex} represent their combination with \cite{MazoTabuada11}.

Summing up, to the best of our knowledge, no previously published
result is available about {\em fully automatic} generation
(with a {\em tool available})
of {\em correct-by-construction control software}
with a {\em guaranteed \gls{WCET}}
from a \gls{DTLHS} model of the plant, 
\emph{system level formal specifications} and
\emph{implementation specifications} ({\em quantization}, that is number of bits in \gls{AD} conversion).

\fussy

\section{Conclusions}
\label{conclu.tex}

\sloppy

We presented an algorithm and a tool \gls{QKS} implementing it,
to support a \emph{Formal Model Based Design}
approach to control software. Our tool takes as input
a formal \gls{DTLHS} model of the plant, 
\emph{implementation specifications} 
(namely, number of bits in \gls{AD} conversion), 
and \emph{system level formal specifications}
(namely, safety and liveness properties for the closed loop system).
It returns as output a correct-by-construction
C implementation (if any) of the control software 
(namely, \texttt{Control\_Law} and \texttt{Controllable\_Region})
with a \gls{WCET} guaranteed to be linear in the number of bits of the 
quantization schema. We have shown feasibility of our proposed
approach by presenting experimental results on
using it to synthesize C controllers for the buck DC-DC converter
and the inverted pendulum.

\fussy


In order to speed-up the computation and to avoid possible 
numerical errors due to \gls{MILP} solvers \cite{NS04}, a natural
possible future research direction is to investigate
fully symbolic control software synthesis algorithms
based on efficient \emph{quantifier elimination} procedures
(e.g., see \cite{quantifier-elimination-cav10} and citations thereof).


%


\subsubsection*{Acknowledgments}
{\small 
We gratefully acknowledge partial support from
FP7 projects GA218815 (ULISSE), 317761 (SmartHG), 600773 (PAEON) and MIUR project DM24283 (TRAMP).}




\renewcommand*{\glossaryname}{ACRONYMS}

\printglossaries


\sloppy
\bibliographystyle{plain}
\bibliography{modelchecking}
\fussy


\clearpage
\appendix




\section*{Appendix}

%
%


%
%


\section{\gls{LTS} controller synthesis}
\label{sec:lts-alg}

Symbolic (\gls{OBDD} based) control software synthesis algorithms for finite state
deterministic \glspl{LTS} have been studied in 
\cite{Tro98} and citations thereof.
In such a
context of course strong and weak solutions are the same. Symbolic
(\gls{OBDD} based) control synthesis algorithms for finite state
nondeterministic \glspl{LTS} have been studied in \cite{strong-planning-98}
in a universal planning setting. In such a context strong and weak
solutions in general differ.

\sloppy

To compute strong solutions, we implemented a variant of the algorithm
in \cite{strong-planning-98} in function \fun{strongCtr}. In our
variant, a strong controller for the given \gls{LTS} control problem is
always returned, even if it is not possible to entirely control the
given initial region (see Sect.~\ref{sec:qfc-syn-outline}).
More precisely, it returns the strong mgo (see Def.~\ref{def:sol}),
i.e. the \emph{unique} strong
solution $K$ to a control problem (${\cal S}$, $I$, $G$) that,
disallowing as few actions as possible, drives as many states as
possible to a state in $G$ along a shortest path. 
For the sake of
completeness, we show the resulting algorithm in Alg.~\ref{strngctr.alg}.

\fussy

Analogously, function \fun{existsWeakCtr} exploits the algorithm in~\cite{Tro98}
to verify the existence of weak solutions. Function \fun{existsWeakCtr} 
is shown in
Alg.~\ref{wkctr.alg}.


\begin{algorithm}
\caption{Building a strong mgo for an \gls{LTS} control problem}
\label{strngctr.alg}
\begin{algorithmic}[1]
\REQUIRE \gls{LTS} control problem $({\cal S}, I, G)$, with \gls{LTS} ${\cal S} = (S, A, T)$.
\ENSURE \fun{strongCtr}(${\cal S}, I, G$)
\STATE $K(s, a) \gets 0$, $D(s) \gets G(s)$, $\tilde{D}(s) \gets 0$
\WHILE {$D(s) \neq \tilde{D}(s)$} 
  \STATE $F(s, a) \gets \exists s': T(s, a, s') \land \forall s'\; [T(s, a, s') \Rightarrow D(s')]$
  \STATE $K(s, a) \gets K(s, a) \lor (F(s, a) \land \not\exists a: K(s, a))$
  \STATE $\tilde{D}(s) \gets D(s)$, $D(s) \gets D(s) \lor \exists a: K(s, a)$
\ENDWHILE
\STATE {\bf return} $\langle \forall s\; [I(s) \Rightarrow \exists a: K(s, a)], \exists a: K(s, a), K(s, a)\rangle$
\end{algorithmic}
\end{algorithm}  


Correctness of function \fun{strongCtr} in Alg.~\ref{strngctr.alg} is proved in
Prop.~\ref{strongCtr.corr.prop}.

\begin{proposition}\label{strongCtr.corr.prop}
Let ${\cal S} = (S, A, T)$ be an \gls{LTS} and ${\cal P} = ({\cal S}, I, G)$ be an \gls{LTS} control
problem. Then, \fun{strongCtr}(${\cal S}, I, G$) returns $\langle b, D,
K\rangle$ s.t. $K$ is the strong mgo for $({\cal S}, \varnothing, G)$, $D = {\rm
Dom}(K)$ and $b$ is {\sc True} iff $K$ is the strong mgo for ${\cal P}$.
\end{proposition}

\begin{proof}
\sloppy
We observe that during a generic iteration $i$ the set of states $\{ s~|~\exists
a : F(s,a)\}$ is exactly the set of states $F_i$   and $\{ s~|~\tilde{D}(s)\}$ is
exactly the set of states $D_i$ considered in the proof of Prop.~\ref{prop:mgo}
in Sect.~\ref{proof.mgo}.  As a consequence, the thesis holds by the proof of
Prop.~\ref{prop:mgo}.
%
\fussy
\end{proof}

\begin{algorithm}
\caption{Existence of weak solutions to \gls{LTS} control problem}
\label{wkctr.alg}
\begin{algorithmic}[1]
\REQUIRE \gls{LTS} control problem $({\cal S}, I, G)$, with \gls{LTS} ${\cal S} = (S, A, T)$.
\ENSURE \fun{existsWeakCtr}(${\cal S}, I, G$)
\STATE $K(s, a) \gets 0$, $D(s) \gets G(s)$, $\tilde{D}(s) \gets 0$
\WHILE {$D(s) \neq \tilde{D}(s)$} 
  \STATE $F(s, a) \gets \exists s': [T(s, a, s') \land D(s')]$
  \STATE $K(s, a) \gets K(s, a) \lor (F(s, a) \land \not\exists a: K(s, a))$
  \IF{$\forall s\; [I(s) \Rightarrow \exists a: K(s, a)]$}
  \RETURN {\sc True}
  \ENDIF
  \STATE $\tilde{D}(s) \gets D(s)$, $D(s) \gets D(s) \lor \exists a: K(s, a)$
\ENDWHILE
\STATE {\bf return} {\sc False}
\end{algorithmic}
\end{algorithm}  

Correctness of function \fun{existsWeakCtr} in Alg.~\ref{wkctr.alg} may be
proved analogously to Prop.~\ref{strongCtr.corr.prop}.

\begin{proposition}\label{existsWeakCtr.corr.prop}
\sloppy
Let ${\cal S} = (S, A, T)$ be an \gls{LTS} and ${\cal P} = ({\cal S}, I, G)$ be an \gls{LTS}
control problem. Then, \fun{existsWeakCtr}(${\cal S}, I, G$) returns {\sc True}
iff there exists a weak mgo for ${\cal P}$.
\fussy
\end{proposition}

\begin{corollary}\label{existsWeakCtr.corr.cor}
Let ${\cal S} = (S, A, T)$ be an \gls{LTS} and ${\cal P} = ({\cal S}, I, G)$ be an \gls{LTS}
control problem. Then, \fun{existsWeakCtr}(${\cal S}, I, G$) returns {\sc True}
iff there exists a weak solution for ${\cal P}$.
\end{corollary}

\section{Details about the Experiments}
\label{expres-details.tex}

In this section we give (Tab.~\ref{expres-table.milp.det.tex}) all details about
\gls{MILP} problems arising in our experiments about the buck DC/DC converter 
of Sects.~\ref{expres.tex.setting}--\ref{sec:controller-performances}. Namely, in
Tab.~\ref{expres-table.milp.det.tex} MILP$i$ has the same meaning as in
Sect.~\ref{sec:milp-analysis}, i.e. MILP1 refers to the \gls{MILP} problems described
in Sect.~\ref{dths-ctr-abs.tex}, i.e. those computing the quantization for $I$
and $G$,  MILP2 refers to \gls{MILP} problems in function \fun{SelfLoop} (see
Alg.~\ref{alg:selfLoop}),  MILP3 refers to the \gls{MILP} problems  used in function
\fun{overImg} (line~\ref{overimg.alg.step} of Alg.~\ref{ctr-abs.alg}),  MILP4
refers to \gls{MILP} problems used to check actions admissibility
(line~\ref{check_s_s_prime.alg.step} of Alg.~\ref{ctr-abs.alg}), and  MILP5
refers to \gls{MILP} problems used to check transitions witnesses
(line~\ref{check_s_u.alg.step} of Alg.~\ref{ctr-abs.alg}). In
Tab.~\ref{expres-table.milp.det.tex} columns $b$, {\em Num}, {\em Avg} and {\em
Tot} are the same as columns $b$, {\em Num}, {\em Avg} and {\em Time} of
Tabs.~\ref{expres-table.tras.milps.tex:1} and~\ref{expres-table.tras.milps.tex:2} thus $b$ shows the number of \gls{AD} bits,
{\em Num} is the number of times that the \gls{MILP} problem of the given type is
called, {\em Tot} is the total CPU time needed to solve all the {\em Num}
instances of \gls{MILP} problem of the given type, and {\em Avg} is the ratio between
{\em Tot} and {\em Num}. In Tab.~\ref{expres-table.milp.det.tex} we also show in
columns {\em Min} and {\em Max} the average, minimum and maximum time to solve
one \gls{MILP} problem of the given type. The standard deviation for such statistics
is given in column {\em DevStd}.

\begin{center}
\begin{table}
\centering
\caption{Complete statistics for Tabs.~\ref{expres-table.tras.milps.tex:1} and~\ref{expres-table.tras.milps.tex:2} of Sect.~\ref{expres.tex}.\label{expres-table.milp.det.tex}}{%
\begin{tabular}{|c||*6{c|}}\hline
\multicolumn{7}{|c|}{MILP1}\\
\hline
{\scriptsize $b$} & {\scriptsize Num} & {\scriptsize Tot} & {\scriptsize Avg} & {\scriptsize Min} & {\scriptsize Max} & {\scriptsize DevStd}\\
\hline
8 & 6.55e+04 & 4.61e+00 & 7.03e-05 & 0.00e+00 & 1.00e-02 & 8.35e-04\\
\hline
9 & 2.62e+05 & 1.84e+01 & 7.02e-05 & 0.00e+00 & 1.00e-02 & 8.35e-04\\
\hline
10 & 1.05e+06 & 2.79e+02 & 2.66e-04 & 0.00e+00 & 1.00e-02 & 1.61e-03\\
\hline
11 & 4.19e+06 & 9.65e+02 & 2.30e-04 & 0.00e+00 & 1.00e-02 & 1.50e-03\\
\hline
\multicolumn{7}{|c|}{MILP2}\\
\hline
{\scriptsize $b$} & {\scriptsize Num} & {\scriptsize Tot} & {\scriptsize Avg} & {\scriptsize Min} & {\scriptsize Max} & {\scriptsize DevStd}\\
\hline
8 & 3.99e+05 & 3.25e+02 & 1.52e-03 & 0.00e+00 & 1.00e-02 & 4.39e-03\\
\hline
9 & 1.59e+06 & 1.12e+03 & 1.41e-03 & 0.00e+00 & 1.00e-02 & 4.14e-03\\
\hline
10 & 6.36e+06 & 1.35e+04 & 3.78e-03 & 0.00e+00 & 1.00e-02 & 6.43e-03\\
\hline
11 & 2.54e+07 & 4.56e+04 & 3.26e-03 & 0.00e+00 & 1.00e-02 & 6.10e-03\\
\hline
\multicolumn{7}{|c|}{MILP3}\\
\hline
{\scriptsize $b$} & {\scriptsize Num} & {\scriptsize Tot} & {\scriptsize Avg} & {\scriptsize Min} & {\scriptsize Max} & {\scriptsize DevStd}\\
\hline
8 & 2.31e+05 & 2.10e+02 & 9.10e-04 & 0.00e+00 & 1.00e-02 & 3.20e-03\\
\hline
9 & 9.21e+05 & 8.44e+02 & 9.16e-04 & 0.00e+00 & 1.00e-02 & 3.18e-03\\
\hline
10 & 3.68e+06 & 1.11e+04 & 3.00e-03 & 0.00e+00 & 2.00e-02 & 4.26e-03\\
\hline
11 & 1.47e+07 & 3.76e+04 & 2.55e-03 & 0.00e+00 & 2.00e-02 & 4.02e-03\\
\hline
\multicolumn{7}{|c|}{MILP4}\\
\hline
{\scriptsize $b$} & {\scriptsize Num} & {\scriptsize Tot} & {\scriptsize Avg} & {\scriptsize Min} & {\scriptsize Max} & {\scriptsize DevStd}\\
\hline
8 & 7.80e+05 & 7.71e+02 & 9.89e-04 & 0.00e+00 & 1.00e-02 & 2.98e-03\\
\hline
9 & 4.42e+06 & 4.49e+03 & 1.02e-03 & 0.00e+00 & 1.00e-02 & 3.02e-03\\
\hline
10 & 3.01e+07 & 7.75e+04 & 2.58e-03 & 0.00e+00 & 2.00e-02 & 4.37e-03\\
\hline
11 & 2.61e+08 & 5.66e+05 & 2.17e-03 & 0.00e+00 & 2.00e-02 & 4.13e-03\\
\hline
\multicolumn{7}{|c|}{MILP5}\\
\hline
{\scriptsize $b$} & {\scriptsize Num} & {\scriptsize Tot} & {\scriptsize Avg} & {\scriptsize Min} & {\scriptsize Max} & {\scriptsize DevStd}\\
\hline
8 & 4.27e+05 & 1.20e+02 & 2.80e-04 & 0.00e+00 & 1.00e-02 & 1.65e-03\\
\hline
9 & 1.71e+06 & 4.87e+02 & 2.85e-04 & 0.00e+00 & 1.00e-02 & 1.66e-03\\
\hline
10 & 6.84e+06 & 1.25e+04 & 1.83e-03 & 0.00e+00 & 2.00e-02 & 3.87e-03\\
\hline
11 & 2.74e+07 & 4.25e+04 & 1.55e-03 & 0.00e+00 & 2.00e-02 & 3.62e-03\\
\hline
\end{tabular}}
\end{table}
\end{center}

\end{document}